\documentclass{article}
\usepackage{authblk}
\usepackage[utf8]{inputenc}
\usepackage{amsmath,amssymb,amsfonts,amsthm}
\usepackage{physics}
\usepackage{dsfont}
\usepackage{graphicx}
\usepackage{xcolor}
\usepackage{cancel}
\usepackage{cancel}
\usepackage[numbers,sort&compress,square]{natbib}
\usepackage{color}
\usepackage{bbm}
\usepackage[linktocpage]{hyperref}
\usepackage{appendix}
\hypersetup{colorlinks=true,citecolor=blue,linkcolor=blue, urlcolor=blue, breaklinks=true}
\usepackage[eulergreek]{sansmath}
\usepackage{mathtools}
\usepackage{scalerel}

\usepackage{bm} %allows for bold math type
\usepackage{subfigure} %allows for 2x2 figures
\usepackage{environ}
\usepackage{url}
\usepackage{hyperref}
\usepackage[margin=1in]{geometry}
\usepackage{adjustbox}
\usepackage{environ}

\usepackage{tikz}
\usepackage{makecell}
\usepackage{enumitem}
\usepackage{booktabs}
 \usepackage{multirow}

\newtheorem{theorem}{Theorem}[section]

\newtheorem{definition}{Definition}

\newcommand{\id}{{\operatorname{id}}}
\newcommand{\inv}{{\operatorname{inv}}}

\newcommand{\ZZ}{{\mathbb Z}}
\newcommand{\RR}{{\mathbb R}}

\newcommand{\Z}{{\mathbb Z}}

\newcommand{\be}{\boldsymbol e}

\newcommand{\bv}{\boldsymbol{v}}

\newcommand{\bface}{\boldsymbol{f}}
\newcommand{\bt}{\boldsymbol{t}}
\newcommand{\bc}{\boldsymbol{c}}

\NewEnviron{eqs}{%
\begin{equation}\begin{split}
    \BODY
\end{split}\end{equation}}

\graphicspath{{./Figures/}}

\newcommand{\change}{\color{black}}

\newcommand{\Sq}{{\mathrm{Sq}}}

\newcommand{\supp}{\operatorname{supp}}
\newcommand{\pd}{\operatorname{PD}}

\begin{document}

\title{Pauli stabilizer formalism for topological quantum field theories and generalized statistics}

\author[1,*]{Yitao Feng}
\author[2,*]{Hanyu Xue}
\author[3]{Ryohei Kobayashi}
\author[4]{Po-Shen Hsin}
\author[1,$\dagger$]{Yu-An Chen}

\affil[1]{International Center for Quantum Materials, School of Physics, Peking University, Beijing 100871, China}
\affil[2]{Department of Physics, Massachusetts Institute of Technology, Cambridge, Massachusetts 02139, USA}
\affil[3]{School of Natural Sciences, Institute for Advanced Study, Princeton, NJ 08540, USA}
\affil[4]{Department of Mathematics, King’s College London, Strand, London WC2R 2LS, UK}

\renewcommand*{\thefootnote}{*}
\footnotetext[1]{These authors contributed equally to this work.}
\renewcommand*{\thefootnote}{$\dagger$}
\footnotetext[1]{Contact author: yuanchen@pku.edu.cn}
\renewcommand{\thefootnote}{\arabic{footnote}}

\date{\today}

\maketitle

\begin{abstract}

Topological quantum field theory (TQFT) provides a unifying framework for describing topological phases of matter and for constructing quantum error-correcting codes, playing a central role across high-energy physics, condensed matter, and quantum information.
A central challenge is to formulate topological order on lattices and to extract the properties of topological excitations from microscopic Hamiltonians.
In this work, we construct new classes of lattice gauge theories as Pauli stabilizer models, realizing a wide range of TQFTs in general dimensions. We develop a lattice description of extended excitations and systematically determine their generalized statistics.
Our main example is the $(4{+}1)$D \emph{fermionic-loop toric code}, obtained by condensing the $e^2 m^2$-loop in the $(4{+}1)$D $\mathbb{Z}_4$ toric code. We show that the loop excitation exhibits fermionic loop statistics: the 24-step loop-flipping process yields a phase of $-1$.
Our Pauli stabilizer models realize all twisted 2-form gauge theories in $(4{+}1)$D, the higher-form Dijkgraaf-Witten TQFT classified by $H^{5}(B^{2}G, U(1))$.
Beyond $(4{+}1)$D, the fermionic-loop toric codes form a family of $\mathbb{Z}_2$ topological orders in arbitrary dimensions, realized as explicit Pauli stabilizer codes using $\mathbb{Z}_4$ qudits.
Finally, we develop a Pauli-based framework that defines generalized statistics for extended excitations in any dimension, yielding computable lattice unitary processes to detect nontrivial statistics.
For example, we propose anyonic membrane statistics in $(6{+}1)$D, as well as fermionic membrane and volume statistics in arbitrary dimensions. We construct new families of $\mathbb{Z}_2$ topological orders: the \emph{fermionic-membrane toric code} and the \emph{fermionic-volume toric code}.
In addition, we demonstrate that $p$-dimensional excitations in $2p+2$ spatial dimensions can support anyonic $p$-brane statistics for only even $p$.

\end{abstract}

\tableofcontents

\bigskip

\section{Introduction}

Topological quantum field theory (TQFT) is a cornerstone of high-energy physics and has critical applications in condensed matter physics and quantum information. It describes the infrared fixed points of topological phases of matter and has been widely used as a design principle for quantum error-correcting codes~\cite{Dijkgraaf:1989pz, Atiyah1988TQFT, Kitaev2002Topologicalquantummemory, freedman2003topological, Nayak2008NonAbelian, Chen2013cohomology}.
While TQFTs in two spatial dimensions are well studied and often admit simple lattice realizations~\cite{Kitaev2003Fault, Levin2005String, Yidun2013quantumdouble, Kong2014Anyon, Levin2012Braiding, Ellison2022Pauli}, explicit constructions in higher dimensions are much less developed~\cite{Ye2023Fusion, Ye2025Diagrammatics}. In particular, lattice Hamiltonians that realize twisted higher-form gauge theories are frequently intricate, and the structure and statistics of their extended excitations—particles, loops, membranes, and higher-dimensional objects—can be difficult to analyze microscopically.

Topological phases realized by stabilizer Hamiltonians provide a natural setting for quantum error correction~\cite{Gottesman1997StabilizerThesis, Terhal2015QECMemories}. Pauli stabilizer codes combine the robustness of topological order with the algebraic simplicity of Pauli operators, and they form the foundation of many fault-tolerant quantum computing schemes~\cite{kitaev2009topological}. Higher-dimensional stabilizer codes are appealing: in $(4{+}1)$D and above, topological orders can exhibit energy barriers between distinct ground states that grow with system size, making them promising candidates for self-correcting quantum memories~\cite{Bombin2013Selfcorrecting, Hsin2025NonAbelian}. Nevertheless, the space of higher-dimensional topological orders that admit Pauli stabilizer realizations remains largely unexplored. A central challenge is to construct explicit stabilizers that realize general TQFTs and to analyze their excitations directly on the lattice.

While many twisted higher-form gauge theories are often described as ``exactly solvable,'' this statement typically refers to the existence of a closed-form characterization of the ground-state wave function (or partition function), rather than a microscopic Hamiltonian with fully tractable excitations. On the lattice, one can often write frustration-free local constraints whose common $+1$ eigenspace reproduces the ground-state sector. However, these local terms need not commute (unless one restricts to a special subspace, e.g., a zero-flux sector). In such cases, beyond the ground state space, the structure of extended excitations becomes much harder to access microscopically. This motivates the search for Pauli stabilizer realizations, whose ground states match the standard gauge-theory wave functions, while keeping excitation operators, fusion rules, and generalized statistics explicitly computable on the lattice.

Beyond providing analytic control, stabilizer realizations are also attractive from the perspective of code design. The $\ZZ_N$ toric code already supports fault-tolerant logical operations via its charge and flux excitations, but more structured stabilizer realizations can offer additional capabilities. Even in two dimensions, many stabilizer codes, such as the color code~\cite{bombin2006topological, kubica2015unfolding} and bivariate bicycle codes~\cite{Kovalev2013QuantumKronecker, Pryadko2022DistanceGB, Bravyi2024HighThreshold, wang2024coprime, Wang2024Bivariate, tiew2024low, wolanski2024ambiguity, shaw2024lowering,  voss2024multivariate, berthusen2025toward, Eberhardt2025Logical, liang2025generalized, liang2025selfdual}, are equivalent to multiple copies of toric codes up to finite-depth circuits~\cite{bombin_Stabilizer_14, haah_module_13, haah2016algebraic, haah_classification_21, Chen2023Equivalence, ruba2024homological}, yet their distinct stabilizer structures enable extra features, including transversal Clifford gates and higher encoding rates. These examples suggest that realizing higher-form TQFTs as Pauli stabilizer models could expand the design space for quantum codes, enabling richer excitation structures that may yield new advantages for both computation and robust quantum memory.

In this work, we develop a Pauli stabilizer formalism that systematically realizes a broad class of higher-dimensional TQFTs as Pauli stabilizer codes. Ref.~\cite{Ellison2022Pauli} constructed Pauli stabilizer models for all Abelian $(2{+}1)$D topological orders with gapped boundaries by condensing anyons in toric codes with qudits. We generalize this construction to arbitrary dimensions and obtain many new classes of higher-dimensional TQFTs, all of which are realized by Pauli operators acting on qudits.

Our main example is a family of $\ZZ_2$ topological orders, which we call the \emph{fermionic-loop toric codes}, defined in arbitrary $d \geq 4$ spatial dimensions. These phases support fermionic flux-loop excitations together with bosonic $(d{-}3)$-brane charge excitations, with a mutual braiding phase of $-1$ between charge and flux. The notion of fermionic loop statistics has been formulated recently~\cite{FHH21, CH21, kobayashi2024generalized, xue2025statisticsabeliantopologicalexcitations}: a loop-flipping process, implemented by the 24-step sequence, produces a phase of $-1$. The $(4{+}1)$D member of this family admits a particularly simple realization, obtained by condensing $e^2m^2$ loops in the standard $(4{+}1)$D $\ZZ_4$ toric code, as shown in Sec.~\ref{sec: (4+1)D fermionic-loop toric code}.

More generally, in Sec.~\ref{sec: Dijkgraaf-Witten 2-form gauge theories}, we construct Pauli stabilizer models for all $(4{+}1)$D Dijkgraaf-Witten $2$-form $G$ gauge theories classified by $H^{5}(B^{2}G,U(1))$, via appropriate loop condensations starting from standard toric codes. Beyond $(4{+}1)$D, the fermionic-loop toric codes extend naturally to higher dimensions using algebraic-topology notation such as higher cup products, as developed in Sec.~\ref{sec: Fermionic-loop toric codes in arbitrary dimensions}. In all cases, our stabilizer constructions make the relevant excitations fully accessible microscopically, enabling explicit analysis of their hopping operators, fusion rules, and consistent condensations.

Another key contribution of this work is a unified framework for computing the statistics of extended excitations within Pauli stabilizer models. Extending the Berry-phase approach of Ref.~\cite{kobayashi2024generalized}, we construct sequences of local unitaries whose many-body Berry phases define quantized invariants characterizing the statistics of particles, loops, membranes, and volume excitations. This provides a general operational method for directly detecting generalized statistics on the lattice, revealing statistical processes beyond conventional braiding. This framework is developed in Sec.~\ref{sec: Statistics in Pauli Stabilizer Models}.

Finally, our formalism naturally leads to new families of exactly solvable higher-dimensional lattice gauge theories realized as stabilizer codes. As concrete examples, we construct the \emph{semionic-membrane toric code} in $(6{+}1)$D by condensing $e^2m^2$ membranes in the standard $(6{+}1)$D $\mathbb{Z}_4$ toric code in Sec.~\ref{sec: Semionic-membrane toric code in (6+1)D}, and we discover a novel $\mathbb{Z}_4$ membrane statistics in 6 spatial dimensions whose generator we refer to as the semionic membrane. However, just as semions exist only in $(2{+}1)$D, this semionic membrane statistics is not stable in higher dimensions: only its $\mathbb{Z}_2$ subgroup survives. Therefore, we construct the \emph{fermionic-membrane toric code} families in $d \geq 7$ spatial dimensions in Sec.~\ref{sec: Fermionic-membrane toric code in arbitrary dimensions}.
The same procedure extends higher dimensions: condensing $e^2m^2$ volume excitations in the standard $(8{+}1)$D $\mathbb{Z}_4$ toric code yields the \emph{fermionic-volume toric code}, as described in Sec.~\ref{sec: Fermionic-volume toric code in arbitrary dimensions}, which forms a family in higher dimensions. Based on these constructions, we observe a periodicity of four in dimension: the $(2{+}1)$D semion is analogous to the $(6{+}1)$D semionic membrane, while the $(4{+}1)$D fermionic loop is analogous to the $(8{+}1)$D fermionic volume.

The main results of this paper are partially summarized in Table~\ref{tab:flux-anomaly-summary} and~\ref{tab: Pauli statistics}. In summary, we provide Pauli stabilizer realizations of all $(4{+}1)$D Dijkgraaf-Witten 2-form gauge theories, together with explicit higher-dimensional models such as the fermionic-loop, semionic-membrane, fermionic-membrane, and fermionic-volume toric codes.
We find that $p$-dimensional excitations in $2p{+}2$ spatial dimensions can exhibit anyonic $p$-brane statistics only when $p$ is even. For odd $p$, such anyonic statistics is absent. For example, the $em$ loop in the $(4{+}1)$D loop-only $\ZZ_N$ toric code is always bosonic~\cite{Chen2023Loops4d}.
Our generalized-statistics framework yields computable lattice Berry phases that detect novel statistics of extended excitations within the Pauli stabilizer setting. These results extend the class of TQFTs and topological orders accessible within the stabilizer formalism, which will be useful for the design of higher-dimensional quantum memories and fault-tolerant logical operations in the future.

\begin{table}[thb]
\centering
\renewcommand{\arraystretch}{1.3}
\begin{tabular}{|c|c|c|c|c|}
\hline
Topological order
&
Action $\frac{S}{2\pi i}$
&
\begin{tabular}{@{}c@{}}
Anomaly of \\
symmetry $(-1)^{\oint a}$
\end{tabular}
&
\begin{tabular}{@{}c@{}}
Statistics  \\
of flux~(Sec.~\ref{sec: Statistics in Pauli Stabilizer Models})
\end{tabular}
&
\begin{tabular}{@{}c@{}}
Stable in \\
higher dimensions
\end{tabular}
\\
\hline
%%%%%%%%%%%%%%%%%%%%%%%%%%%%%%%%%%%%%%%%%%%%%%%%%%%%%%%%%%%%
\begin{tabular}{@{}c@{}}
$(2{+}1)$D \\
double semion~\cite{Levin2012Braiding, Ellison2022Pauli}
\end{tabular}
&
\begin{tabular}{@{}l@{}}
$\frac{1}{2} a_1 \cup \delta b_1 $ \\
$+ \frac{1}{4} b_1 \cup \delta b_1$
\end{tabular}
&
$\tfrac14\bigl(B_2 \cup B_2 + B_2 \cup_1 \delta B_2\bigr)$
&
semion
&
No 
\\[1em]
\hline
%%%%%%%%%%%%%%%%%%%%%%%%%%%%%%%%%%%%%%%%%%%%%%%%%%%%%%%%%%%%
\begin{tabular}{@{}c@{}}
$(3{+}1)$D \\
fermionic-particle \\
toric code~\cite{Levin2006Quantumether, Chen:2018nog}
\end{tabular}
&
\begin{tabular}{@{}l@{}}
$\tfrac12 a_1 \cup \delta b_2$ \\
$+ \tfrac12 b_2 \cup b_2$
\end{tabular}
&
\begin{tabular}{@{}l@{}}
$\tfrac12 B_3 \cup_1 B_3$ \\
$= \tfrac12\,\operatorname{Sq}^2  B_3$ \\
$= \tfrac12\, w_2 \cup B_3$
\end{tabular}
&
fermion
&
\begin{tabular}{@{}c@{}}
    Yes \\
    $\tfrac12 a_1 \cup \delta b_n$ \\
    $+ \tfrac12\operatorname{Sq}^2 b_n$
\end{tabular}
\\[2em]
\hline
%%%%%%%%%%%%%%%%%%%%%%%%%%%%%%%%%%%%%%%%%%%%%%%%%%%%%%%%%%%%
\begin{tabular}{@{}c@{}}
$(4{+}1)$D \\
fermionic-loop\\
toric code~(Sec.~\ref{sec: (4+1)D fermionic-loop toric code})
\end{tabular}
&
\begin{tabular}{@{}l@{}}
$\tfrac12 a_2 \cup \delta b_2$ \\
$+ \tfrac14 b_2 \cup \delta b_2$
\end{tabular}
&
\begin{tabular}{@{}l@{}}
$\tfrac14\bigl(B_3 \cup B_3 + B_3 \cup_1 \delta B_3\bigr)$ \\
$= \tfrac12\,\operatorname{Sq}^2 \operatorname{Sq}^1 B_3$ \\
$= \tfrac12\, w_3 \cup B_3$
\end{tabular}
&
\begin{tabular}{@{}c@{}}
fermionic \\
loop
\end{tabular}
&
\begin{tabular}{@{}c@{}}
    Yes~(Sec.~\ref{sec: Fermionic-loop toric codes in arbitrary dimensions}) \\
    $\tfrac12 a_2 \cup \delta b_n $ \\
    $+ \tfrac12 \operatorname{Sq}^2 \operatorname{Sq}^1 b_n$
\end{tabular}
\\[2em]
\hline
%%%%%%%%%%%%%%%%%%%%%%%%%%%%%%%%%%%%%%%%%%%%%%%%%%%%%%%%%%%%
\begin{tabular}{@{}c@{}}
$(6{+}1)$D \\
semionic-membrane\\
toric code~(Sec.~\ref{sec: Semionic-membrane toric code in (6+1)D})
\end{tabular}
&
\begin{tabular}{@{}l@{}}
$\tfrac12 a_3 \cup \delta b_3$ \\
$+ \tfrac14 b_3 \cup \delta b_3$
\end{tabular}
&
$\tfrac14\bigl(B_4 \cup B_4 + B_4 \cup_1 \delta B_4\bigr)$
&
\begin{tabular}{@{}c@{}}
semionic \\
membrane
\end{tabular}
&
No
\\
\hline
%%%%%%%%%%%%%%%%%%%%%%%%%%%%%%%%%%%%%%%%%%%%%%%%%%%%%%%%%%%%
\begin{tabular}{@{}c@{}}
$(7+1)$D \\
fermionic-membrane \\
toric code~(Sec.~\ref{sec: Fermionic-membrane toric code in arbitrary dimensions})
\end{tabular}
&
\begin{tabular}{@{}l@{}}
$\tfrac12 a_3 \cup \delta b_4 $ \\
$+ \tfrac12 b_4 \cup b_4$
\end{tabular}
&
\begin{tabular}{@{}l@{}}
$\tfrac12 B_5 \cup_1 B_5$ \\
$= \tfrac12\,\operatorname{Sq}^4  B_5$ \\
$= \tfrac12\, (w_4+w_2^2) \cup B_5$
\end{tabular}
&
\begin{tabular}{@{}c@{}}
fermionic \\
membrane
\end{tabular}
&
\begin{tabular}{@{}c@{}}
    Yes \\
    $\tfrac12 a_3 \cup \delta b_n$\\
    $+ \tfrac12\operatorname{Sq}^4 b_n$
\end{tabular}
\\[2em]
\hline
%%%%%%%%%%%%%%%%%%%%%%%%%%%%%%%%%%%%%%%%%%%%%%%%%%%%%%%%%%%%
\begin{tabular}{@{}c@{}}
$(8{+}1)$D \\
fermionic-volume\\
toric code~(Sec.~\ref{sec: Fermionic-volume toric code in arbitrary dimensions})
\end{tabular}
&
\begin{tabular}{@{}l@{}}
$\tfrac12 a_4 \cup \delta b_4$ \\
$+ \tfrac14 b_4 \cup \delta b_4$
\end{tabular}
&
\begin{tabular}{@{}l@{}}
$\tfrac14\bigl(B_5 \cup B_5 + B_5 \cup_1 \delta B_5\bigr)$ \\ 
$= \tfrac12\,\operatorname{Sq}^4 \operatorname{Sq}^1 B_5 $ \\
$= \tfrac12\, w_5 \cup B_5$
\end{tabular}
&
\begin{tabular}{@{}c@{}}
fermionic \\
volume
\end{tabular}
&
\begin{tabular}{@{}c@{}}
    Yes \\
    $\tfrac12 a_4 \cup \delta b_n$ \\
    $+ \tfrac12\operatorname{Sq}^4 \operatorname{Sq}^1 b_n$
\end{tabular}
\\[2em]
\hline
%%%%%%%%%%%%%%%%%%%%%%%%%%%%%%%%%%%%%%%%%%%%%%%%%%%%%%%%%%%%
\end{tabular}
\caption{Summary of $\ZZ_2$ topological orders presented in this work, including their topological actions, anomalies of fluxes, flux statistics, and whether the corresponding cohomology operation is stable (and hence extends to higher dimensions).
Here $a_n$, $b_n$, and $B_n$ denote $\ZZ_2$-valued $n$-cochains.
The first row reproduces the $(2{+}1)$D double-semion order.
In $(4{+}1)$D, the fermionic-loop toric code is characterized by the Stiefel-Whitney class $w_3$ (equivalently, $\operatorname{Sq}^2\operatorname{Sq}^1$), which is stable; accordingly, this construction extends to a family of $\ZZ_2$ topological orders in $(d{+}1)$D for all $d\ge4$.
By contrast, the $(6{+}1)$D semionic-membrane toric code is analogous to the double-semion theory and does not extend to higher dimensions.
For $d\ge7$, the surviving $\ZZ_2$ subgroup of the $\ZZ_4$ semionic membrane statistics yields a stable topological phase, the fermionic-membrane toric code.
Overall, $(2n{+}1)$D topological orders twisted by $\tfrac{1}{4}\, b_n\cup \delta b_n$ are stable only when $n$ is even. When $n$ is odd, the twist $\tfrac{1}{4}\, b_n\cup \delta b_n$ instead produces semionic statistics.
}

\label{tab:flux-anomaly-summary}
\end{table}

\begin{table}[t]
\centering
\renewcommand{\arraystretch}{1.8}
\begin{tabular}{@{}lll@{}}
\toprule
Excitation & Spatial dimension $d$ & Statistics detector \\
\midrule
\multirow{2}{*}{$\ZZ_N$ particle}
& $d=1$~\cite{CarolynZhangSPTEntangler, kobayashi2024generalized}
& $\bigl[U_{01},\,U_{02}\bigr]^N$ \\
& $d\ge 2$~\cite{Levin2003Fermions}
& $\bigl[U_{01},\,U_{02}\bigr]\,
   \bigl[U_{02},\,U_{03}\bigr]\,
   \bigl[U_{03},\,U_{01}\bigr]$ \\
\midrule
$\ZZ_{2N}$ loop
& $d\ge 3$~\cite{FHH21, kobayashi2024generalized}
& $\bigl[U_{012},\,U_{034}\bigr]^{2N^2}\,
   \bigl[U_{014},\,U_{023}\bigr]^{2N^2}\,
   \bigl[U_{013},\,U_{024}\bigr]^{2N^2}$ \\
\midrule
\multirow{2}{*}{$\ZZ_N$ membrane}
& $d=5$~(Theorem~\ref{thm: even p dim 2p+1 statistics})
& $\displaystyle {\prod}'_{\sigma_3 \sqcup \tau_3 =\{1,2,3,4,5,6\}} \bigl[U_{0\sigma_3},\,U_{0\tau_3}\bigr]^{o_\pm \cdot N}, \quad o_\pm = (-1)^{\sum \sigma_3}$\\
& $d\ge 6$~(Theorem~\ref{thm: even p dim 2p+2 statistics})
& $\displaystyle {\prod}'_{k\sqcup\sigma_3\sqcup\tau_3=\{1,2,3,4,5,6,7\}} \bigl[U_{0\sigma_3},\,U_{0\tau_3}\bigr]^{o'_\pm}, \quad o'_\pm = (-1)^{k+ \#(\sigma_3<k)+\sum \sigma_3}$
\\
\midrule
$\ZZ_{2N}$ volume
& $d\ge 7$~(Theorem~\ref{thm: order 2 statistics for Z_2N})
& $\displaystyle {\prod}'_{\sigma_4 \sqcup \tau_4=\{1,2,3,4,5,6,7,8\}}
   \bigl[U_{0\sigma_4},\,U_{0\tau_4}\bigr]^{2N^2}$ \\
\bottomrule
\end{tabular}
\caption{Summary of the statistics of particle, loop, membrane, and volume excitations in spatial dimension $d$. Here, $U_{\Delta}$ denotes a Pauli operator supported near a simplex $\Delta$ that creates an excitation along its boundary $\partial\Delta$. For example, $U_{ij}$ is a string operator creating a particle at $\langle j\rangle$ and an antiparticle at $\langle i\rangle$. For each $d$, we use the triangulation $S^d=\partial\Delta^{d+1}$ illustrated in Fig.~\ref{fig: complex}. The label $\ZZ_N$ for the fusion group means that $\bigl(U_{\Delta}\bigr)^N$ acts trivially on the vacuum sector up to a phase. It does not imply $\bigl(U_{\Delta}\bigr)^N=1$ as an operator, nor does it require the underlying physical degrees of freedom to be $\ZZ_N$ qudits. Statistics are expressed as products of commutators $[A,B]:=A^{-1}B^{-1}AB$. The sets $\sigma_i$ and $\tau_i$ are disjoint index sets with $|\sigma_i|=|\tau_i|=i$, and the primed product indicates that the unordered partition $\{\sigma_i,\tau_i\}$ is counted only once. For instance, in the $\ZZ_N$ membrane statistic in $d=5$, the partitions $\{\sigma_3=\{1,2,3\},\,\tau_3=\{4,5,6\}\}$ and $\{\sigma_3=\{4,5,6\},\,\tau_3=\{1,2,3\}\}$ only contribute a single factor $[U_{0123},U_{0456}]^{N}$. These formulas are proved in Sec.~\ref{sec: Statistics in Pauli Stabilizer Models}.}
\label{tab: Pauli statistics}
\end{table}

\begin{figure*}[t]
    \centering
    \subfigure[(1+1)D]{\raisebox{1.5cm}{\includegraphics[scale=0.6]{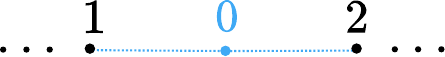}}\label{fig: complex (a)}}
    \hspace{0.05\linewidth}
    \subfigure[(2+1)D]{\includegraphics[scale=0.6]{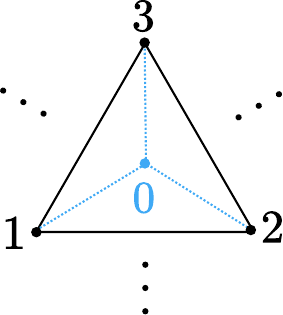}\label{fig: complex (b)}}
    \hspace{0.1\linewidth}
    \subfigure[(3+1)D]{\includegraphics[scale=0.6]{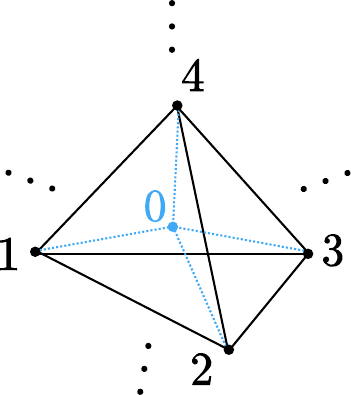}\label{fig: complex (c)}}
    \caption{Simplicial complexes used in different spacetime dimensions (adapted from Ref.~\cite{kobayashi2024generalized}). 
    In $(1{+}1)$D we use a segment subdivided by an interior vertex; in $(2{+}1)$D a triangle with a central vertex; and in $(3{+}1)$D a tetrahedron subdivided by a central vertex.
    Each complex is understood as embedded into a larger spatial manifold (indicated by $\cdots$).
    For convenience in computations, we often compactify the ambient manifold by wrapping the segment into a circle $S^ = \partial \Delta^2$, the triangle into a 2-sphere $S^2 = \partial \Delta^3$, and the tetrahedron into a 3-sphere $S^3 = \partial \Delta^4$.
    We assume that for any simplex $\Delta$ there exists a finite-depth unitary $U_\Delta$ that creates an invertible excitation on its boundary $\partial\Delta$.
    For example, the string operator $U_{01}$ creates a particle at vertex $1$ and an anti-particle at vertex $0$.
    Likewise, on the 2-simplex $\Delta_{012}$, the membrane operator $U_{012}$ creates a loop excitation along the boundary edges $\partial\Delta_{012}$.
    }
    \label{fig: complex}
\end{figure*}

%%%%%%%%%%%%%%%%%%%%%%%%%%%%%%%%%%%%%%%%%%%%%%%%%%%%%%%%
\section{Fermionic-loop toric code in $(4{+}1)$D}\label{sec: (4+1)D fermionic-loop toric code}

Before presenting the general construction of Pauli stabilizer models for $(4{+}1)$D loop-only topological orders, we first provide a concrete example. In this section, we construct a $(4{+}1)$D $\mathbb{Z}_2$ topological order that we refer to as the \emph{fermionic-loop toric code}.

Our construction is directly inspired by the $(2{+}1)$D double-semion stabilizer model of Ref.~\cite{Ellison2022Pauli}, where the double-semion phase is obtained by condensing the $e^2 m^2$ anyon in the $(2{+}1)$D $\mathbb{Z}_4$ toric code, yielding a $\mathbb{Z}_4$-qudit Pauli stabilizer realization of the double-semion order. The fermionic-loop toric code model developed here is a higher-dimensional analogue of this approach.
In Sec.~\ref{sec: Loop-only toric code in four spatial dimensions}, we introduce the $\ZZ_4$ loop-only toric code in four spatial dimensions, which is the direct generalization of the toric code in two spatial dimensions but with loop excitations rather than particle excitations.
In Sec.~\ref{sec: (4+1)D fermionic-loop toric code via condensation}, we construct the desired Pauli stabilizer model by condensing the $e^2 m^2$ loop excitations of a $\mathbb{Z}_4$ loop-only toric code.
We then demonstrate explicitly in Sec.~\ref{sec: Fermionic loop statistics in (4+1)D} that the resulting loop excitations exhibit $w_3$ fermionic statistics, using the 24-step process introduced in Ref.~\cite{kobayashi2024generalized}.\footnote{Equivalently, the 36-step process of Ref.~\cite{FHH21} can also be used to detect fermionic loop statistics.}
Finally, in Sec.~\ref{sec: Field-theoretic description in (4+1)D}, we present a field-theoretic analysis of our constructions, which provides a clear physical interpretation of the models.

%%%%%%%%%%%%%%%%%%%%%%%%%%%%%%%%%%%%%%%%%%%%%%
\subsection{Review of the $(4{+}1)$D loop-only $\ZZ_4$ toric code}\label{sec: Loop-only toric code in four spatial dimensions}

We begin by describing the Hilbert space of the $\mathbb{Z}_4$ loop-only toric code in $(4{+}1)$D.  
The spatial manifold is taken to be a four-dimensional hypercubic lattice.  
A $\mathbb{Z}_4$ qudit is placed on every two-dimensional face $f$ of the lattice, with computational basis states
\(
\{\ket{0},\ket{1},\ket{2},\ket{3}\} \cong \mathbb{Z}_4.
\)
The generalized Pauli operators acting on face $f$ are the shift and clock operators,
\begin{align}
X_f = \sum_{j\in \mathbb{Z}_4} |j+1\rangle\!\langle j|,
\qquad
Z_f = \sum_{j\in\mathbb{Z}_4} i^{\,j}\, |j\rangle\!\langle j|,
\end{align}
which in matrix form are
\begin{align}
X_f =
\begin{pmatrix}
0 & 0 & 0 & 1 \\
1 & 0 & 0 & 0 \\
0 & 1 & 0 & 0 \\
0 & 0 & 1 & 0
\end{pmatrix},
\qquad
Z_f =
\begin{pmatrix}
1 & 0 & 0 & 0 \\
0 & i & 0 & 0 \\
0 & 0 & -1 & 0 \\
0 & 0 & 0 & -i
\end{pmatrix}.
\end{align}
They satisfy
\begin{equation}
X_f^{4}=Z_f^{4}=1,
\qquad
Z_f X_{f'} =
\begin{cases}
i\, X_{f'} Z_f, & f=f', \\
X_{f'} Z_f, & f\neq f'.
\end{cases}
\end{equation}

We label the four unit lattice vectors by
\begin{eqs}
    e_1 = (1,0,0,0),\quad
    e_2 = (0,1,0,0),\quad
    e_3 = (0,0,1,0),\quad
    e_4 = (0,0,0,1).
\end{eqs}
For convenience, we write 
\(
\langle v;\alpha_1,\ldots,\alpha_p\rangle
\)
for the $p$-dimensional hypercube formed by the $2^p$ vertices
\[
\{\, v+\sum_{j=1}^p c_j \alpha_j \;|\; c_j\in\{0,1\} \,\}.
\]
For example, the square with vertices  
$(0,0,0,0)$, $(1,0,0,0)$, $(0,1,0,0)$, $(1,1,0,0)$  
is written as  
\[
\langle (0,0,0,0); e_1,e_2\rangle.
\]  
Note that this representation is not unique. For example, the same square can also be written as  
\[
\langle (1,0,0,0); -e_1, e_2\rangle,
\]  
which corresponds to choosing a different origin and reversing the orientation of one of the basis vectors.

The Hamiltonian is
\begin{equation}
    H = -\sum_{e\in\mathrm{edges}} A_e \;- \sum_{c\in\mathrm{cubes}} B_c \;+\; h.c.~,
\end{equation}
where the edge term for  
\(
e = \langle v; e_j\rangle
\)
is
\begin{equation}
A_{\langle v; e_j\rangle}
=
\prod_{\substack{1\le k\le 4\\ k\neq j}}
X_{\langle v; e_j,e_k\rangle}\;
X^{\dagger}_{\langle v; e_j,-e_k\rangle},
\end{equation}
and the cube term for  
\(
c=\langle v; e_j,e_k,e_l\rangle, \; j<k<l,
\)
is
\begin{eqs}
    B_{\langle v; e_j,e_k,e_l\rangle}
    &=
    Z_{\langle v; e_j,e_k\rangle}\;
    Z_{\langle v+e_k; e_j,e_l\rangle}\;
    Z_{\langle v; e_k,e_l\rangle}\;
    Z^{\dagger}_{\langle v+e_l; e_j,e_k\rangle}\;
    Z^{\dagger}_{\langle v; e_j,e_l\rangle}\;
    Z^{\dagger}_{\langle v+e_j; e_k,e_l\rangle}.
\end{eqs}
Thus, $A_e$ is the product of the $X_f$ operators over the \change{six} two-dimensional plaquettes $f$ containing the edge $e$, while $B_c$ is the product of the $Z_f$ operators over the six plaquettes $f$ in the boundary of the cube $c$, with signs determined by the orientation.
One could verify directly that all $A_e$ commute with all $B_c$, so the Hamiltonian is a stabilizer code.

Acting with a single $Z_f$ changes the eigenvalues of the four adjacent edge terms $A_e$ on $\partial f$, creating four violated stabilizers arranged along a closed loop. This corresponds to an $e$-loop excitation.  
Likewise, acting with $X_f$ flips the eigenvalues of the four cube terms $B_c$ sharing $f$, creating an $m$-loop excitation on the dual lattice.  
Both $e$- and $m$-loops are extended one-dimensional excitations, characteristic of the loop-only toric code in $(4{+}1)$D.  
No point-like excitations appear in this model.

Using the simplicial cohomology notation introduced in Refs.~\cite{Chen2023Highercup, sun2025cliffordQCA}, the stabilizers on an arbitrary triangulation take the compact form
\begin{equation}
    A_e = \prod_f X_f^{\,\delta\boldsymbol{e}(f)} := X_{\delta\boldsymbol{e}},
    \qquad
    B_t = \prod_f Z_f^{\,\boldsymbol{f}(\partial t)} := Z_{\partial t}.
\end{equation}
Here, $A_e$ multiplies $X_f$ over the coboundary of the edge $e$, while $B_t$ multiplies $Z_f$ over the boundary of the tetrahedron $t$, with the orientation signs already encoded in the definitions of the boundary and coboundary operators.  
This compact representation will be particularly convenient for our computation of the loop statistics in later sections.

%%%%%%%%%%%%%%%%%%%%%%%%%%%%%%%%%%%%%%%%%%%%%%%%%%%%%%%%%%%%%%%%%%%%%%%%%%%%%%%%%%%%%%%%%%
\subsection{$(4{+}1)$D fermionic-loop toric code via condensation}\label{sec: (4+1)D fermionic-loop toric code via condensation}

In this subsection, we construct a stabilizer model realizing the fermionic-loop toric code on the four-dimensional hypercubic lattice, with a $\mathbb{Z}_4$ qudit placed on every 2-dimensional face.  
We show that the resulting model implements the twisted 2-form $\mathbb{Z}_2$ gauge theory associated with the nontrivial cohomology class
\begin{equation}
    \frac{1}{4}\, b_2 \cup \delta b_2
    =
    \frac{1}{2}\, b_2 \cup (b_2 \cup_1 b_2)
    =
    \frac{1}{2}\, \mathrm{Sq}^2 \mathrm{Sq}^1 b_2
    \;\in\;
    H^5(K(\mathbb{Z}_2,2), \mathbb{R}/\mathbb{Z})\simeq\ZZ_2~.
\end{equation}
Our construction begins with the untwisted $(4{+}1)$D $\mathbb{Z}_4$ toric code reviewed in the previous subsection.  
We then condense the $e^2 m^2$ loop excitation.  
Operationally, this is achieved by introducing, for every face $f$, a hopping term $C_f$ that transports an $e^2 m^2$ loop across $f$.  
For a face written as $f=\langle v; e_j, e_k\rangle$, the term takes the form
\begin{equation}
    C_{\langle v; e_j,e_k\rangle}
    =
    X^2_{\langle v; e_j,e_k\rangle}\,
    Z^2_{\langle v; -e_l,-e_m\rangle},
\end{equation}
where $\{j,k,l,m\}=\{1,2,3,4\}$.  
Thus, each vertex $v$ contributes six such hopping terms:
\begin{equation}
    C_{\langle v; e_1,e_2\rangle},\;
    C_{\langle v; e_1,e_3\rangle},\;
    C_{\langle v; e_1,e_4\rangle},\;
    C_{\langle v; e_2,e_3\rangle},\;
    C_{\langle v; e_2,e_4\rangle},\;
    C_{\langle v; e_3,e_4\rangle}.
\end{equation}

Condensing the $e^2 m^2$ loop restricts the full stabilizer group of the $\mathbb{Z}_4$ toric code to the subgroup that commutes with all $C_f$.  
This commuting subgroup is generated by operators of the form
\begin{equation}
    A_{\langle v; -e_j\rangle}\,
    B_{\langle v; e_k,e_l,e_m\rangle},
    \qquad
    B_c^2,
\end{equation}
where $\{j,k,l,m\}=\{1,2,3,4\}$ and $c$ is a three-dimensional cube. For each vertex $v$, the first operator gives rise to four distinct stabilizers:
\begin{align}
    A_{\langle v; -e_1\rangle}\, B_{\langle v; e_2,e_3,e_4\rangle}, \quad
    A_{\langle v; -e_2\rangle}\, B_{\langle v; e_1,e_3,e_4\rangle}, \quad
    A_{\langle v; -e_3\rangle}\, B_{\langle v; e_1,e_2,e_4\rangle}, \quad
    A_{\langle v; -e_4\rangle}\, B_{\langle v; e_1,e_2,e_3\rangle}.
\end{align}

The Hamiltonian of the condensed theory becomes
\begin{equation}
    H_{\mathrm{condensed}}
    =-\!\!\sum_{\substack{\langle v; -e_j\rangle\in\mathrm{edges}\\ k,l,m\neq j}}
        A_{\langle v; -e_j\rangle} B_{\langle v; e_k,e_l,e_m\rangle}
      \;-\!
        \sum_{c} B_c^2
      \;-\!
        \sum_{f} C_f
      \;+\; \mathrm{h.c.}
\end{equation}
For convenience, we denote the first type of stabilizer by  
\begin{equation}
    G_e := A_{\langle v; -e_j\rangle}\,
           B_{\langle v; e_k,e_l,e_m\rangle}.
\end{equation}

\paragraph*{Loop excitations.}
The condensed theory supports two types of loop excitations. 

\begin{enumerate}
    \item \textbf{Charge loops.}  
    A charge loop is created by the membrane operator
    \begin{equation}
        V^C_f := Z_f^2,
    \end{equation}
    which changes the eigenvalues of $G_e$ for all edges $e \in \partial f$.  
    Thus, $V^C_f$ creates a charge loop along the boundary $\partial f$.  
    Since $(V^C_f)^2 = 1$, these loops have a manifest $\mathbb{Z}_2$ fusion rule.

    \item \textbf{Flux loops.}  
    Flux loops are created by membrane operators supported on faces
    $f = \langle v; -e_j, -e_k \rangle$:
    \begin{equation}
        V^F_{\langle v; -e_j,-e_k\rangle}
        :=
        X_{\langle v; -e_j,-e_k\rangle}\,
        Z_{\langle v; e_l,e_m\rangle},
    \end{equation}
    with $\{j,k,l,m\}=\{1,2,3,4\}$.  
    These operators commute with all hopping terms $C_f$ but violate certain stabilizers $G_e$ and $B_c^2$ whenever $f$ lies in the boundary of cube $c$.  
    The violated stabilizers form a loop on the dual lattice, so $V^F_f$ creates a flux-loop excitation supported on the coboundary $\delta \bface$.
\end{enumerate}

\paragraph*{Algebraic-topology formulation.}

Using the algebraic-topology notation of Refs.~\cite{CET2021, Chen2023Highercup, sun2025cliffordQCA}, the hopping term becomes
\begin{equation}
    C_f=
    X_f^{\,2}
    \prod_{f'} Z_{f'}^{2\int \boldsymbol{f}' \cup \boldsymbol{f}},
\end{equation}
and the Hamiltonian can be rewritten as
\begin{equation}
    H_{\mathrm{condensed}}
    = -\sum_e G_e
      \;- \sum_t B_t^{2}
      \;- \sum_f C_f
      \;+\; \mathrm{h.c.},
\label{eq: H_condensed fermionic-loop TC in 4+1D}
\end{equation}
where $t$ stands for a tetrahedron, and the first stabilizer can be expressed as
\begin{equation}
    G_e := A_e \prod_t B_t^{\int \boldsymbol{e} \cup \boldsymbol{t}}
    = X_{\delta \be} \prod_t Z_{\partial t}^{\int \boldsymbol{e} \cup \boldsymbol{t}} = X_{\delta \be}  \prod_{f} Z_{f}^{\int \delta \be \cup \bface}
    ~.
\label{eq: Ge in (4+1)D}
\end{equation}
The two membrane operators take the compact form
\begin{equation}
    V^C_f = Z_f^2,
    \qquad
    V^F_f = X_f \prod_{f'} Z_{f'}^{\int \boldsymbol{f} \cup \boldsymbol{f}'}.
\end{equation}
A subtlety arises when analyzing the fusion rule of loops.  
For the charge loop created by $V^C_f$, the fusion group is manifestly $\mathbb{Z}_2$, since $(V^C_f)^2 = 1$ as an operator. 
The flux loop created by $V^F_f$, however, behaves differently: it obeys $(V^F_f)^4 = 1$, not $(V^F_f)^2 \neq 1$, at the operator level, suggesting at first glance a $\mathbb{Z}_4$ fusion structure.  
The actual fusion rule, however, is determined not by the operator algebra alone but by how these operators act on ground states.  
One may verify that for any two-dimensional disk $D$ in the dual lattice, the operator  
\begin{equation}
    \prod_{f \in D^*} (V^F_f)^2
\end{equation}
can be written as a product of stabilizers supported within $D$, up to local operators near the boundary $\partial D$, where $D^*$ denotes the Poincaré dual.  
Consequently, the doubled flux loop created by $\prod_{f\in D^*} (V^F_f)^2$ is topologically trivial—it lies in the same \emph{superselection sector} as the vacuum loop excitation~\cite{Kitaev:2005hzj}.

To make this structure explicit, we decorate $V^F_f$ by a correction near the coboundary $\delta \bface$:
\begin{equation}
    \tilde{V}^F_f:=
    V^F_f \prod_{f'} Z_{f'}^{\int \boldsymbol{f}' \cup_1 \delta \bface}
    =
    X_f \prod_{f'}
    Z_{f'}^{\int \boldsymbol{f} \cup \boldsymbol{f}'
        + \boldsymbol{f}' \cup_1 \delta \bface}
    =
    X_f \prod_{f'}
    Z_{f'}^{\int \boldsymbol{f}' \cup \boldsymbol{f}
        - \delta \boldsymbol{f}' \cup_1 \bface},
\label{eq: modified flux loop in (4+1)D fermionic-loop toric code}
\end{equation}
where we have used the Leibniz rule of higher cup products~\cite{Steenrod1947, Gaiotto:2015zta, Tsui:2019ykk, Chen2020}
\begin{equation}
    \delta(A_2 \cup_1 B_2)
    =
    \delta A_2 \cup_1 B_2
    + A_2 \cup_1 \delta B_2
    - A_2 \cup B_2
    + B_2 \cup A_2~.
\end{equation}
We can check
\begin{equation}
    (\tilde{V}^F_f)^{2}
    =
    X_f^{2}
    \prod_{f'} Z_{f'}^{2\int \boldsymbol{f}' \cup \boldsymbol{f}
      - \delta \boldsymbol{f}' \cup_1 \bface}
    =
    C_f
    \prod_t B_t^{2\int \bt \cup_1 f},
\end{equation}
which is a product of stabilizers.  
Therefore, the flux loop created by $\tilde{V}^F_f$ also obeys a $\mathbb{Z}_2$ fusion rule.
Since the decoration in $\tilde{V}^F_f$ is supported only near the coboundary $\delta f$, $\tilde{V}^F_f$ remains a membrane operator that creates a loop excitation. To make this more explicit, it is convenient to choose different generators of the stabilizers in the Hamiltonian. In particular, we may rewrite the Hamiltonian as
\begin{equation}
    H'_{\mathrm{condensed}}
    = -\sum_e G_e
      \;- \sum_t B_t^{2}
      \;- \sum_f \tilde{C}_f
      \;+\; \mathrm{h.c.},
      \label{eq: H'_condensed fermionic-loop TC in 4+1D}
\end{equation}
with
\begin{equation}
    \tilde{C}_f=
    X_f^{\,2}
    \prod_{f'} Z_{f'}^{2\int \boldsymbol{f}' \cup \boldsymbol{f}+ \bface \cup_1 \delta \bface'}
    = C_f \prod_t B_t^{2 \int \bface \cup_1 \bt}
    ~.
\end{equation}
The operator $\tilde{C}_f$ differs from $C_f$ only by multiplication with appropriate $B_t^2$ terms, and therefore yields an equivalent Hamiltonian with the same ground-state subspace.
It is straightforward to verify that $\tilde{V}^F_{f_1}$ commutes with $\tilde{C}_{f_2}$:
\begin{eqs}
    [\tilde{V}^F_{f_1}, \tilde{C}_{f_2}] &= \left(\tilde{V}^F_{f_1}\right)^{-1} \left(\tilde{C}_{f_2}\right)^{-1} \tilde{V}^F_{f_1} \tilde{C}_{f_2} 
    =(-1)^{\int \bface_1 \cup \bface_2 + \bface_2 \cup_1 \delta \bface_1 + \bface_2 \cup \bface_1 + \delta \bface_2 \cup_1 \bface_1}
    =(-1)^{\int \delta (\bface_1 \cup_1 \bface_2)} = 1~.
    \label{eq: commutation between V_F and C}
\end{eqs}
Therefore, in the Hamiltonian $H'_{\mathrm{condensed}}$, the operator $\tilde{V}^F_f$ violates only the $B_t^2$ stabilizers for tetrahedra $t$ whose boundary contains $f$. This makes explicit that $\tilde{V}^F_f$ creates a well-defined loop excitation localized along the coboundary of $f$.

%%%%%%%%%%%%%%%%%%%%%%%%%%%%%%%%%%%%%%%%%%%%%%%%%%%%%%%%%%%%%%%%%%%%%
\subsection{Fermionic loop statistics}\label{sec: Fermionic loop statistics in (4+1)D}

\begin{figure}[tbh]
    \centering
    \vspace{-3em}
    \includegraphics[width=0.85\textwidth]{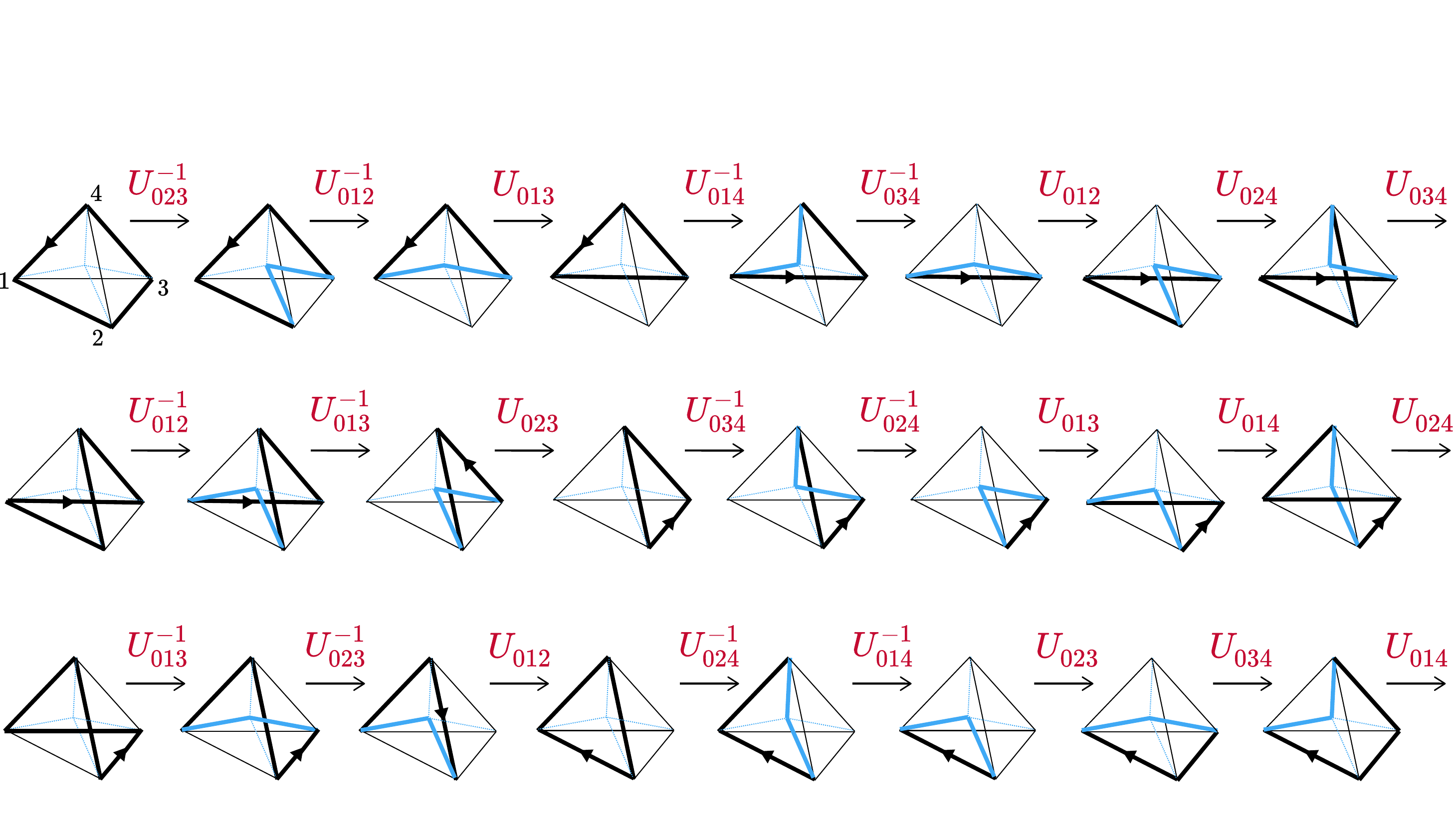}
    \vspace{-2em}
    \caption{The 24-step process for detecting the statistics of loop excitations with $\ZZ_2$ fusion in three spatial dimensions and higher (adapted from Ref.~\cite{kobayashi2024generalized}). The loop excitations are supported on the edges of the centered tetrahedron $\partial\langle 01234\rangle$. For $\ZZ_2$ loops, reversing the orientation does not change the physical configuration, so the initial and final states, related by an orientation flip, represent the same loop.}
    \label{fig: 24 step process}
\end{figure}

In $(3{+}1)$D and higher dimensions, the loop statistic is determined by the loop-flipping process,
\begin{eqs}
    \mu_{24} := ~& U_{014} U_{034} U_{023} U_{014}^{-1} U_{024}^{-1} U_{012} U_{023}^{-1} U_{013}^{-1} \\
    \times & U_{024} U_{014} U_{013} U_{024}^{-1} U_{034}^{-1} U_{023} U_{013}^{-1} U_{012}^{-1} \\
    \times & U_{034} U_{024} U_{012} U_{034}^{-1} U_{014}^{-1} U_{013} U_{012}^{-1} U_{023}^{-1}~,
    \label{eq: 24-step process}
\end{eqs}
where each unitary $U_{ijk}$ creates a loop excitation on $\partial\langle ijk\rangle=\langle ij\rangle+\langle jk\rangle-\langle ik\rangle$. 
This process is shown explicitly in Fig.~\ref{fig: 24 step process}.
{\change
The physical intuition behind Eq.~\eqref{eq: 24-step process} is provided in Appendix~\ref{app: statistics as anomaly}. 
}

When the membrane operators $U_f$ are Pauli operators, their commutator
\begin{equation}
    [U_{f_1}, U_{f_2}] := U_{f_1}^{-1} U_{f_2}^{-1} U_{f_1} U_{f_2}
\end{equation}
is a $U(1)$ phase, and the 24-step process simplifies to
\begin{equation}
    \mu_{24}
    = 
    [U_{012}, U_{034}]^{2}\,
    [U_{013}, U_{024}]^{2}\,
    [U_{014}, U_{023}]^{2}.
    \label{eq: 24-step process in Pauli}
\end{equation}
We now prove the following.
\begin{theorem}
    The flux loop created by $\tilde{V}^F_f$ in Eq.~\eqref{eq: modified flux loop in (4+1)D fermionic-loop toric code} has fermionic loop statistics, $\mu_{24}=-1$.
\end{theorem}
\begin{proof}
In principle, this loop statistic could be evaluated directly on a four-dimensional hypercubic lattice. However, visualizing the 24-step process in four dimensions is impractical. We therefore work entirely in an algebraic formulation, which not only avoids this difficulty but also generalizes straightforwardly to higher dimensions; we will reuse the same strategy in Sec.~\ref{sec: Fermionic-loop toric codes in arbitrary dimensions}.

We extend the definition of the operator to an arbitrary $2$-cochain $\lambda\in C^2(M_4,\ZZ)$ (where $M_4$ is the spatial manifold) by
\begin{equation}
    \tilde{V}^F_\lambda
    =
    \prod_f X_f^{\lambda(f)}
    \prod_{f'} Z_{f'}^{\int \boldsymbol{f}'\cup\lambda
        - \delta\boldsymbol{f}'\cup_1 \lambda}
    :=
    X_\lambda
    \prod_{f'} Z_{f'}^{\int \boldsymbol{f}'\cup\lambda
        - \delta\boldsymbol{f}'\cup_1 \lambda}
    ~.
\end{equation}
A direct calculation gives the basic commutation relation
\begin{eqs}
    [\tilde{V}^F_{\lambda},\tilde{V}^F_{\lambda'}]
    =&~
    (\tilde{V}^F_{\lambda})^{-1}
    (\tilde{V}^F_{\lambda'})^{-1}
    \tilde{V}^F_{\lambda}
    \tilde{V}^F_{\lambda'} \\
    =&~
    i^{\int \lambda'\cup\lambda-\lambda\cup\lambda'
        + \delta\lambda\cup_1\lambda'
        - \delta\lambda'\cup_1\lambda} \\
    =&~
    i^{\int -\delta\lambda\cup_1\lambda'
        - \lambda\cup_1\delta\lambda'
        + \delta\lambda\cup_1\lambda'
        - \delta\lambda'\cup_1\lambda } \\
    =&~
    i^{\int -\lambda\cup_1\delta\lambda'
        - \delta\lambda'\cup_1\lambda }
    ~=~
    i^{-\int \delta\lambda\cup_2\delta\lambda'}~.
\label{eq: (4+1)D commutator of V}
\end{eqs}

Flux-loop excitations reside on the dual lattice, so the complexes in Fig.~\ref{fig: 24 step process} should be interpreted as dual simplices.
For the dual triangle $\langle 012\rangle$, we define $\lambda_{012}:=(012)^*$ as its Poincar\'e-dual cochain and write $U_{012}:=\tilde{V}^F_{\lambda_{012}}$.
Using Eq.~\eqref{eq: 24-step process in Pauli} and Eq.~\eqref{eq: (4+1)D commutator of V}, we obtain
\begin{equation}
    \mu_{24}
    =
    (-1)^{\int \delta\lambda_{012}\cup_2\delta\lambda_{034}
    + \delta\lambda_{013}\cup_2\delta\lambda_{024}
    + \delta\lambda_{014}\cup_2\delta\lambda_{023}}~.    
\end{equation}
We next decompose
\[
\delta\lambda_{012}
= \delta (012)^*
= (\partial\,012)^*
= (01)^* + (12)^* - (02)^*
= \nu_{01}+\nu_{12}-\nu_{02}~,
\]
where we denote $(ij)^*$ by $\nu_{ij}$, a $3$-cochain on the direct lattice.
If $i,j,k,l$ are all distinct, then $\nu_{ij}$ and $\nu_{kl}$ have disjoint support and hence $\nu_{ij}\cup_2\nu_{kl}=0$.
It follows that
\begin{eqs}
\mu_{24}
=&~
(-1)^{\int
(\nu_{01}+\nu_{02})\cup_2(\nu_{03}+\nu_{04})
+ (\nu_{01}+\nu_{03})\cup_2(\nu_{02}+\nu_{04})
+ (\nu_{01}+\nu_{04})\cup_2(\nu_{02}+\nu_{03})}
\\
=&~
(-1)^{\int
\nu_{02}\cup_2\nu_{03}
+ \nu_{03}\cup_2\nu_{02}
+ \nu_{02}\cup_2\nu_{04}
+ \nu_{04}\cup_2\nu_{02}
+ \nu_{03}\cup_2\nu_{04}
+ \nu_{04}\cup_2\nu_{03}}~.
\end{eqs}
Using the Leibniz rule for the $\cup_3$ product,
\begin{equation}
    \delta (A_3 \cup_3 B_3)
    =
    \delta A_3 \cup_3 B_3
    - A_3 \cup_3 \delta B_3
    - A_3 \cup_2 B_3
    - B_3 \cup_2 A_3~,
\end{equation}
we rewrite this as
\begin{eqs}
\mu_{24}
=&~
(-1)^{\int
\delta\nu_{02}\cup_3\nu_{03}
+ \nu_{02}\cup_3\delta\nu_{03}
+ \delta\nu_{02}\cup_3\nu_{04}
+ \nu_{02}\cup_3\delta\nu_{04}
+ \delta\nu_{03}\cup_3\nu_{04}
+ \nu_{03}\cup_3\delta\nu_{04}}~.
\end{eqs}
Let $\delta\nu_{ij}= p_i - p_j$ with $p_i:=(i)^*$.
Moreover, $\nu_{ij}\cup_3 p_k=0$ whenever $i,j,k$ are distinct.
Therefore,
\begin{eqs}
\mu_{24}
=&~
(-1)^{\int
p_0\cup_3\nu_{03}
+ \nu_{02}\cup_3 p_0
+ p_0\cup_3\nu_{04}
+ \nu_{02}\cup_3 p_0
+ p_0\cup_3\nu_{04}
+ \nu_{03}\cup_3 p_0}
\\
=&~ (-1)^{\int p_0 \cup_3 \nu_{03} + \nu_{03} \cup_3 p_0 }
= ~ (-1)^{\int p_0 \cup_4 \delta \nu_{03}}
=~
(-1)^{\int p_0\cup_4 p_0 }~.
\end{eqs}
Since $p_0$ is a $4$-cochain, the definition of the $\cup_4$ product implies $p_0 \cup_4 p_0 = p_0 \pmod{2}$.
Hence,
\begin{equation}
    \mu_{24} = (-1)^{\int (0)^*} = -1~.
\end{equation}
The structure of the calculation is transparent: higher cup products encode successive failures of commutativity of lower cup products.
Although $\tilde{V}^F_\lambda$ is defined using only ordinary cup products and $2$-cochains, the evaluation of $\mu_{24}$ forces---via coboundary identities and locality---the appearance of higher cup-$i$ products and higher-degree cochains.
Ultimately, $\mu_{24}$ detects a higher-order noncommutativity of the operators $\tilde{V}^F_\lambda$, yielding the fermionic loop statistics.
\end{proof}

%%%%%%%%%%%%%%%%%%%%%%%%%%%%%%%%%%%%%%%%%%
\subsection{Field-theoretic description}\label{sec: Field-theoretic description in (4+1)D}

In this section, we provide a field-theoretic intuition for the above construction.
We start from the $\ZZ_4$ loop-only toric code, described by the topological action
\begin{equation}
    S = \frac{2\pi}{4} \int a_2 \cup \delta b_2~,
\label{eq: S standard Z4 TC in (4+1)D}
\end{equation}
where $a_2$ and $b_2$ are $\ZZ_4$-valued $2$-cochains on the spacetime manifold. 
Physically, $a_2$ and $b_2$ represent the worldsheets of $m$ and $e$ loops, respectively. The equations of motion are
\begin{equation}
    \delta a_2 = 0 \pmod{4}, \qquad \delta b_2 = 0 \pmod{4},
\end{equation}
indicating that both $e$ and $m$ are $\ZZ_4$ bosonic loops.

To condense the $e^2 m^2$ loop, we introduce a dynamical $\ZZ_2$ $3$-cocycle as described by the pair of cochains $(c_3,u)$ and couple it to the worldsheet of $e^2 m^2$:
\begin{equation}
    S_{\mathrm{condensed}} = \frac{2\pi}{4} \left(\int a_2 \cup \delta b_2 
    + 2(a_2+b_2) \cup c_3\right)+\pi\int \delta u\cup c_3~.
\label{eq: S condensed in (4+1)D}
\end{equation}
Physically, condensing the $e^2m^2$ loop means that we insert the operators that create and annihilate these loops everywhere in the path integral. The new dynamical cocycle is the Poincar\'e dual for the support of these operators.
{\change
This is analogous to the standard statement that fluctuating domain walls of a symmetry can be described by a dynamical gauge field in the transverse direction.}

Integrating out $c_3$ treats it as a Lagrange multiplier, imposing the constraint
\begin{equation}
    2 (a_2 + b_2) = 2\delta u \pmod{4}~.
\end{equation}
Thus we can write
\begin{equation}
    a_2 = b_2 + 2 \tilde{a}_2~,
\end{equation}
for some $\ZZ_2$-valued $2$-cochain $\tilde{a}_2$. Substituting back into the action gives
\begin{eqs}
    S_{\mathrm{condensed}}
    &= \frac{2\pi}{4} \int (b_2 + 2 \tilde{a}_2) \cup \delta b_2 = \pi \int \tilde{a}_2 \cup \delta b_2 
       + \frac{2\pi}{4} \int b_2 \cup \delta b_2~.
\label{eq: S' of (4+1)D DW TQFT}
\end{eqs}
Here $b_2$ is no longer constrained to be a $\ZZ_4$ cocycle: the equation of motion reduces to
\begin{equation}
    \delta b_2 = 0 \pmod{2}~.
\end{equation}
We may therefore regard $b_2$ as a $\ZZ_2$-valued $2$-cocycle.
The remaining theory is precisely the Dijkgraaf–Witten TQFT for a 2-form $\ZZ_2$ gauge symmetry, with topological term
\begin{equation}
    \frac{1}{4}  b_2 \cup \delta b_2 
    = \frac{1}{2}  \mathrm{Sq}^2 \mathrm{Sq}^1 b_2 
    \;\in\; H^5(B^2\ZZ_2, \RR/\ZZ)~,
\end{equation}
and furthermore, when integrated on closed oriented manifolds, this equals \cite{Kapustin:2017jrc, Gukov:2020btk}
\begin{equation}
     \frac{2\pi}{4}\int b_2\cup \delta b_2= \pi\int  w_3 \cup b_2~,
\end{equation}
where $w_3$ is the third Stiefel–Whitney class. If we use the last expression and substitute back into $S'$ in Eq.~\eqref{eq: S' of (4+1)D DW TQFT}, the action becomes
\begin{equation}
    S_{\mathrm{condensed}} =  \pi \int \tilde{a}_2 \cup \delta b_2 
        + \pi\int w_3 \cup b_2~,
\end{equation}
whose equation of motion, obtained by integrating out $b_2$, reads
\begin{equation}
    \delta \tilde{a}_2 = w_3 \pmod{2}~.
\end{equation}
This shows that the worldsheet $\tilde{a}_2$ is sourced by $w_3$ and therefore describes a fermionic loop excitation. {\change The relation between the dynamical field $\tilde{a}_2$ and $w_3$ can be regarded as a fractionalization similar to Refs.~\cite{Hsin:2019gvb,Hsin:2021qiy}. The connection between such fractionalization and fermionic loop statistics is explained in Ref.~\cite{CH21}.}

Now, we take a backward viewpoint: we first embed the target $\ZZ_2$ theory into a $\ZZ_4$ gauge theory that admits a Pauli-stabilizer realization, and then project back to $\ZZ_2$ by enforcing an order-$2$ constraint via condensation. This logic generalizes straightforwardly and applies uniformly in higher dimensions.
Our target theory is the nontrivial twisted $\ZZ_2$ 2-form gauge theory in Eq.~\eqref{eq: S' of (4+1)D DW TQFT}, whose twist is represented by the nontrivial cocycle in $H^5(B^2\ZZ_2,\RR/\ZZ)$.
Although the cochain expression $b_2\cup \delta b_2$ is quadratic, the coefficient $\frac{1}{4}$ obstructs a realization by a Pauli stabilizer Hamiltonian on $\ZZ_2$ qubits: it requires at least the $S=\sqrt{Z}$ gate.
A standard way around this is to \emph{trivialize the cocycle after a group extension}, i.e., to lift the $\ZZ_2$ gauge field to a $\ZZ_4$ gauge field.
Concretely, the same cocycle becomes trivial upon promoting $b_2$ to a $\ZZ_4$ 2-cocycle: for $\delta b_2=0$ mod 4,
\begin{equation}
    \left[\frac{1}{4}\, b_2 \cup \delta b_2 \right] = [0]
    \in H^5(B^2 \ZZ_4, \RR/\ZZ)~.
\end{equation}
Moreover, in a $\ZZ_4$ qudit Pauli algebra, fourth roots of unity are available, so the factor $\frac{1}{4}$ is compatible with a Pauli stabilizer realization.

This corresponds to the following $\ZZ_4$ action:
\begin{equation}
    S' = \frac{2\pi}{4} \int \bigl(a_2' \cup \delta b_2 + b_2 \cup \delta b_2\bigr)~,
\end{equation}
which is in the same topological phase as the standard $\ZZ_4$ loop-only toric code in Eq.~\eqref{eq: S standard Z4 TC in (4+1)D}.
Indeed, redefining $a_2 := a_2' + b_2$ brings $S'$ to $\frac{1}{4}\int a_2\cup \delta b_2$.

To recover the desired twisted $\ZZ_2$ theory, we impose that $a_2'$ takes values in the order-$2$ subgroup of $\ZZ_4$, i.e., $a_2' = 2\tilde a_2$ (equivalently $2a_2' = 0 \pmod{4}$).
This constraint can be enforced by introducing a $\ZZ_2$ Lagrange multiplier $c_3$:
\begin{equation}
    S'_{\mathrm{condensed}}
    = \frac{2\pi}{4}\int\left( \bigl(a_2' \cup \delta b_2 + b_2 \cup \delta b_2\bigr)
      + 2 a_2' \cup c_3\right)+\pi\int \delta u\cup  c_3~.
\label{eq: S' condensed in (4+1)D}
\end{equation}
Since $a_2 = a_2' + b_2$, Eq.~\eqref{eq: S' condensed in (4+1)D} provides a straightforward explanation of why the lattice construction amounts to condensing the $e^2m^2$ loop in Eq.~\eqref{eq: S condensed in (4+1)D}.
We emphasize that, although $\tfrac{1}{4}\, b_2 \cup \delta b_2$ is a trivial cocycle, it still influences the resulting theory once we condense the $2a_2'$ loop excitations.

This logic extends directly to any $(d{+}1)$D with $d\ge 4$.
The target theory is a twisted $\ZZ_2$ $(d-2)$-form gauge theory classified by
$H^{d+1}(B^{d-2}\ZZ_2,\RR/\ZZ)$:
\begin{eqs}
    S'_{\mathrm{target}}
    &= \pi\int \tilde a_2 \cup \delta b_{d-2}
      + \pi\int w_3 \cup b_{d-2} \\
    &= \pi\int \tilde a_2 \cup \delta b_{d-2}
      + \frac{2\pi}{4}\int \Bigl(
            b_{d-2}\cup_{d-5} b_{d-2}
          + b_{d-2}\cup_{d-4}\delta b_{d-2}
        \Bigr)~,
\end{eqs}
where the relation between $w_3$ and the stable $\ZZ_2$ cohomology operation (e.g. $Sq^2Sq^1$) is reviewed in Appendix~\ref{app: Calculation of stable Z2 operations}.
As before, we lift $b_{d-2}$ to a $\ZZ_4$ cocycle and consider
\begin{equation}
    S' = \frac{2\pi}{4}\int a_2' \cup \delta b_{d-2}
      + \frac{2\pi}{4}\int \Bigl(
            b_{d-2}\cup_{d-5} b_{d-2}
          + b_{d-2}\cup_{d-4}\delta b_{d-2}
        \Bigr)~.
\label{eq: S' for higher dimensions}
\end{equation}
When $b_{d-2}$ is $\ZZ_4$-closed (i.e. $\delta b_{d-2}=0\pmod{4}$), the second integral in Eq.~\eqref{eq: S' for higher dimensions} is a coboundary.
One convenient way to see this is to use the higher-cup coboundary identity to obtain
\begin{eqs}
    \delta\Bigl(b_{d-2}\cup_{d-4} b_{d-2}
      + b_{d-2}\cup_{d-3}\delta b_{d-2}\Bigr)
    = 2(-1)^d \Bigl(
        b_{d-2}\cup_{d-5} b_{d-2}
      + b_{d-2}\cup_{d-4}\delta b_{d-2}
      \Bigr)
      + \delta b_{d-2}\cup_{d-3}\delta b_{d-2}~, \nonumber
\end{eqs}
so dividing by $8$ and using $\delta b_{d-2}=0\pmod{4}$ gives
\begin{equation}
    \frac{1}{4}\Bigl(
        b_{d-2}\cup_{d-5} b_{d-2}
      + b_{d-2}\cup_{d-4}\delta b_{d-2}
    \Bigr)
    = \frac{(-1)^d}{8}\,
      \delta\Bigl(b_{d-2}\cup_{d-4} b_{d-2}
      + b_{d-2}\cup_{d-3}\delta b_{d-2}\Bigr)
    \pmod{1}~.
\end{equation}
Hence, it does not affect the topological phase, and Eq.~\eqref{eq: S' for higher dimensions} is in the same phase as the standard $\ZZ_4$ toric code $\tfrac{1}{4}\int a_2 \cup \delta b_{d-2}$.

Finally, to return to the desired twisted $\ZZ_2$ theory, we again impose $a_2' = 2\tilde a_2$ by adding a Lagrange multiplier $c_{d-1}$:
\begin{equation}
    S'_{\mathrm{condensed}}
    = \frac{2\pi}{4}\int\left( a_2' \cup \delta b_{d-2}
      +\Bigl(
            b_{d-2}\cup_{d-5} b_{d-2}
          + b_{d-2}\cup_{d-4}\delta b_{d-2}
        \Bigr)
      +  2a_2' \cup c_{d-1}\right)+\pi\int \delta u\cup c_{d-1}~.
\label{eq: S' condensed in (d+1)D}
\end{equation}
In Sec.~\ref{sec: Fermionic-loop toric codes in arbitrary dimensions}, we implement this derivation on a lattice and construct the Pauli stabilizer Hamiltonian of the fermionic-loop toric code in $(d{+}1)$D using $\ZZ_4$ qudits.

%%%%%%%%%%%%%%%%%%%%%%%%%%%%%%%%%%%%%%%%%%%%%%%%%%%%%%%%%%%%%%%%%%%%%%%%%%%%%
\section{Dijkgraaf--Witten 2-form gauge theories in $(4{+}1)$D}
\label{sec: Dijkgraaf-Witten 2-form gauge theories}

In this section, we construct Pauli stabilizer Hamiltonians for twisted $2$-form Dijkgraaf--Witten gauge theories in $(4{+}1)$ spacetime dimensions. 
These theories are classified by cocycles in $H^5(B^2G,\RR/\ZZ)$, where $G$ is the $2$-form gauge group and $B^2G=K(G,2)$ is the corresponding Eilenberg--MacLane space. 
The $(4{+}1)$D fermionic-loop toric code of Sec.~\ref{sec: (4+1)D fermionic-loop toric code} is recovered as the special case $G=\ZZ_2$.

Our construction proceeds in two steps.
First, at the field-theory level, we show how a given cocycle class can be realized by gauging appropriate $2$-form symmetries of standard toric codes with enlarged gauge groups; this analysis identifies the bosonic loop excitations that should be condensed.
Second, we implement the corresponding condensation on the lattice: starting from a stack of loop-only toric codes with qudits on $2$-simplices, we condense the selected bosonic loops to obtain a stabilizer Hamiltonian whose ground state realizes the desired twisted $2$-form gauge theory.

%%%%%%%%%%%%%%%%%%%%%%%%%%%%%%%%%%%%%%%%%%%%%%%%%%%%%%%%%%%%%%%%%%%%%%%%%%%%%
\subsection{Condensing layers of toric code}

As discussed in Refs.~\cite{Johnson-Freyd:2021tbq, Cordova:2023bja, Hsin2025NonAbelian}, if a $(4{+}1)$D topological order supports only loop excitations, then these excitations must obey Abelian fusion rules. 
Consequently, such loop-only topological orders are necessarily Abelian. 
Their fusion rules are therefore described by a finite Abelian group
\begin{equation}
    G=\prod_{j=1}^s \ZZ_{N_j}~.
\end{equation}
Without loss of generality, we impose the divisibility condition $N_j\mid N_{j+1}$ for $1\le j\le s-1$, i.e., we work in the invariant factor decomposition of $G$.

The Dijkgraaf--Witten classification of $2$-form $G$ gauge theories is
\begin{equation}
    H^5(B^2G,\RR/\ZZ)
    = \prod_j \ZZ_{(N_j,2)} \prod_{j<k} \ZZ_{(N_j,N_k)},
\end{equation}
where $(\alpha,\beta)$ denotes the greatest common divisor of $\alpha$ and $\beta$.
Let $b_j$ be a $\ZZ_{N_j}$-valued $2$-form field. A convenient representative cocycle is
\begin{equation}
    \sum_j \frac{l_j}{N_j}\int b_j \cup \left(\frac{\delta b_j}{N_j}\right)
    + \sum_{j<k} \frac{l_{jk}}{N_j}\int b_j \cup \left(\frac{\delta b_k}{N_k}\right),
\label{eq: H^5 of general G}
\end{equation}
with $l_j=0,1,\dots,(N_j,2)-1$ and $l_{jk}=0,1,\dots,(N_j,N_k)-1$.

Here $\delta b_j/N_j$ denotes the Bockstein cochain $\mathrm{Bock}(b_j)\in C^3(M,\ZZ_{N_j})$ associated with the short exact sequence
$0\to \ZZ_{N_j}\to \ZZ_{N_j^2}\to \ZZ_{N_j}\to 0$.
Concretely, choose any lift $\tilde b_j\in C^2(M,\ZZ_{N_j^2})$ satisfying $\tilde b_j\equiv b_j \pmod{N_j}$.
Then $\delta\tilde b_j$ is divisible by $N_j$, and we define
\[
    \frac{\delta b_j}{N_j}
    := \mathrm{Bock}(b_j)
    := \frac{1}{N_j}\,\delta\tilde b_j
    \in C^3(M,\ZZ_{N_j})~.
\]

To realize a twisted $2$-form gauge theory, we introduce dynamical $2$-form gauge fields $a_j,b_j\in C^2(M^5,\ZZ_{N_j})$ and start from the untwisted action
\begin{equation}
    \sum_j \frac{1}{N_j}\int a_j\cup \delta b_j~.
\end{equation}
We then add the cocycle term in Eq.~\eqref{eq: H^5 of general G}. The full action becomes
\begin{align}
    S
    &=
    2\pi \sum_j\left(
        \frac{1}{N_j}\int a_j \cup \delta b_j
        + \frac{l_j}{N_j}\int b_j \cup \frac{\delta b_j}{N_j}
    \right)
    +2\pi\sum_{j<k}\frac{l_{jk}}{N_j}\int b_j \cup \frac{\delta b_k}{N_k}~.\label{action of twist theory}
\end{align}

We now explain how to obtain an arbitrary twisted $(4{+}1)$D loop-only theory by condensing bosonic loop excitations in a stack of toric codes. 
While there are many possible condensation schemes, we present a systematic choice that will also be used in our lattice stabilizer construction.

We start from a stack of $\ZZ_{N_j^2}$ toric codes, with action
\begin{equation}
    S_0
    =
    2\pi\sum_{j}\frac{1}{N_j^2}\int a_j\cup \delta b_j~.
\end{equation}
For each $j$, we condense the bosonic loop excitation $\varphi_j^{N_j}$, where
\begin{equation}
    \varphi_j
    :=
    m_j\, e_j^{\,l_j}\prod_{k>j} e_k^{\,l_{jk}N_k/N_j}~.
\end{equation}
In the field-theoretic description, this condensation can be implemented by introducing auxiliary dynamical fields $p_j$ and adding the term
\begin{equation}
    2\pi\int p_j\cup\left(
        \frac{1}{N_j}a_j
        -\frac{l_j}{N_j}b_j
        -\sum_{k>j}\frac{l_{jk}}{N_k}b_k
    \right).
\end{equation}
Integrating out $p_j$ imposes the constraint
\begin{equation}
    a_j - l_j b_j - \sum_{k>j} l_{jk}\frac{N_j}{N_k} b_k
    = N_j x_j~,
\end{equation}
where $x_j$ is an integer-valued field. Substituting this back into $S_0$, we obtain
\begin{align}
    S
    &=
    2\pi\sum_{j}\frac{1}{N_j^2}\int
    \left(
        l_j b_j
        +\sum_{k>j} l_{jk}\frac{N_j}{N_k} b_k
        + N_j x_j
    \right)\cup \delta b_j \nonumber\\
    &=
    2\pi\sum_{j}\frac{l_j}{N_j^2}\int b_j\cup \delta b_j
    +2\pi\sum_{j<k}\frac{l_{jk}}{N_jN_k}\int b_k\cup \delta b_j
    +2\pi\sum_j\frac{1}{N_j}\int x_j\cup \delta b_j~,
\end{align}
which reproduces precisely the general twisted $(4{+}1)$D loop-only theory specified by the cocycle data in Eq.~\eqref{eq: H^5 of general G}.
 
%%%%%%%%%%%%%%%%%%%%%%%%%%%%%%%%%%%%%%%%%%%%%%%%%%%%%%%%%%%%%%%%%%%%%%%%%%%%%
\subsection{Pauli Stabilizer Formalism}

Having obtained the twisted $(4{+}1)$D theory action from condensing the loop excitations in the field theory, now we repeat the process in the Hamiltonian formalism. We obtain the stabilizer model for 2-form gauge theories by condensing certain excitations from multiple layers of toric code. 

We consider $s$ layers of toric code, with the group of $i$th toric code $\Z_{N_i^2}$. For each toric code, there is a group of $X$ and $Z$ operators on each face $f$, which we denote by $X_{j,f}$ and $Z_{j,f}$. From the stabilizer model of the $j$th toric code, we have the following commuting set of operators:
\begin{align}
    A_{j,e}=\prod_{f} X_{j,f}^{\delta \boldsymbol{e}(f)}:=X_{j,\delta \boldsymbol{e}}, 
    \qquad
    B_{j,t}= \prod_f Z_{j,f}^{\,\boldsymbol{f}(\partial t)} := Z_{j,\partial t}.\label{Aje and Bjt}
\end{align}
The hopping terms for $e_j$ and $m_j$ loop excitations at $f$ are given by
\begin{align}
    e_j:\ X_{j,f},
    \qquad
    m_j:\ \prod_{f'}Z_{j,f'}^{\int \bface'\cup \bface}.
\end{align}
Here $X_{j,f}$ creates a loop excitation at its coboundary $\delta f$, and $Z_{j,f}$ creates a loop excitation at its boundary $\partial f$. The hopping operator of a mixture of different types of loop excitations is given by the product of individual operators. In particular, the hopping operators for 
\begin{align}
    \varphi_j^{N_j}=m_j^{N_j}e_j^{N_jl_j}\prod_{k>j} e_k^{l_{jk}N_k}
\end{align}
are given by
\begin{eqs}
    W_{j,f}=\left(\prod_{f'}Z_{j,f'}^{{\change N_j\int \bface'\cup \bface}}\right)
    X_{j,f}^{l_jN_j}\prod_{k>j}X_{k,f}^{l_{jk}N_k}~.
\end{eqs}
These operators commute with each other under the assumption $N_j\mid N_{j+1}$. However, they only commute with a subgroup of the stabilizers in Eq.~\eqref{Aje and Bjt}. To identify this subgroup, consider general combinations of $A_{j,e}$ and $B_{j,t}$ as follows
\begin{eqs}
    G_{e}^{(n,m)}=\prod_{j} X_{j,\delta e}^{n_j}\prod_{k,f}Z_{k,f}^{m_k\int \delta \be \cup \bface}~,
\end{eqs}
where $n=(n_1, n_2,\cdots,n_s)$ and $m=(m_1,m_2,\cdots, m_s)$ are $s$-tuples that label the type of combination. 

For $G_e^{(n,m)}$ to be included in the condensed stabilizer model, we require it to commute with each $W_{j,f}$. By calculating the commutators $[G_e^{(n,m)}, W_{j,f}]$, we obtain the following condition for $(n,m)$, 
\begin{eqs}
    \sum_j\frac{\delta_{ij}}{N_j}n_j-\sum_{k\geq i}\frac{l_{ik}}{N_k}m_k\in \ZZ,
\end{eqs}
for each $1\leq i\leq s$. We can rewrite this condition into a compact matrix equation
\begin{eqs}
    n=NLN^{-1}m+Np~, \label{mat criterion}
\end{eqs}
where $n,m,p\in \ZZ^{s}$, and the matrices $N$ and $L$ are defined as
\begin{align}
    N=\mathrm{diag}(N_1,N_2,\cdots,N_s),
\end{align}
\begin{align}
    L=
    \begin{pmatrix}
        l_1 & l_{12} & \cdots & l_{1s}\\
        0 & l_2 & \cdots & l_{2s}\\
        0&0&\ddots & \vdots\\
        0&0&\cdots&l_s
    \end{pmatrix}.
\end{align}

The tuples $(n,m)$ that satisfy the equation above form a lattice $\Lambda$ in $\Z^{2s}$. They give all possible combinations of the original stabilizers that commute with the condensed hopping operators. Finally, the Hamiltonian of the condensed theory becomes
\begin{eqs}
    H_{\mathrm{condensed}}=-\sum_{j,f} W_{j,f}-\sum_{\substack{e\in\mathrm{edges} \\ (n,m)\in\Lambda}} G_e^{(n,m)} +\mathrm{h.c.}~
\end{eqs}

\paragraph*{Loop Excitations.}
For each $(n,m)\in \Lambda$, there is a type of loop excitations, created by the following membrane operator
\begin{eqs}
    U_{f}^{(n,m)}=\prod_{j}X_{j,f}^{n_j}\prod_{k,f'}Z_{k,f'}^{m_k\int \bface\cup \bface'-\delta \bface \cup_1 \bface'}~.
\end{eqs}
It commutes with the $W_{j,f}$ stabilizers, and only violates $G_e^{(n,m)}$ whenever $e\in\delta f$. Therefore, $U_f^{(n,m)}$ creates a $(n,m)$-type loop along the coboundary $\delta f$. The fusion group is given by $\Lambda$ modulo the group generated by condensed excitations. 

{\change
\subsubsection{Example from $\ZZ_2\times \ZZ_2$ Gauge Theory}

Now, let us illustrate the above general formalism by applying it to a specific model, namely, the $\ZZ_2\times \ZZ_2$ type-II twisted gauge theory. In terms of field theory, it is defined by the following action,
\begin{eqs}
    S=2\pi\cdot \frac{1}{2}\int b_1\cup \frac{\delta b_2}{2}~,
\end{eqs}
where $b_1$ and $b_2$ are two $\ZZ_2$-valued gauge fields. Following the general procedure, we first extend the gauge group to $\ZZ_4\times \ZZ_4$, and then condense the bosonic loop excitation $m_2^2$ together with
\begin{eqs}
    \varphi^2=m_1^2 e_2^2
\end{eqs}
from the $\ZZ_4\times \ZZ_4$ toric code. The hopping operators for those excitations are
\begin{eqs}
    W_{m_2^2,f}=\prod_{f'}Z_{2,f'}^{2\int \bface'\cup \bface},
\label{eq:Wm2}
\end{eqs}
and
\begin{eqs}
    W_{\varphi^2,f}=\left(\prod_{f'}Z_{1,f'}^{2 \int \bface'\cup \bface}\right)
    X_{2,f}^2~.
\label{eq:Wphi2}
\end{eqs}
By condensing these excitations, we include the operators of Eqs.~\eqref{eq:Wm2} and \eqref{eq:Wphi2} into our stabilizer model. The rest of the stabilizers are given by the original toric code stabilizers that commute with the hopping operators. Applying the criterion Eq.~\eqref{mat criterion} with
\begin{align}
    N=
    \begin{pmatrix}
        2 & 0 \\
        0 & 2 
    \end{pmatrix}
    \quad \mathrm{and}\quad
    L=
    \begin{pmatrix}
        0 & 1 \\
        0 & 0 
    \end{pmatrix},
\end{align}
we obtain them as follows
\begin{eqs}
    G_e^{I}=\prod_{f'}Z_{1,f'}^{\int \delta \be \cup \bface'}, \quad G_e^{II}=\prod_{f'}Z_{2,f'}^{2\int \delta \be \cup \bface'}, \quad G_e^{III}= X_{1,\delta e}\prod_{f'}Z_{2,f'}^{\int \delta \be \cup \bface'}~,
\label{eq:Ge 1 to 3}
\end{eqs}
corresponding respectively to the closed membrane operators for the $m_1$, $m_2^2$, and $e_1m_2$ loop excitations.
Finally, the stabilizer group for our condensed theory is generated by Eqs.~\eqref{eq:Wm2}, \eqref{eq:Wphi2} and \eqref{eq:Ge 1 to 3}.

In fact, the $\ZZ_2\times \ZZ_2$ twisted gauge theory above is in the same phase as the $\ZZ_4$ toric code. We prove this claim by constructing an explicit local unitary transformation between these two models. Without loss of generality, we restrict ourselves to the hypercube, where there is a unique face $f'$ such that
\begin{eqs}
    \int \bface'\cup \bface=1~,
\end{eqs}
which we denote by $f'=*f$. Now we can rewrite Eqs.~\eqref{eq:Wm2} and \eqref{eq:Wphi2} as
\begin{eqs}
    W_{m_2^2,f}=Z_{2,*f}^2~, \qquad W_{\varphi^2,f}=Z_{1,*f}^2 X_{2,f}^2~.
\end{eqs}
We then regroup the lattice by combining the qudit on $*f$ in the first layer with the qudit on $f$ in the second layer into a single site. After this regrouping, both $W_{m_2^2,f}$ and $W_{\varphi^2,f}$ become on-site operators. On each site, the simultaneous $+1$ eigenspace of these two operators is
\begin{eqs}
    \mathrm{span}\Big\{
    \ket{0}\otimes\ket{+},
    \ket{1}\otimes\ket{-},
    \ket{2}\otimes\ket{+},
    \ket{3}\otimes\ket{-}
    \Big\}~,
    \label{basis of condensed site}
\end{eqs}
where the first qudit is located at $*f$ in the first layer and the second qudit is located at $f$ in the second layer. Here
\begin{equation}
    \ket{+}=\frac{1}{\sqrt{2}}(\ket{0}+\ket{2}),
    \qquad
    \ket{-}=\frac{1}{\sqrt{2}}(\ket{0}-\ket{2}) .
\end{equation}
We can treat the terms of Eq.~\eqref{basis of condensed site} as the canonical basis of a $\ZZ_4$-qudit, with Pauli operators
\begin{eqs}
    \mathcal{X}_f=X_{1,*f}Z_{2,f}, \qquad \mathcal{Z}_f=Z_{1,*f}~.
\end{eqs}
Note that when we restrict to the subspace Eq.~\eqref{basis of condensed site}, we automatically have $G_e^{II}=1$, while the rest of Eq.~\eqref{eq:Ge 1 to 3} become
\begin{eqs}
    G_e^I=\prod_{f}\mathcal{Z}_f^{\delta e(f)},\qquad G_e^{III}=\mathcal{X}_{\delta e}~,
\end{eqs}
which are exactly the stabilizers of the standard $(4{+}1)$D $\ZZ_4$ toric code. 
}

%%%%%%%%%%%%%%%%%%%%%%%%%%%%%%%%%%%%%%%%%%%%%%%%%%%%%%%%%%%%%%%%%%%%%%%%%%%%%
\subsection{Pauli Completeness of Condensed Models}
In order to show that the above Pauli stabilizer models are maximally commutative, we should prove that every local Pauli operator that commutes with all stabilizers is contained in the stabilizer group. This is called the completeness condition. We prove completeness only on the hypercubic lattice, where the cup product of cochains is non-degenerate. First, let us review the construction of Pauli operators and express it in a more general setting. 

In our original multi-layer toric code theory, the different $e$ and $m$ loop excitations generate an abelian fusion group
\begin{align}
    \mathcal{G}=\Z_{N_1^2}\times \Z_{N_1^2}\times \cdots\times \Z_{N_s^2}\times \Z_{N_s^2},
\end{align}
on which we have a symplectic structure $\omega: \mathcal{G}\times \mathcal{G}\to \RR/\Z$, representing the mutual braiding of $e$ and $m$ excitations. We denote the subgroup of condensed excitations by $K$. We set $K$ as an isotopic subgroup of $\mathcal{G}$, which denotes the absence of braiding between condensed excitations. Let $K^{\perp}$ be the orthogonal complement of $K$, then we have $K\subseteq K^{\perp}$. The stabilizers of the condensed theory can be rewritten as below:
\begin{eqs}
    \tilde{W}_{f}^{(n,m)}=\prod_{j}X_{j,f}^{n_j}\prod_{k,f'}Z_{k,f'}^{m_k\int \bface'\cup \bface-\delta \bface'\cup_1 \bface}
    \qquad
    (n,m)\in K~
\end{eqs}
and
\begin{eqs}
    G_{e}^{(n,m)}=\prod_{j} X_{j,\delta e}^{n_j}\prod_{k,f}Z_{k,f}^{m_k\int \delta \be \cup \bface}
    \qquad
    (n,m)\in K^{\perp}~.
\end{eqs}
Here, our introduction of the cup-$1$ terms in $W_f^{(n,m)}$ does not change the stabilizer group, but simplifies the argument below. 

Now, suppose that we have a Pauli-type operator of the following form
\begin{eqs}
    P=\prod_{j,f}X_{j,f}^{\beta_j(f)} \prod_{j,f}Z_{j,f}^{f(B_j)}~,
\end{eqs}
where $\beta_j\in C^2(M,\ZZ_{N_j^2})$, and $B_j\in C_2(M,\ZZ_{N_j^2})$. We use $\beta$ to denote the $n$-tuple of $\beta_j$ and $B$ to denote the $n$-tuple of $B_j$. Also, we introduce the dot product
\begin{eqs}
    (v_1,v_2,\cdots,v_s)\cdot (w_1,w_2,\cdots,w_s):=\sum_{j}\frac{1}{N_j^2}v_jw_j~,
\end{eqs}
where the summation is in $\RR/\ZZ$. If $P$ commutes with all $W_f^{(n,m)}$ and $G_e^{(n,m)}$, then we obtain an equivalent set of conditions
\begin{eqs}
    f(B)\cdot n = \left(\int \bface \cup \beta \right)\cdot m
    \qquad
    (n,m)\in K~, \label{fB eq f cup beta}
\end{eqs}
and
\begin{eqs}
    e(\partial B)\cdot n=\left(\int \be \cup \delta \beta \right) \cdot m
    \qquad
    (n,m)\in K^{\perp}~.\label{e parB eq e cup dbeta}
\end{eqs}
Note that both $\beta$ and $B$ take values in the group $\prod_j\ZZ_{N_j^2}$. We can combine them together to obtain a $2$-chain $A\in C_2(M,\mathcal{G})$ that satisfies
\begin{eqs}
    f(A)=\left(\int \bface \cup \beta, f(B)\right)~,
\end{eqs}
where $f$ is an arbitrary face. Now the conditions Eqs.~\eqref{fB eq f cup beta} and \eqref{e parB eq e cup dbeta} impose the following restriction to the coefficient of $A$ and $\partial A$, 
\begin{eqs}
    A\in C_2(M,K^{\perp}), 
    \qquad
    \partial A\in C_1(M,\bar{K})~,
\end{eqs}
in which $\bar{K}=K^{\perp\perp}$ is the orthogonal complement of $K^{\perp}$. Equivalently, it says that $A$ is closed if we take the quotient map from $K^{\perp}$ to $K^{\perp}/\bar{K}$. Given that $P$ is a local operator, we further derive that $A$ is the boundary of some $K^{\perp}/\bar{K}$-valued chain. Therefore, $A$ has the following form
\begin{eqs}
    A=A_W+\partial A_G~,
\end{eqs}
where $A_W\in C_2(M,\bar{K})$ and $A_G\in C_3(M,K^{\perp})$. Finally, in the case where $K=\bar{K}$, we can decompose $P$ into the following product of stabilizers
\begin{eqs}
    P=\prod_{f}W_f^{\alpha_W(f)}\prod_{e}G_e^{\alpha_G(e)}~,
\end{eqs}
where $\alpha_W$ and $\alpha_G$ are dual forms of $A_W$ and $A_G$ defined by the following property,
\begin{eqs}
    \int \bface \cup \alpha_W=f(A_W),
    \qquad
    \int \bt \cup \alpha_G=t(A_G)~
\end{eqs}
for each face $f$ and tetrahedra $t$. Here $\alpha_W$ and $\alpha_G$ exist because we are working on the hypercubic lattice. It remains to prove that $K=\bar{K}$, but this follows from the fact that
\begin{eqs}
    |K|\cdot|K^{\perp}|=|\mathcal{G}|~
\end{eqs}
for any subgroup $K\in \mathcal{G}$. Thus, we have proved the Pauli completeness of our condensed models.

{\change
%%%%%%%%%%%%%%%%%%%%%%%%%%%%%%%%%%%%%%%%%%%%%%%%%%%%%%%%%%%%%%%%%%%%%%%%%%%%%
\subsection{Field-theoretic derivation of the 't Hooft anomaly}

In this subsection, we give a field-theoretic derivation of the 't Hooft anomaly
associated with the flux-loop excitations of the $(4{+}1)$D 2-form gauge
theories discussed above. The resulting anomaly functional encodes the self-
and mutual-statistics of the flux loops, in agreement with the lattice
condensation construction.

We start from the twisted 2-form gauge theory in Eq.~\eqref{action of twist theory}.
The fields $a_j$ couple to the worldsheets of flux loops. To probe these flux
loops, we introduce background $3$-form fields $B_j$ and couple them to $a_j$:
\begin{equation}
    S_B
    =
    S
    -
    2\pi \sum_j \frac{1}{N_j}\int a_j\cup B_j~.
\end{equation}
The background $B_j$ is Poincar\'e dual to the spacetime location of the
$j$-th type of flux-loop excitation. In the absence of $B_j$, the original
theory is invariant under the $2$-form gauge transformation
$b_j\mapsto b_j+\delta\epsilon_j$. In the presence of nontrivial $B_j$, however,
this invariance is spoiled unless we add appropriate counterterms.

To restore gauge invariance, it is convenient to temporarily lift the finite
gauge group $G=\prod_j \ZZ_{N_j}$
to the free Abelian group $G'=\prod_j \ZZ$.
Let
\begin{equation}
    \pi:\;G'=\prod_j \ZZ\longrightarrow G=\prod_j\ZZ_{N_j}
\end{equation}
be the natural quotient map, and denote its kernel by $K$. Thus we have the
short exact sequence
\begin{equation}
    0\longrightarrow K\longrightarrow G'\xrightarrow{\pi}G\longrightarrow 0~.
\end{equation}
After this lift, both $b_j$ and $B_j$ are regarded as integer-valued cochains.
Integrating out the fields $a_j$ imposes the constraint
\begin{eqs}
    N_j\mid \delta b_j-B_j~,
\label{eq:condition of N_j}
\end{eqs}
or equivalently, $\delta b_j-B_j=N_j q_j$
for some integer-valued $3$-cochain $q_j$. Thus $B_j$ measures the failure of
$b_j$ to be closed modulo $N_j$:
\begin{eqs}
    \delta b_j=B_j\mod N_j~.
\end{eqs}

The lift from $G$ to $G'$ introduces two redundancies. The corrected action
should be invariant under both of them. The first is the ordinary local gauge
transformation
\begin{equation}
    (\mathrm{G}1):\qquad
    b_j\longmapsto b_j+\delta\epsilon_j,
    \qquad
    \epsilon_j\in C^1(M,\ZZ)~.
\end{equation}
Here $\epsilon_j$ is an arbitrary integer-valued $1$-cochain. Physically, this
is the usual gauge redundancy of the $2$-form gauge field: shifting $b_j$ by a
coboundary changes only the local representative of the gauge field, since
$\delta^2\epsilon_j=0$.

The second redundancy is the freedom to change the integer lift of the
$\ZZ_{N_j}$-valued field:
\begin{equation}
    (\mathrm{G}2):\qquad
    b_j\longmapsto b_j+N_j c_j,
    \qquad
    c_j\in C^2(M,\ZZ)~.
\end{equation}
Here $c_j$ is an arbitrary integer-valued $2$-cochain. This transformation
expresses the fact that two integer lifts differing by $N_j$ times an arbitrary
cochain represent the same physical $\ZZ_{N_j}$ gauge field. The condition
\eqref{eq:condition of N_j} is preserved under both (G1) and (G2).

Let us first examine the failure of the original twist terms to be invariant
under (G2). Up to total coboundaries and integer-valued terms, the variation is
\begin{align}
    \Delta_{\mathrm{G}2}S_{\mathrm{tw}}
    &=
    2\pi\sum_j \frac{l_j}{N_j^2}
    \int N_j\left(c_j\cup B_j-B_j\cup c_j\right) 
    +
    2\pi\sum_{j<k}\frac{l_{jk}}{N_jN_k}
    \int\left(
        N_j c_j\cup B_k
        -
        N_k B_j\cup c_k
    \right)~.
\end{align}
This variation can be canceled by adding local counterterms. The resulting
corrected action is
\begin{align}
\begin{split}
    S'
    &=
    2\pi\sum_j \frac{1}{N_j}
    \int a_j\cup(\delta b_j-B_j) 
    +
    2\pi\sum_j\frac{l_j}{N_j^2}
    \int\left(
        b_j\cup\delta b_j
        +
        B_j\cup b_j
        -
        b_j\cup B_j
    \right) \\
    &\quad
    +
    2\pi\sum_{j<k}\frac{l_{jk}}{N_jN_k}
    \int\left(
        b_j\cup\delta b_k
        +
        B_j\cup b_k
        -
        b_j\cup B_k
    \right)~.
\label{eq: S' with counter term}
\end{split}
\end{align}
By construction, $S'$ is invariant under the change of integer lift (G2). A
direct check also shows that $S'$ is invariant under the ordinary gauge
transformation (G1), provided that the background fields are closed:
\begin{eqs}
    \delta B_j=0
    \qquad
    \text{for all }j~.
\end{eqs}

Thus the counterterms remove the unphysical dependence on the chosen integer
lift of $b_j$. The corrected theory is a gauge-invariant theory of the physical
$G=\prod_j\ZZ_{N_j}$-valued $2$-form gauge field, rather than a theory of the
auxiliary integer-valued representatives. Nevertheless, the corrected boundary
action still has a nontrivial dependence on the background fields $B_j$. This
remaining background dependence is the 't Hooft anomaly of the flux-loop
excitations.

Equivalently, this anomaly can be canceled by placing the $(4{+}1)$D theory on
the boundary of a $(5{+}1)$D invertible bulk theory. The corresponding bulk
term is a relative cocycle
\begin{eqs}
    T[B]\in H^6(B^3G',B^3K;\RR/\ZZ),
    \qquad
    K=\ker\pi~.
\end{eqs}
The word ``relative'' means that backgrounds valued entirely in $K$ are
physically trivial after projection to $G$. In components, a $K$-valued background has the form $B_j=N_j\widetilde B_j$,
which represents the zero background in $\ZZ_{N_j}$.

We now compute the bulk cocycle explicitly. Let $\mathcal L'$ denote the
integrand of $S'/(2\pi)$. Using $\delta B_j=0$, we have
\begin{align}
    \delta\left(
        b_j\cup\delta b_j
        +
        B_j\cup b_j
        -
        b_j\cup B_j
    \right)
    &=
    \delta b_j\cup\delta b_j
    -
    B_j\cup\delta b_j
    -
    \delta b_j\cup B_j \notag\\
    &=
    (\delta b_j-B_j)\cup(\delta b_j-B_j)
    -
    B_j\cup B_j~.
\end{align}
The first term is divisible by $N_j^2$ because of
Eq.~\eqref{eq:condition of N_j}, and hence contributes an integer to the action.
Similarly,
\begin{align}
    \delta\left(
        b_j\cup\delta b_k
        +
        B_j\cup b_k
        -
        b_j\cup B_k
    \right)
    &=
    \delta b_j\cup\delta b_k
    -
    B_j\cup\delta b_k
    -
    \delta b_j\cup B_k \notag\\
    &=
    (\delta b_j-B_j)\cup(\delta b_k-B_k)
    -
    B_j\cup B_k~,
\end{align}
and the first term is divisible by $N_jN_k$. Therefore, after flipping the sign
and dividing by $2\pi$, the anomaly is
\begin{eqs}
    T[B]
    =
    \sum_j \frac{l_j}{N_j^2}B_j\cup B_j
    +
    \sum_{j<k}\frac{l_{jk}}{N_jN_k}B_j\cup B_k~.
\label{eq: anomaly of twist}
\end{eqs}
If $B_j=N_j\widetilde B_j$ for every $j$, then each term in
Eq.~\eqref{eq: anomaly of twist} is integer-valued and hence vanishes in
$\RR/\ZZ$. Thus $T[B]$ is indeed relative with respect to $K=\ker\pi$.

Although we used the infinite group $G'=\prod_j\ZZ$ to define integer lifts, this choice is only for convenience. One may instead choose a sufficiently large finite group. For example, in the invariant-factor convention
$N_j\mid N_k$ for $j<k$, one may take
\begin{eqs}
    G''=\prod_i \ZZ_{N_iN_s}~.
\end{eqs}
Eq.~\eqref{eq: anomaly of twist}
is then interpreted using integer lifts of the $\ZZ_{N_iN_s}$-valued cocycles
$B_i$. This expression is independent of the chosen lifts because changing
$B_j$ by $N_jN_s$ times an integer cochain changes each term in
$T[B]$ by an integer-valued cochain. Therefore Eq.~\eqref{eq: anomaly of twist}
is a well-defined $\RR/\ZZ$-valued cocycle for this finite choice of $G''$.

Let us now explain the allowed values of the coefficients $l_j$ and $l_{jk}$.
These coefficients are not arbitrary integers. They are defined modulo shifts
that change $T[B]$ only by anomaly-trivial terms, namely terms removable by
local counterterms in the relative boundary theory.

First consider the self term. A shift $l_j\mapsto l_j+2$ changes the anomaly by
\begin{eqs}
    \Delta T[B]
    =
    \frac{2}{N_j^2}B_j\cup B_j~.
\end{eqs}
Since $B_j$ is a degree-$3$ cocycle, the cup-$1$ identity gives
\begin{eqs}
    \delta(B_j\cup_1 B_j)
    =
    -2B_j\cup B_j~.
\end{eqs}
Therefore
\begin{eqs}
    \frac{2}{N_j^2}B_j\cup B_j
    =
    -\delta\left(
        \frac{1}{N_j^2}B_j\cup_1 B_j
    \right)~.
\end{eqs}
Thus $l_j\mapsto l_j+2$ changes the anomaly only by a coboundary.
There is also a relative equivalence $l_j\mapsto l_j+N_j$. Indeed, the
constraint \eqref{eq:condition of N_j} implies that, on the boundary,
\begin{eqs}
    B_j=\delta b_j-N_jq_j
\end{eqs}
for some integer-valued $3$-cochain $q_j$. Hence
\begin{align}
    \frac{1}{N_j}B_j\cup B_j
    &=
    \frac{1}{N_j}(\delta b_j-N_jq_j)\cup B_j 
    \equiv
    \frac{1}{N_j}\delta b_j\cup B_j 
    =
    \delta\left(
        \frac{1}{N_j}b_j\cup B_j
    \right)
    \pmod{1}~.
\end{align}
Thus, the shift $l_j\mapsto l_j+N_j$ is also trivial in the relative anomaly
problem. Combining the two identifications $l_j\sim l_j+2$ and
$l_j\sim l_j+N_j$, we find
\begin{eqs}
    l_j\in \ZZ_{\gcd(N_j,2)}~.
\end{eqs}
Equivalently, if $N_j$ is odd then the self-statistics coefficient is trivial,
whereas if $N_j$ is even then $l_j=0,1$ labels a $\ZZ_2$ choice.

Next, consider the mutual term. A shift $l_{jk}\mapsto l_{jk}+N_j$ changes the
anomaly by
\begin{eqs}
    \Delta T[B]
    =
    \frac{1}{N_k}B_j\cup B_k~.
\end{eqs}
This expression is not necessarily trivial as an absolute cocycle for arbitrary
closed integer-valued fields $B_j$ and $B_k$. It becomes trivial in the relative
boundary problem because $B_k$ is accompanied by the boundary field $b_k$
satisfying
\begin{eqs}
    B_k=\delta b_k-N_kq_k~.
\end{eqs}
Using this relation, we obtain
\begin{align}
    \frac{1}{N_k}B_j\cup B_k
    &=
    \frac{1}{N_k}B_j\cup(\delta b_k-N_kq_k) 
    \equiv
    \frac{1}{N_k}B_j\cup\delta b_k 
    =
    -\delta\left(
        \frac{1}{N_k}B_j\cup b_k
    \right)
    \pmod{1}~,
\end{align}
where we used $\delta B_j=0$ and $\deg B_j=3$. Therefore, the shift
$l_{jk}\mapsto l_{jk}+N_j$ is removed by a local boundary counterterm.
Similarly, using $B_j=\delta b_j-N_jq_j$, we have
\[
    \frac{1}{N_j}B_j\cup B_k
    \equiv
    \delta\left(\frac{1}{N_j}b_j\cup B_k\right)
    \quad \pmod{1},
\]
so the shift $l_{jk}\mapsto l_{jk}+N_k$ is also anomaly-trivial.

The counterterms used above are compatible with the change of integer lift. For
example, under $b_k\mapsto b_k+N_kc_k$,
\begin{eqs}
    \frac{1}{N_k}B_j\cup b_k
    \longmapsto
    \frac{1}{N_k}B_j\cup b_k+B_j\cup c_k~,
\end{eqs}
and the extra term is integer-valued, hence trivial in $\RR/\ZZ$. Therefore, these are legitimate counterterms in the relative theory.

Consequently, $l_{jk}$ is defined modulo the subgroup of $\ZZ$ generated by
$N_j$ and $N_k$. This gives
\begin{eqs}
    l_{jk}\in \ZZ_{\gcd(N_j,N_k)}~.
\end{eqs}

In summary, the independent anomaly coefficients are
\begin{eqs}
    l_j\in \ZZ_{\gcd(N_j,2)},
    \qquad
    l_{jk}\in \ZZ_{\gcd(N_j,N_k)}~.
\end{eqs}
Equivalently,
\begin{eqs}
    l_j=0,1,\ldots,\gcd(N_j,2)-1,
    \qquad
    l_{jk}=0,1,\ldots,\gcd(N_j,N_k)-1~.
\end{eqs}
Here, ``trivial'' should be understood in the relative sense: the corresponding
shift of $T[B]$ can be canceled using the boundary fields $b_j$ satisfying
$B_j=\delta b_j\mod N_j$. It is not the statement that terms such as
$\frac{1}{N_k}B_j\cup B_k$ vanish as absolute cocycles for arbitrary closed
integer-valued $3$-cocycles $B_j$ and $B_k$.

Finally, the anomaly \eqref{eq: anomaly of twist} directly encodes the
statistical properties of the flux loops. The terms
\begin{eqs}
    \frac{l_j}{N_j^2}B_j\cup B_j
\end{eqs}
are self-statistics anomalies of the $j$-th flux loop, detected by the
loop-flipping process. The terms
\begin{eqs}
    \frac{l_{jk}}{N_jN_k}B_j\cup B_k
\end{eqs}
are mutual-statistics anomalies between the $j$-th and $k$-th types of flux loops, detected by mutual braiding of two loops in four spatial dimensions.
}

%%%%%%%%%%%%%%%%%%%%%%%%%%%%%%%%%%%%%%%%%%
\section{Fermionic-loop toric codes in arbitrary dimensions}\label{sec: Fermionic-loop toric codes in arbitrary dimensions}

In higher spatial dimensions, one can construct topological phases that exhibit nontrivial \emph{fermionic loop statistics}. For dimensions $d \ge 5$, this $\mathbb{Z}_2$ fermionic loop statistics is the only possible nontrivial topological invariant of loop excitations~\cite{kobayashi2024generalized, xue2025statisticsabeliantopologicalexcitations}.
This is analogous to the fact that particles in $d\ge 3$ spatial dimensions can only be bosons or fermions: anyonic or mutual braiding statistics no longer exist once the codimension is sufficiently large.

The nontrivial element of this $\mathbb{Z}_2$ statistical phase corresponds to an ’t~Hooft anomaly, which can be described equivalently by the Stiefel–Whitney class $w_3$ or by the stable cohomology operation $\mathrm{Sq}^2 \mathrm{Sq}^1$. In this section, we generalize the construction in Sec.~\ref{sec: (4+1)D fermionic-loop toric code} to obtain Pauli stabilizer realizations of the corresponding topological phases and verify their fermionic loop statistics using the 24-step process. Concretely, we construct stabilizer models that realize the following family of TQFTs: 
\begin{equation}\label{eq: action fermionic loop toric code}
    S = \pi \int_{M_{d+1}} a_2 \cup \delta b_{d-2}
    \;+\;
\pi \int_{M_{d+1}} \mathrm{Sq}^2 \mathrm{Sq}^1 b_{d-2},
\end{equation}
where $a_2 \in C^2(M_{d+1}, \mathbb{Z}_2)$ and $b_{d-2} \in C^{d-2}(M_{d+1}, \mathbb{Z}_2)$ are cochains on a $(d{+}1)$-dimensional spacetime manifold $M_{d+1}$. We refer to this family as the \textbf{fermionic-loop toric code} in $d$ spatial dimensions.

%%%%%%%%%%%%%%%%%%%%%%%%%%%%%%%%%%%%%%%%%%%%%%%%%%%%%%%%%%%%%%%%%%%%%%%%%%%%%
\subsection{$(5{+}1)$D fermionic-loop toric code}

Consider $\ZZ_4$ qudits placed on each tetrahedron. The standard $\ZZ_4$ toric code is defined by
\begin{equation}
    H_{\ZZ_4-\mathrm{TC}} = -\sum_{f\in\mathrm{faces}} A_f \;- \sum_{c_4 \in 4-\mathrm{cells}} B_{c_4} \;+\; h.c.~,
    \label{eq: standard (5+1)D Z4 toric code}
\end{equation}
where
\begin{equation}
    A_f := X_{\delta \bface} = \prod_{t} X_t^{\delta \bface(t)}, \quad B_{c_4}:= Z_{\partial c_4} = \prod_{t} Z_t^{\bt (\partial c_4)}~.
    \label{eq: A and B in standard (5+1)D Z4 toric code}
\end{equation}

In words, $A_f$ is a product of $X_t$ over all tetrahedra $t$ for which $f$ appears in the boundary of $t$, while $B_{c_4}$ is a product of $Z_t$ over all tetrahedra on the boundary of the 4-cell $c_4$. Orientation is taken into account through the definitions of the boundary and coboundary operators. In the standard $(5{+}1)$D toric code, a single $Z_t$ violates the $A_f$ terms on the faces in $\partial t$, creating an $e$-membrane excitation. Conversely, a single $X_t$ violates the $B_{c_4}$ terms on the coboundary $\delta \bt$, producing an $m$-loop excitation on the dual lattice. We can think of the $A_f$ term as a small contractible $m$-loop that detects $e$-membrane, and the $B_{c_4}$ term as a small contractible $e$-membrane that detects the $m$-loops.

Therefore, unlike in Sec.~\ref{sec: (4+1)D fermionic-loop toric code}, we cannot directly condense the $e^2 m^2$ excitation here, since $e$ and $m$ have different dimensions. 

To proceed, we adopt an alternative formulation of the $(5{+}1)$D $\ZZ_4$ toric code. We begin with a symmetry-protected topological (SPT) phase protected by a 2-form $\ZZ_4$ symmetry. We choose the SPT phase characterized by the trivial cocycle $\frac{1}{4} B_3 \cup B_3 \in H^6(K(\ZZ_4,3), \RR/\ZZ)$. Gauging this 2-form symmetry produces a model with the same topological order as the standard $\ZZ_4$ toric code, but whose Hamiltonian takes a slightly modified form:
\begin{equation}
    H'_{\ZZ_4-\mathrm{TC}} = -\sum_{f\in\mathrm{faces}} G_f \;- \sum_{c_4 \in 4-\mathrm{cells}} B_{c_4} \;+\; h.c.~,
\label{eq: modified (5+1)D Z4 toric code}
\end{equation}
where
\begin{equation}
    G_f := X_{\delta \bface} \prod_t Z_t^{\int \delta \bface \cup_1 \bt }~,
\end{equation}
which directly generalizes the operator $G_e$ in Eq.~\eqref{eq: Ge in (4+1)D} for the $(4{+}1)$D case. Since $\frac{1}{4} B_3 \cup B_3$ is a trivial cocycle, one can verify that Eqs.~\eqref{eq: standard (5+1)D Z4 toric code} and \eqref{eq: modified (5+1)D Z4 toric code} differ only by conjugation with a non-Clifford finite-depth quantum circuit (FDQC)~\cite{fidkowski2024qca, sun2025cliffordQCA}.

In the Hamiltonian $H'_{\ZZ_4-\mathrm{TC}}$, the operator
\begin{equation}
    X_t \prod_{t'} Z_{t'}^{\int \bt' \cup_1 \bt}
\end{equation}
commutes with all $G_f$ and violates $B_{c_4}$ on the coboundary $\delta \bt$, creating a loop excitation on the dual lattice.
We denote this excitation as an $m'$ loop. We now condense the ${m'}^2$ loop by adding the term
\begin{equation}
    C_t := X_t^2 \prod_{t'} Z_{t'}^{2\int \bt' \cup_1 \bt}
\end{equation}
to the Hamiltonian and retaining only the stabilizers that commute with $C_t$. The resulting condensed Hamiltonian is
\begin{equation}
    H_{\mathrm{condensed}}
    =
    -\sum_{f\in\mathrm{faces}} G_f
    \;-\sum_{c_4 \in 4\text{-}\mathrm{cells}} B_{c_4}^2
    \;-\sum_{t \in \mathrm{tetrahedra}} C_t
    \;+\; \mathrm{h.c.}~,
    \label{eq: H_condensed fermionic-loop TC in 5+1D}
\end{equation}
which is the natural $(5{+}1)$D analogue of Eq.~\eqref{eq: H_condensed fermionic-loop TC in 4+1D}. This Hamiltonian realizes the fermionic-loop toric code in $(5{+}1)$D.

To facilitate the computations below, we absorb the $B_{c_4}^2$ factors into $C_t$ and define
\begin{equation}
    \tilde{C}_t
    :=
    C_t \prod_{c_4} B_{c_4}^{2 \int \bt \cup_2 \bc_4}
    =
    X_t^2 \prod_{t'} Z_{t'}^{2\int \bt' \cup_1 \bt + \bt \cup_2 \delta \bt'}~,
\end{equation}
so that the Hamiltonian can be rewritten as
\begin{equation}
    H'_{\mathrm{condensed}}
    =
    -\sum_{f\in\mathrm{faces}} G_f
    \;-\sum_{c_4 \in 4\text{-}\mathrm{cells}} B_{c_4}^2
    \;-\sum_{t \in \mathrm{tetrahedra}} \tilde{C}_t
    \;+\; \mathrm{h.c.}~,
    \label{eq: H'_condensed fermionic-loop TC in 5+1D}
\end{equation}
which is directly analogous to Eq.~\eqref{eq: H'_condensed fermionic-loop TC in 4+1D}.

\paragraph*{Loop and membrane excitations.}
This model supports two types of extended excitations. 
First, there is a \emph{charge membrane} excitation created by
\begin{equation}
    V^C_t := Z_t^2~,
\end{equation}
which violates the $G_f$ stabilizers along the boundary $\partial t$. 
Second, there is a \emph{flux loop} excitation created by
\begin{equation}
    \tilde{V}^F_t := 
    X_t \prod_{t'}
    Z_{t'}^{\int \boldsymbol{t}' \cup_1 \boldsymbol{t}
        + \delta \boldsymbol{t}' \cup_2 \bt}~,
    \label{eq: tilde V^F_t for fermionic loop}
\end{equation}
which commutes with all $\tilde{C}_t$ and $G_f$ terms. 
The commutation with $\tilde{C}_t$ follows from the same argument as in Eq.~\eqref{eq: commutation between V_F and C}, with the degrees of cochains and higher cup products shifted by one, and the commutation with $G_f$ is immediate. 
However, $\tilde{V}^F_t$ violates $B_{c_4}^2$ whenever $t\subset \partial c_4$. 
The violated stabilizers form a closed loop on the dual lattice. Equivalently, $\tilde{V}^F_t$ creates a flux loop supported on the coboundary $\delta \bt$.

A direct computation shows that its square is a product of stabilizers:
\begin{equation}
    \bigl(\tilde{V}^F_t\bigr)^2
    =
    X_t^2 \prod_{t'}
    Z_{t'}^{2\int \boldsymbol{t}' \cup_1 \boldsymbol{t}
        + \delta \boldsymbol{t}' \cup_2 \bt}
    =
    \tilde{C}_t \prod_{c_4} B_{c_4}^{2\int \bc_4 \cup_2 \bt + \bt \cup_2 \bc_4}~.
\end{equation}
Therefore, the flux loop obeys a $\ZZ_2$ fusion rule.
Analogous to the $(4{+}1)$D analysis in Sec.~\ref{sec: Fermionic loop statistics in (4+1)D}, we now verify that this flux loop has fermionic loop statistics.
\begin{theorem}
The flux loop created by $\tilde{V}^F_t$ in Eq.~\eqref{eq: tilde V^F_t for fermionic loop} has fermionic loop statistics:
\begin{equation}
    \mu_{24}=-1~,
\end{equation}
where $\mu_{24}$ is defined in Eq.~\eqref{eq: 24-step process in Pauli}.
\end{theorem}
\begin{proof}
This statement is the special case $d=5$ of Theorem~\ref{thm: fermionic flux loop statistics in (d+1)D}, which we will prove in Sec.~\ref{sec: (d+1)D fermionic-loop toric code}.
\end{proof}

%%%%%%%%%%%%%%%%%%%%%%%%%%%%%%%%%%%%%%%%%%%%%%%%%%%%%%%%%%%%%%%%%%%%%%%%%%%%%
\subsection{$(d{+}1)$D fermionic-loop toric code}\label{sec: (d+1)D fermionic-loop toric code}

In this section, we give a direct construction of the fermionic-loop toric code in arbitrary dimensions. Inspecting the $(4{+}1)$D and $(5{+}1)$D Hamiltonians in Eqs.~\eqref{eq: H'_condensed fermionic-loop TC in 4+1D} and~\eqref{eq: H'_condensed fermionic-loop TC in 5+1D}, it is natural to extend them to $(d{+}1)$D as
\begin{equation}
    H_{\mathrm{condensed}} = 
    -\sum_{c_{d-3} \in (d-3)\text{-cells}} G_{c_{d-3}} 
    \;-\; \sum_{c_{d-1} \in (d-1)\text{-cells}} B_{c_{d-1}}^2 
    \;-\; \sum_{c_{d-2} \in (d-2)\text{-cells}} \tilde{C}_{c_{d-2}}
    \;+\; \mathrm{h.c.}~,
    \label{eq: H_condensed fermionic-loop TC in d+1D}
\end{equation}
where a $\ZZ_4$ qudit lives on each $(d-2)$-cell, and
\begin{eqs}
    G_{c_{d-3}} &:= X_{\delta \bc_{d-3}} \prod_{c_{d-2}} Z_{c_{d-2}}^{\int \delta \bc_{d-3} \cup_{d-4} \bc_{d-2}}~, \\
    B_{c_{d-1}} &:= Z_{\partial c_{d-1}}~, \\
    \tilde{C}_{c_{d-2}} &:= X_{c_{d-2}}^2 \prod_{c_{d-2}'} Z_{c_{d-2}'}^{2\int \bc_{d-2}' \cup_{d-4} \bc_{d-2} + \bc_{d-2} \cup_{d-3} \delta \bc'_{d-2}}~.
\end{eqs}

\paragraph*{Loop and $(d{-}3)$-brane excitations.}

This model contains two types of extended excitations. First, there is a boson \emph{charge $(d{-}3)$-brane} excitation created by
\begin{equation}
    V^C_{c_{d-2}}:= Z_{c_{d-2}}^2~,
\end{equation}
which violates the terms $G_{c_{d-3}}$ on the boundary $\partial c_{d-2}$. Second, there is a fermionic \emph{flux loop} excitation created by
\begin{equation}
    \tilde{V}^F_{c_{d-2}}:= X_{c_{d-2}} \prod_{c'_{d-2}}
    Z_{c'_{d-2}}^{\int \bc'_{d-2} \cup_{d-4} \bc_{d-2}
        - (-1)^d \delta \bc'_{d-2} \cup_{d-3} \bc_{d-2}}~,
    \label{eq: tilde V^F_c_d-2 for fermionic loop}
\end{equation}
which commutes with all $\tilde{C}_{c_{d-2}}$ and $G_{c_{d-3}}$ but violates $B_{c_{d-1}}^2$ whenever $c_{d-2}$ lies in the boundary of a $(d-1)$-cell $c_{d-1}$. The violated stabilizers form a closed loop on the dual lattice, so $\tilde{V}^F_{c_{d-2}}$ creates a flux loop supported on the coboundary $\delta \bc_{d-2}$.

A direct computation shows that its square becomes a product of stabilizers:
\begin{equation}
    \left(\tilde{V}^F_{c_{d-2}} \right)^2
    =
    X_{c_{d-2}}^2 \prod_{c'_{d-2}}
    Z_{c'_{d-2}}^{2\int \bc'_{d-2} \cup_{d-4} \bc_{d-2}
        + \delta \bc'_{d-2} \cup_{d-3} \bc_{d-2}}
    =
    \tilde{C}_{c_{d-2}} \prod_{c_{d-1}} B_{c_{d-1}}^{2\int \bc_{d-1} \cup_{d-3} \bc_{d-2} + \bc_{d-2} \cup_{d-3} \bc_{d-1}}~.
\end{equation}
Therefore, the flux loop obeys a $\ZZ_2$ fusion rule.
We now verify that this flux loop has fermionic loop statistics.
\begin{theorem}
    The flux loop created by $\tilde{V}^F_{c_{d-2}}$ in Eq.~\eqref{eq: tilde V^F_c_d-2 for fermionic loop} has fermionic loop statistics:
    \begin{equation}
        \mu_{24}=-1~,
    \end{equation}
    where $\mu_{24}$ is defined in Eq.~\eqref{eq: 24-step process in Pauli}.
    \label{thm: fermionic flux loop statistics in (d+1)D}
\end{theorem}
\begin{proof}
We first promote the definition of $\tilde{V}^F_{c_{d-2}}$ to any $(d-2)$-cochain $\lambda \in C^{d-2}(M_d,\ZZ)$ (where $M_d$ is the spatial manifold) by
\begin{equation}
    \tilde{V}^F_\lambda :=
    X_\lambda \prod_{c_{d-2}'}
    Z_{c_{d-2}'}^{\int c_{d-2}' \cup_{d-4} \lambda
        - (-1)^d \delta c_{d-2}' \cup_{d-3} \lambda}~.
\end{equation}
We begin by computing its commutator:
\begin{eqs}
    [\tilde{V}^F_{\lambda},\tilde{V}^F_{\lambda'}]
    &=
    i^{-\int
        \lambda \cup_{d-4} \lambda'
        - \lambda' \cup_{d-4} \lambda
        -(-1)^d \delta\lambda \cup_{d-3} \lambda'
        +(-1)^d \delta\lambda' \cup_{d-3} \lambda
    } \\
    &= i^{-\int
    (-1)^d  \delta\lambda \cup_{d-3} \lambda'
    + \lambda \cup_{d-3} \delta\lambda'
    -(-1)^d \delta\lambda \cup_{d-3} \lambda'
    +(-1)^d \delta\lambda' \cup_{d-3} \lambda} \\
    &= i^{-\int
        \lambda \cup_{d-3} \delta\lambda'
        + (-1)^d \delta\lambda' \cup_{d-3} \lambda}
    ~=~ i^{(-1)^{d+1}\int
    \delta \lambda \cup_{d-2} \delta \lambda'}~,
\end{eqs}
where, in the last equality, we have used the higher-cup identity
\begin{eqs}
    \delta (A_{d-2} \cup_{d-2} B_{d-1})
    =&~ \delta A_{d-2} \cup_{d-2} B_{d-1} + (-1)^d A_{d-2} \cup_{d-2} \delta B_{d-1} \\
    &- (-1)^d A_{d-2} \cup_{d-3} B_{d-1} - B_{d-1} \cup_{d-3} A_{d-2}~,
\end{eqs}
together with the fact that total coboundaries integrate to zero on a closed manifold.

We now compute the loop statistics using the 24-step process in Fig.~\ref{fig: 24 step process}. For example, for the dual triangle $\langle 012\rangle$, we define the $(d-2)$-cochain $\lambda_{012}=(012)^*$ as its Poincar\'e dual, and write $U_{012}=\tilde{V}^F_{\lambda_{012}}$. Then
\begin{equation}\label{eq: mu24}
    \mu_{24}
    =
    (-1)^{\int \delta\lambda_{012}\cup_{d-2}\delta\lambda_{034}
    + \delta\lambda_{013}\cup_{d-2}\delta\lambda_{024}
    + \delta\lambda_{014}\cup_{d-2}\delta\lambda_{023}}~.    
\end{equation}
As before, we further decompose each piece by
\[
\delta\lambda_{012}
= \delta( 012)^*
= (\partial\, 012)^*
= (01)^* + (12)^* - (02)^*
= \nu_{01}+\nu_{12}-\nu_{02}~,
\]
where we denote $(ij)^*$ by $\nu_{ij}$, which is a $(d-1)$-cochain on the direct lattice.
If $i,j,k,l$ are all distinct, then $\nu_{ij}$ and $\nu_{kl}$ have disjoint support, and hence
$\nu_{ij}\cup_{d-2}\nu_{kl}=0$.
Thus
\begin{eqs}
\mu
=&~
(-1)^{\int
(\nu_{01}+\nu_{02})\cup_{d-2}(\nu_{03}+\nu_{04})
+ (\nu_{01}+\nu_{03})\cup_{d-2}(\nu_{02}+\nu_{04})
+ (\nu_{01}+\nu_{04})\cup_{d-2}(\nu_{02}+\nu_{03})}
\\
=&~
(-1)^{\int
\nu_{02}\cup_{d-2}\nu_{03}
+ \nu_{03}\cup_{d-2}\nu_{02}
+ \nu_{02}\cup_{d-2}\nu_{04}
+ \nu_{04}\cup_{d-2}\nu_{02}
+ \nu_{03}\cup_{d-2}\nu_{04}
+ \nu_{04}\cup_{d-2}\nu_{03}}
\\
=&~
(-1)^{\int
\delta\nu_{02}\cup_{d-1}\nu_{03}
+ \nu_{02}\cup_{d-1}\delta\nu_{03}
+ \delta\nu_{02}\cup_{d-1}\nu_{04}
+ \nu_{02}\cup_{d-1}\delta\nu_{04}
+ \delta\nu_{03}\cup_{d-1}\nu_{04}
+ \nu_{03}\cup_{d-1}\delta\nu_{04}}~,
\end{eqs}
where in the last step we used the Leibniz rule for the $\cup_{d-1}$ product:
\begin{equation}
    \delta (A_{d-1} \cup_{d-1} B_{d-1}) = \delta A_{d-1} \cup_{d-1} B_{d-1} 
    -(-1)^d \Big( A_{d-1} \cup_{d-1} \delta B_{d-1} + A_{d-1} \cup_{d-2} B_{d-1} + B_{d-1} \cup_{d-2} A_{d-1} \Big)~.
\end{equation}
Let $\delta\nu_{ij}= p_i - p_j$ with $d$-cochain $p_i=(i)^*$, and use $\nu_{ij}\cup_{d-1} p_k=0$ for distinct $i,j,k$. Then
\begin{eqs}
\mu
=&~
(-1)^{\int
p_0\cup_{d-1}\nu_{03}
+ \nu_{02}\cup_{d-1} p_0
+ p_0\cup_{d-1}\nu_{04}
+ \nu_{02}\cup_{d-1} p_0
+ p_0\cup_{d-1}\nu_{04}
+ \nu_{03}\cup_{d-1} p_0}
\\
=&~ (-1)^{\int p_0 \cup_{d-1} \nu_{03} + \nu_{03} \cup_{d-1} p_0 }
= ~ (-1)^{\int p_0 \cup_d \delta \nu_{03}}
=~
(-1)^{\int p_0\cup_d p_0 }~.
\end{eqs}
Because $p_0$ is a $d$-cochain, we have $p_0 \cup_d p_0 = p_0 \pmod{2}$. Hence
\begin{equation}
    \mu_{24} = (-1)^{\int (0)^*} = -1~.
\end{equation}
This shows that the flux loop excitation indeed carries fermionic loop statistics in $(d{+}1)$D fermionic-loop toric code.
The computation is exactly parallel to Sec.~\ref{sec: Fermionic loop statistics in (4+1)D}, with all cochains and higher cup products shifted up by certain degrees.
\end{proof}

{\change

\subsection{A $\ZZ_8$-qudit realization}
\label{subsec:z8-qudit-realization}

In this subsection we give another Pauli stabilizer realization of the fermionic-loop toric code in $(d{+}1)$D spacetime dimensions.  The construction starts from the ordinary $\ZZ_8$ toric code, but condenses a decorated square of the $m$-loop rather than the bare $m^2$ loop.  This distinction is important: condensing the bare $m^2$ loop gives the ordinary $\ZZ_2$ toric code, while the decorated condensation below gives the twisted $\ZZ_2$ theory.

Put a $\ZZ_8$ qudit on each $(d-2)$-cell of the spatial manifold $M_d$.  We write
\begin{equation}
    \omega_8=e^{2\pi i/8},
    \qquad
    X^8=Z^8=1,
    \qquad
    ZX=\omega_8 XZ .
\end{equation}
For $b\in C^{d-2}(M_d,\ZZ_8)$, let $\ket b$ denote the computational basis vector.  The standard $\ZZ_8$ toric code is generated by
\begin{equation}
    X_{\delta c_{d-3}},
    \qquad
    Z_{\partial c_{d-1}} .
\end{equation}
Now define a decorated membrane operator, for $\lambda\in C^{d-2}(M_d,\ZZ_8)$, by
\begin{equation}
\label{eq:z8-U-lambda-simple}
    U_\lambda\ket b
    =
    \omega_8^{\int_{M_d}\delta b\cup_{d-2}\delta\lambda}
    \ket{b+\lambda} .
\end{equation}
This is the ordinary $m$-loop membrane operator multiplied by a local $Z$-type decoration supported at the boundary $\delta\lambda$.  A direct computation gives
\begin{equation}
\label{eq:z8-comm-simple}
    [U_{\lambda'},U_\lambda]
    =
    \omega_8^{\int_{M_d}\delta\lambda\cup_{d-2}\delta\lambda'
    -\delta\lambda'\cup_{d-2}\delta\lambda} .
\end{equation}
Using the higher-cup-product identity in Appendix~\ref{app:cochain-conventions}, one obtains
\begin{equation}
    [U_{2\lambda'},U_{2\lambda}]=1 .
\end{equation}
Thus the decorated square $m'^2$ can be condensed.  The condensed stabilizer Hamiltonian is generated by
\begin{equation}
\label{eq:z8-stabilizers-simple}
    X_{\delta c_{d-3}},
    \qquad
    Z_{4\partial c_{d-1}},
    \qquad
    \widehat U_{2\lambda}
    :=
    (-1)^{\int_{M_d}\lambda\cup_{d-4}\lambda+
    \lambda\cup_{d-3}\delta\lambda}
    U_{2\lambda} .
\end{equation}
Here $Z_{4\partial c_{d-1}}=(Z_{\partial c_{d-1}})^4$.  It imposes
\begin{equation}
    \delta b=0\pmod 2,
\end{equation}
rather than the stronger condition $\delta b=0\pmod 8$.  The phase in the definition of $\widehat U_{2\lambda}$ is only a choice of stabilizer eigenvalue convention; it is fixed so that the ground state below is a $+1$ eigenstate of all condensation terms.

Let $N_{d+1}$ be a filling of $M_d$, with $\partial N_{d+1}=M_d$.  Define
\begin{equation}
\label{eq:z8-twisted-ground-state-simple}
    \ket{\Omega_N}
    :=
    \sum_{\substack{b\in C^{d-2}(N_{d+1},\ZZ_8)\\
    \delta b=0\;({\rm mod}\;2)}}
    \omega_8^{\int_{N_{d+1}}\delta b\cup_{d-3}\delta b}
    \ket{b|_{M_d}} .
\end{equation}
The condition $\delta b=0\pmod 2$ immediately implies
\begin{equation}
    Z_{4\partial c_{d-1}}\ket{\Omega_N}=\ket{\Omega_N} .
\end{equation}
Also,
\begin{equation}
    X_{\delta c_{d-3}}\ket{\Omega_N}=\ket{\Omega_N},
\end{equation}
because shifting $b$ by a coboundary does not change $\delta b$ or the phase in Eq.~\eqref{eq:z8-twisted-ground-state-simple}.

Finally, we apply $U_{2\lambda}$ on the state $\ket{\Omega_N}$.  Choose an
arbitrary extension of \(\lambda\) from \(M_d=\partial N_{d+1}\) to
\(N_{d+1}\), and denote the extension again by \(\lambda\).  Acting with
\(U_{2\lambda}\) on the state in Eq.~\eqref{eq:z8-twisted-ground-state-simple} gives
\begin{align}
U_{2\lambda}|\Omega_N\rangle
&=
\sum_{\substack{
b\in C^{d-2}(N_{d+1},\mathbb Z_8)\\
\delta b=0\;(\mathrm{mod}\;2)
}}
\omega_8^{\int_{N_{d+1}}\delta b\cup_{d-3}\delta b}\,
\omega_8^{2\int_{M_d}\delta b\cup_{d-2}\delta\lambda}\,
\bigl|b|_{M_d}+2\lambda\bigr\rangle .
\end{align}
modulo \(8\).
By Stokes' theorem and the higher-cup-product identity,
\begin{align*}
2\int_{M_d}\delta b\cup_{d-2}\delta\lambda
&=
2\int_{N_{d+1}}
\delta\bigl(\delta b\cup_{d-2}\delta\lambda\bigr) \\
&\equiv
2\int_{N_{d+1}}
\left(
\delta b\cup_{d-3}\delta\lambda
+
\delta\lambda\cup_{d-3}\delta b
\right)
\qquad (\mathrm{mod}\;8).
\end{align*}
The possible signs in the higher-cup identity are dropped in this line, because \(\delta b\) is even and the whole expression is multiplied by
\(2\).
Therefore,
\begin{eqs}
\int_{N_{d+1}}\delta b\cup_{d-3}\delta b
+
2\int_{M_d}\delta b\cup_{d-2}\delta\lambda 
&\equiv
\int_{N_{d+1}}
\left[
\delta b\cup_{d-3}\delta b
+
2\,\delta b\cup_{d-3}\delta\lambda
+
2\,\delta\lambda\cup_{d-3}\delta b
\right] \\
&\equiv
\int_{N_{d+1}}
\delta(b+2\lambda)\cup_{d-3}\delta(b+2\lambda)
-
4\int_{N_{d+1}}\delta\lambda\cup_{d-3}\delta\lambda
\qquad (\mathrm{mod}\;8).
\nonumber
\end{eqs}
Since \(\omega_8^{-4x}=(-1)^x\), this gives
\begin{align}
U_{2\lambda}|\Omega_N\rangle
&=
(-1)^{\int_{N_{d+1}}\delta\lambda\cup_{d-3}\delta\lambda}
\sum_{\substack{
b\in C^{d-2}(N_{d+1},\mathbb Z_8)\\
\delta b=0\;(\mathrm{mod}\;2)
}}
\omega_8^{
\int_{N_{d+1}}
\delta(b+2\lambda)\cup_{d-3}\delta(b+2\lambda)
}
\bigl|(b+2\lambda)|_{M_d}\bigr\rangle .
\end{align}
Now set \(b'=b+2\lambda\).  Since
\(\delta b'=\delta b+2\delta\lambda\), the condition \(\delta b=0\;(\mathrm{mod}\;2)\) is equivalent to \(\delta b'=0\;(\mathrm{mod}\;2)\).
The sum is unchanged by this change of variables, and hence
\[
U_{2\lambda}|\Omega_N\rangle
=
(-1)^{\int_{N_{d+1}}\delta\lambda\cup_{d-3}\delta\lambda}
|\Omega_N\rangle .
\]
It remains to rewrite the bulk sign as a boundary term.  Reducing modulo \(2\),
the higher-cup-product identity gives
\begin{eqs}
\delta\bigl(\lambda\cup_{d-3}\delta\lambda\bigr)
&=
\delta\lambda\cup_{d-3}\delta\lambda
+
\lambda\cup_{d-4}\delta\lambda
+
\delta\lambda\cup_{d-4}\lambda, \\
\delta\bigl(\lambda\cup_{d-4}\lambda\bigr)
&=
\delta\lambda\cup_{d-4}\lambda
+
\lambda\cup_{d-4}\delta\lambda .
\end{eqs}
Adding these two equations, the mixed terms cancel modulo \(2\), so
\[
\delta\lambda\cup_{d-3}\delta\lambda
=
\delta\left(
\lambda\cup_{d-3}\delta\lambda
+
\lambda\cup_{d-4}\lambda
\right)
\qquad (\mathrm{mod}\;2).
\]
Therefore, using \(\partial N_{d+1}=M_d\),
\[
\int_{N_{d+1}}\delta\lambda\cup_{d-3}\delta\lambda
=
\int_{M_d}
\left(
\lambda\cup_{d-4}\lambda
+
\lambda\cup_{d-3}\delta\lambda
\right)
\qquad (\mathrm{mod}\;2).
\]
Combining the above, we obtain
\begin{equation}
\begin{aligned}
U_{2\lambda}|\Omega_N\rangle
&=
(-1)^{\int_{N_{d+1}}\delta\lambda\cup_{d-3}\delta\lambda}
|\Omega_N\rangle 
=
(-1)^{
\int_{M_d}
\lambda\cup_{d-4}\lambda
+
\lambda\cup_{d-3}\delta\lambda
}
|\Omega_N\rangle .
\end{aligned}
\label{eq:z8-U2-ground-simple}
\end{equation}
Hence,
\begin{equation}
    \widehat U_{2\lambda}\ket{\Omega_N}=\ket{\Omega_N} .
\end{equation}
Therefore, Eq.~\eqref{eq:z8-twisted-ground-state-simple} is a ground state of the stabilizer Hamiltonian
\begin{equation}
\label{eq:z8-hamiltonian-simple}
    H_{\ZZ_8}
    =-
    \sum_{c_{d-3}}X_{\delta c_{d-3}}
    -\sum_{c_{d-1}}Z_{4\partial c_{d-1}}
    -\sum_\lambda \widehat U_{2\lambda}
    +h.c. .
\end{equation}
For a local Hamiltonian, it is sufficient to take $\lambda$ over a local generating set of $(d-2)$-cochains.

We now explain why the ground state is twisted.  On the support of the sum, let
\begin{equation}
    \bar b:=b\pmod 2\in Z^{d-2}(N_{d+1},\ZZ_2),
    \qquad
    q:=\frac{1}{2}\delta b\pmod 2\in Z^{d-1}(N_{d+1},\ZZ_2).
\end{equation}
The cochain $q$ is the Bockstein of $\bar b$, hence
\begin{equation}
    q=\Sq^1\bar b .
\end{equation}
Therefore the phase in Eq.~\eqref{eq:z8-twisted-ground-state-simple} reduces to
\begin{equation}
\label{eq:z8-phase-reduces-to-sq2sq1}
    \frac{1}{8}\delta b\cup_{d-3}\delta b
    =
    \frac{1}{2}q\cup_{d-3}q
    =
    \frac{1}{2}\Sq^2\Sq^1\bar b
    \qquad (\mathrm{mod}\;1).
\end{equation}
Thus the state in Eq.~\eqref{eq:z8-twisted-ground-state-simple} is the ground-state wave function of the nontrivial twisted $\ZZ_2$ $(d-2)$-form gauge theory appearing in Eq.~\eqref{eq: action fermionic loop toric code}.  By contrast, condensing the bare $X_{2\lambda}$ in the ordinary $\ZZ_8$ toric code would give the same constraint $\delta b=0\pmod 2$ but without the phase in Eq.~\eqref{eq:z8-phase-reduces-to-sq2sq1}; that is the ordinary bosonic $\ZZ_2$ toric code.  Hence the decoration of the condensed loop changes the resulting ground state from the untwisted theory to the twisted fermionic-loop theory.

For completeness, we also record the resulting loop statistics.  Since $U_{2\lambda}$ is condensed, the excitation created by $U_\lambda$ has $\ZZ_2$ fusion in the ground-state sector.  Squaring Eq.~\eqref{eq:z8-comm-simple} gives
\begin{equation}
    [U_{\lambda'},U_\lambda]^2
    =
    (-1)^{\int_{M_d}\delta\lambda\cup_{d-2}\delta\lambda'} .
\end{equation}
This is the same commutator entering the $24$-step calculation in Eq.~\eqref{eq: mu24}; the computation in Sec.~\ref{sec: (d+1)D fermionic-loop toric code} therefore gives $\mu_{24}=-1$.}

\section{Generalized statistics in Pauli stabilizer models}
\label{sec: Statistics in Pauli Stabilizer Models}

In previous sections, we detected fermionic-loop statistics using Eq.~\eqref{eq: 24-step process in Pauli}, which follows from the full $24$-step process in Eq.~\eqref{eq: 24-step process} once we assume a Pauli setting, i.e., that any two relevant operators commute up to a $U(1)$ phase. A general microscopic theory of statistics for invertible excitations in arbitrary dimensions was developed in Refs.~\cite{kobayashi2024generalized,xue2025statisticsabeliantopologicalexcitations}, but the corresponding statistical processes typically become increasingly involved in higher dimensions. In Pauli stabilizer models, these processes simplify substantially. This motivates a Pauli-adapted formulation of generalized statistics, parallel to the general non-Pauli framework. The definitions below are modeled on the axioms in Ref.~\cite{xue2025statisticsabeliantopologicalexcitations}, but are tailored to Pauli stabilizer Hamiltonians.

\subsection{Definitions and general framework}

We specify a type of excitation by a set of configuration states $\{\ket{a}\mid a\in A\}$ together with a family of excitation operators $\{U(s)\mid s\in S\}$ satisfying suitable axioms. The choice of the label sets $A$ and $S$ encodes the dimensionality of the excitation and its fusion group, while the statistical information is extracted from commutator processes built from the operators $U(s)$ and their action on the states.

Throughout, we take the spatial manifold to be a $d$-sphere $S^d$ equipped with a fixed triangulation $X$, most commonly $X=\partial\Delta^{d+1}$. For a $p$-dimensional excitation with fusion group $G\simeq \ZZ$ or $\ZZ_N$, we adopt the canonical choice
\begin{equation}
    S := X_{p+1},
    \qquad
    A := B_p(X,G),
\end{equation}
where $X_{p+1}$ is the set of $(p{+}1)$-simplices of $X$ and $B_p(X,G)$ is the group of $p$-boundaries with coefficients in $G$. The intended picture is that an operator $U(s)$ creates a localized excitation whose configuration is determined by the boundary $\partial s\in A$, and that $U(s)$ is supported near the simplex $s$.

These considerations lead to the following definition.

\begin{definition}\label{defexcitationModel}
	An \textit{excitation pattern} \(m\)  consists of the following data $m=(A,S,\partial,\supp)$:
	\begin{enumerate}
		\item An Abelian group \(A\), referred to as the configuration group.
		\item A finite set \(S\), referred to as formal excitation operators, and a map \(\partial: S \longrightarrow A\) such that $\{\partial s|s\in S\}$ generates \(A\).
		\item A topological space $M$ and a closed subspace \(\operatorname{supp}(s) \subset M\) for every $s\in S$.
	\end{enumerate}

\end{definition}

The previous construction can be easily generalized to any Abelian fusion group $G$. Let $G\simeq \ZZ_{N_1}\oplus\cdots \ZZ_{N_k}$ (when $N_k=0$, we say $\ZZ_{N_k}=\ZZ$); we will define the corresponding excitation pattern $m_p(X,G)$ as follows:
	\begin{enumerate}
		\item \(A = B_p(X, G)\).
		\item \(S = G_0\times X_{p+1}\), where \(X_{p+1}\) is the set of all \((p+1)\)-simplices of $X$, and \(G_0\) is a subset of $G$, containing a generator for every $\ZZ_{N_i}$ component. Note that $S$ is a basis of $C_{p+1}(X,G)$.
		\item For $g\in G_0$ and $\alpha\in X_{p+1}$, \(\operatorname{supp}(g\alpha)\) gives the simplex $\alpha$, while \(\partial (g\alpha)\) gives its homological boundary.
	\end{enumerate}
In this work, we focus on excitation patterns of the form $m_p(\partial \Delta^{d+1},G)$.
    
This definition of excitation pattern is the same as that in Ref.~\cite{xue2025statisticsabeliantopologicalexcitations}: an excitation pattern is only a collection of geometric data without any physical kinematics, and Pauli operators are not meaningful at this level. At the kinematics level, the data of an excitation pattern should be realized as a Hilbert space $\mathcal{H}$, a collection of states $\{|a\rangle|a\in A\}$ and operators $\{U(s)|s\in S\}$, where Pauli realizations are distinguished from non-Pauli realizations. In the non-Pauli case, the classification of statistics is described by an Abelian group $T^*(m)$, and then it is proved that $T^*(m_p(\partial\Delta^{d+1},G))\simeq H^{d+2}(B^{d-p}G,\RR/\ZZ)$ \cite{xue2025statisticsabeliantopologicalexcitations}. We aim to construct an Abelian group $T_{\mathrm{P}}^*(m)$ as the classification of statistics of Pauli realizations and also its Pontryagin dual $T_{\mathrm{P}}(m)=\hom(T^*_{\mathrm{P}}(m),\RR/\ZZ)$ describing how we detect statistics by computing certain invariants. In previous sections, we have constructed lattice models representing various elements in $T^*_{\mathrm{P}}(m)$, and Eq.~\eqref{eq: 24-step process in Pauli} represents the generator in $T_{\mathrm{P}}(m)$ for $m=m_1(\partial\Delta^{4+1},\ZZ_2)$.

We require that in a Pauli realization, $\{U(s)|s\in S\}$ are all Pauli operators, so they should commute up to a phase factor. This condition makes Pauli models especially easy to solve, and it also transforms Eq.~\eqref{eq: 24-step process} into Eq.~\eqref{eq: 24-step process in Pauli}.
\begin{itemize}
    \item \textbf{Pauli assumption:} for any $s,s'\in S$, we have 
    \begin{equation}\label{axiomPauli}
        [U(s'),U(s)]=e^{i\varphi(s',s)},
    \end{equation}
    where $[U(s'),U(s)]=U'(s)^{-1}U(s)^{-1}U'(s)U(s)$ and $\varphi(s',s)\in \RR/2\pi\ZZ$ is a phase factor.
\end{itemize}

It is obvious that 
\begin{equation}\label{eqAlternating}
    \left\{
    	\begin{aligned}
		\varphi(s,s)&=0~,\\
        \varphi(s,s')&=-\varphi(s',s)~.
	\end{aligned}
	\right.
\end{equation}
Also, we have $[U(s_1)U(s_2),U(s)]=e^{i\varphi(s_1,s)+i\varphi(s_2,s)}$. We shall generalize the notation $\varphi$ to a product of operators, such as
$[U(s_1s_2),U(s_3s_4)]=e^{i\varphi(s_1s_2,s_3s_4)}$; then, we have
\begin{equation}
    \varphi(s_1s_2,s)=\varphi(s_1,s)+\varphi(s_2,s)~,
\end{equation}
and this linearity is similar for the second variable. In general, we extend $\varphi$ to a bilinear and alternating function $\ZZ[S]\times \ZZ[S]\rightarrow \RR/2\pi\ZZ$, where $\ZZ[S]$ is the free Abelian group generated by $S$, whose elements are like $s_1^{n_1}\cdots s_k^{n_k}$, where $s_i\in S$ and $n_i\in \ZZ$. We extend the map $\partial: S\rightarrow A$ to $\ZZ[S]\rightarrow A$ by 
\begin{equation}
    \partial(s_1^{n_1}\cdots s_k^{n_k})=n_1\partial s_1+\cdots +n_k\partial s_k.
\end{equation}
In a many-body Hilbert space, operators with disjoint supports must commute.
Note that we do not fully characterize the tensor-product structure, but only interpret it as the following axiom.
\begin{itemize}
    \item \textbf{Locality axiom:} If $s,s'\in S$ satisfy $\supp(s)\cap\supp(s')=\emptyset$, then $U(s)$ and $U(s')$ commute. In other words,
    \begin{equation}\label{axiomLocality}
        \supp(s)\cap \supp(s')=\emptyset\implies \varphi(s',s)=0.
    \end{equation}
\end{itemize}

The configuration axiom says that for any $s\in S$ and $a\in A$, we should have 
\begin{equation}\label{axiomConfigurationOld}
    U(s)|a\rangle\propto |a+\partial s\rangle.
\end{equation}
For Pauli stabilizer models, this can be expressed in a more convenient form. If $V_1=U(s_n\cdots s_1)$\footnote{This notation means $U(s_n)\cdots U(s_1)$. Different ordering of $s_1,\cdots,s_n$ may differ by a phase factor, which does not influence the commutator.} and  $V_2=U(s_m'\cdots s_1')$ are two operators satisfying $ \partial (s_n\cdots s_1)=\partial(s_m'\cdots s_1')=0$, then any $|a\rangle$ is the common eigenstate of $V_1$ and $V_2$, so $[V_1, V_2]=1$. Thus, we have
\begin{itemize}
    \item \textbf{Configuration axiom:} 
    \begin{equation}\label{axiomConfiguration}
        h_1,h_2\in \ZZ[S],\partial h_1=\partial h_2=0\implies\varphi(h_1,h_2)=0.
    \end{equation}
\end{itemize}

We will denote the kernel of $\partial:\ZZ[S]\rightarrow A$ by $K$, referred to as the condensed subgroup. Physically, for any $k\in K$, $U(k)|a\rangle\propto|a\rangle$; this means that $U(k)$ does not create any excitation but only introduces a phase factor (which may depend on $a$).

Choosing a subset $K_0\subset K$ generating $K$, Eq.~\eqref{axiomConfiguration} is equivalent to that \begin{equation}
        \varphi(k_1,k_2)=0,\forall k_1,k_2\in K_0.
\end{equation}
For the excitation pattern $m_p(X,\ZZ)$, $\partial$ is the boundary map $\partial:C_{p+1}(X,\ZZ)\rightarrow B_{p}(X,\ZZ)$, so $K=Z_{p+1}(X,\ZZ)$. Our language aligns with the intuition of string-net models: closed strings are condensed, while open strings create excitations at their boundaries. For $m_p(X,\ZZ_N)$, the kernel of $\partial:C_{p+1}(X,\ZZ)\rightarrow B_{p}(X,\ZZ_N)$ is $Z_{p+1}(X,\ZZ)+N\cdot C_{p+1}(X,\ZZ)$. Intuitively, besides closed strings, we also need to condense all anyons that are zero in the fusion group.

Eq.~\eqref{axiomPauli}, \eqref{axiomLocality}, and \eqref{axiomConfiguration} characterize Pauli realizations. Note that although the configuration axiom is deduced from Eq.~\eqref{axiomConfigurationOld}, all of the three axioms do not have direct references to configuration states $|a\rangle$. Actually, we can reconstruct them as the common eigenstates of all $U(k)$ for $k$ condensed. This does not determine $|a\rangle$, especially when these eigenstates have degeneracy, but in our theory of statistics, we only care about the information extracted from the function $\varphi$ in Eq.~\eqref{axiomPauli}. Thus, we define Pauli realization of an excitation pattern as follows.
\begin{definition}
    Let $m=(A,S,\partial,\supp)$ be an excitation pattern. A Pauli realization of $m$ is a Hilbert space $\mathcal{H}$ and a collection of unitary operators $\{U(s)|s\in S\}$ satisfying Eq.~\eqref{axiomPauli}, \eqref{axiomLocality}, and \eqref{axiomConfiguration}. We say two Pauli realizations are equal if and only if their corresponding functions $\varphi:\ZZ[S]\times \ZZ[S]\rightarrow A$ are equal. 
\end{definition}
We denote all Pauli realizations of $m$ by $R(m)$. Equivalently, $R(m)$ consists of functions $\varphi:S\times S\rightarrow \RR/2\pi\ZZ$ that is a solution of Eq.~\eqref{axiomPauli}, \eqref{axiomLocality}, and \eqref{axiomConfiguration}. By definition $R(m)$ is a sub-Abelian group of $(\RR/2\pi\ZZ)[S\times S]$.

$R(m)$ is, in general, very large, and not all its information is essential. Later, we will quotient it into several equivalence classes, each corresponding to a specific type of statistics. Then, to measure these statistics, we construct some particular linear maps $e: R(m)\rightarrow \RR/2\pi\ZZ$ which have different values only in different equivalence classes of $R(m)$. We will call these maps \textit{statistical expressions}. For linear maps that do not have this property, we will generally call them \textit{expressions}.

\begin{definition}
    For any excitation pattern $m=(A,S,\partial,\supp)$, its Pauli expression group $E(m)=\ZZ[S\times S]$ is the free Abelian group generated by $\mathsf{c}(s',s)$ for $s',s\in S$, whose elements are called expressions. 
\end{definition}
$\mathsf{c}(s',s)$ ($\mathsf{c}$ stands for commutator) is not a phase factor but only a basis element in $E(m)$, while it can be viewed as a linear map:
\begin{equation}
    \begin{aligned}
        \mathsf{c}(s',s): R(m)&\rightarrow \RR/2\pi\ZZ;\\
        \varphi&\mapsto \varphi(s',s).
    \end{aligned}
\end{equation}
From this view, the concept of expression and realization are dual in some sense. Actually, $E(m)\simeq \ZZ[S\times S]$ is the Pontryagin dual of $(\RR/2\pi\ZZ)[S\times S]$. Because  $R(m)$ is a subgroup: $R(m)\subset (\RR/2\pi\ZZ)[S\times S]$, its Pontryagin dual is a quotient group $E(m)/E_\id(m)$, where $E_\id(m)$ is the annihilator of $R(m)$:
\begin{equation}
    E_\id(m)=\{e\in E(m)|e(\varphi)=0,\forall \varphi\in R(m)\}.
\end{equation}
We will use the notation $e\sim0$ to denote that $e\in E_\id(m)$. Thus, the definition of $E_\id$ means:
\begin{equation}\label{eqAlternating}
    \left\{
    	\begin{aligned}
		\mathsf{c}(s,s)&\sim0;\\
        \mathsf{c}(s,s')&\sim-\mathsf{c}(s',s);\\
        \mathsf{c}(s',s)&\sim 0 \text{ if }\supp(s)\cap \supp(s')=\emptyset;\\
        \mathsf{c}(k',k)&\sim 0 \text{ if }\partial k=\partial k'=0.
	\end{aligned}
	\right.
\end{equation}
Here, we have bilinearly extended $\mathsf{c}$ to $\ZZ[S]$ for both variables.

Next, we define statistical expressions, which form a subgroup
$E_{\inv}(m)\supset E_{\id}(m)$. In the non-Pauli setting, several closely related criteria for selecting $E_{\inv}$ have been proposed in the literature. However, some of them are not suitable for Pauli realizations. In particular, the criterion based on invariance under local perturbations (see, e.g., Ref.~\cite{FHH21}, Definition~III.2 of Ref.~\cite{kobayashi2024generalized}, and Theorem~III.3 of Ref.~\cite{xue2025statisticsabeliantopologicalexcitations}) does not work in the Pauli case. Likewise, the criterion that $E_{\inv}(m)/E_{\id}(m)$ is the torsion subgroup of $E(m)/E_{\id}(m)$ fails here: for fusion group $G=\ZZ_n$ one always has $n^2\,\mathsf{c}(s',s)\in E_{\id}(m)$. In our Pauli framework, the definition that remains effective is the \emph{localization} criterion (Definition~III.7 of Ref.~\cite{xue2025statisticsabeliantopologicalexcitations}), which we adopt below.

For any excitation pattern $m=(A,S,\partial,\supp)$ on a topological space $M$ and any point $x\in M$, we define a new excitation pattern
\[
    m|_x=(A,S,\partial,\supp')
\]
by modifying only the support map:
\[
    \supp'(s):=\supp(s)\cap\{x\}.
\]
Equivalently, $\supp'(s)$ is nonempty precisely when $x\in\supp(s)$, and $\supp'(s_1)\cap\supp'(s_2)=\emptyset$ unless both $\supp(s_1)$ and $\supp(s_2)$ contain $x$.
We call $m|_x$ the \emph{localization} of $m$ at $x$. Intuitively, a realization of $m|_x$ is a realization of $m$ in which every operator $U(s)$ is additionally supported at $x$; in particular, if $x\notin\supp(s)$ then $U(s)$ must be a pure phase.

Formally, $R(m|_x)\subset R(m)$, and $\varphi\in R(m|_x)$ must satisfy the additional constraint
\begin{equation}
    \varphi(s',s)=0
    \qquad
    \text{if } x\notin\supp(s') \text{ or } x\notin\supp(s)~.
\end{equation}
Dually, $E_{\id}(m|_x)\supset E_{\id}(m)$: it contains the additional generators $\mathsf{c}(s',s)$ whenever $x\notin\supp(s')$ or $x\notin\supp(s)$.

Since statistics is inherently nonlocal, any realization of $m|_x$ should be regarded as having trivial statistics. This motivates the following definition.
\begin{definition}
The \textbf{Pauli statistics} of $m$ is
\begin{equation}
    T_{\mathrm P}^*(m)
    :=
    R(m)\Big/\sum_{x\in M} R(m|_x),
\end{equation}
where $\sum_{x\in M} R(m|_x)$ denotes the subgroup generated by the subgroups $R(m|_x)\subset R(m)$. Its dual group is
\begin{equation}
    T_{\mathrm P}(m):=E_{\inv}(m)/E_{\id}(m),
\end{equation}
where
\begin{equation}
    E_{\inv}(m)
    \simeq
    \bigcap_{x\in M} E_{\id}(m|_x)
\end{equation}
is called the group of \emph{statistical expressions}.
\end{definition}

In Appendix~\ref{app: Statistics in non-Pauli models}, we review the definitions of statistics $T(m)$ and $T^*(m)$ in the general (non-Pauli) setting. We then construct a canonical map
\[
    T_{\mathrm P}^*(m)\longrightarrow T^*(m)
\]
and its dual map $T(m)\to T_{\mathrm P}(m)$. The map $T_{\mathrm P}^*(m)\to T^*(m)$ is not surjective in general, which indicates that some statistical invariants cannot be realized within Pauli stabilizer models. If this map were not injective, it would suggest the existence of Pauli realizations that can be trivialized only by non-Pauli transformations. At present, we do not know whether $T_{\mathrm P}^*(m)\to T^*(m)$ is injective: we have neither a proof nor a counterexample.

For small values of $p$, $d$, and $G$, we compute $T_{\mathrm P}\!\bigl(m_p(\partial\Delta^{d+1},G)\bigr)$ directly from the definition using a computer, following an approach similar to Refs.~\cite{kobayashi2024generalized,xue2025statisticsabeliantopologicalexcitations}. Below, we summarize several representative outcomes, together with conjectural patterns for general $G$. In all cases, we find
\[
    T_{\mathrm P}^*\!\bigl(m_p(\partial\Delta^{d+1},G)\bigr)\subset H^{d+2}(B^{d-p}G,\RR/\ZZ).
\]
\begin{itemize}
    \item \textbf{Particles} ($p=0$). We always find $T\simeq T_{\mathrm P}$; that is, all particle statistics can be realized by Pauli stabilizer models.
    \item \textbf{Loops} ($p=1$). We find $T_{\mathrm P}=0$ for $d=2$ and $T_{\mathrm P}\simeq T$ for $d\ge 3$. Thus, Pauli stabilizer models realize all loop statistics in spatial dimension $d\ge3$, but not in $d=2$.
    \item \textbf{Membranes} ($p=2$). We find $T_{\mathrm P}=0$ for $d\le 4$. For $d\ge 5$, Pauli stabilizer models do not realize the Pontryagin statistics of Ref.~\cite{feng2025anyonic}, while the remaining part is realizable. In particular,
    \[
        T_{\mathrm P}\!\bigl(m_2(\partial\Delta^{6+1},\ZZ_2)\bigr)\simeq \ZZ_4,
    \]
    corresponding to a ``semionic-membrane'' statistic in $d=6$ spatial dimensions, analogous to semions in $d=2$.
\end{itemize}

%%%%%%%%%%%%%%%%%%%%%%%%%%%%%%%%%%%%%%%%%%%%%%%%%%%%%%%%%%%%%%%%%%%%%%
\subsection{Examples of statistical expressions}

As an example,  let $m=m_1(\partial\Delta^{3+1},\ZZ_2)$. We have already shown that the expression detecting fermionic loops
\begin{equation}\label{eqStatisticalExpressionFermionicLoop}
    e=2\mathsf{c}(s_{012},s_{034})+2\mathsf{c}(s_{013},s_{024})+2\mathsf{c}(s_{014},s_{023})\in E_{\operatorname{inv}}(m)
\end{equation}
 is nontrivial, i.e., $e\notin E_\id(m)$, since it gives nonzero values for fermionic-loop toric-code model. Next, we show that $e\in E_\inv(m)$. We have $e\in E_{\operatorname{id}}(m|_x)$ for $x=1,2,3,4$ trivially, so it remains to check $e\in E_\id(m|_0)$. 

 For notation simplicity, we denote a $p$-simplex by writing down its vertices directly; we also write $\ZZ[S]$ additively as chains. Then, we have

 \begin{equation}\label{eqStatisticalExpressionFermionicLoop}
    \begin{aligned}
        e&=2\mathsf{c}(012,\,034)+2\mathsf{c}(013,\,024)+2\mathsf{c}(014,\,023)\\&=\mathsf{c}(2\cdot012,\,034)+\mathsf{c}(2\cdot013,\,024)+\mathsf{c}(2\cdot014,\,023)~,
    \end{aligned}
\end{equation}
where ``$2\cdot0ij$'' means multiplying a chain $0ij$ by two. Because $2\cdot 012$ and $\partial(0134)$ are both in $K$, we have $\mathsf{c}(2\cdot 012,\,\partial(0134))\sim 0$. Thus, we have

\begin{equation}
    e\sim \mathsf{c}(2\cdot012,\,034+\partial(0134))+\mathsf{c}(2\cdot013,\,024+\partial(0124))+\mathsf{c}(2\cdot014,\,023+\partial(0123))~.
\end{equation}
Expanding it and using $4c(s,s')\sim 0$, we have (in the original notation)
\begin{equation}
    e\sim2\mathsf{c}(s_{012},s_{134})+2\mathsf{c}(s_{013},s_{124})+2\mathsf{c}(s_{014},s_{123})~.
\end{equation}
This expression looks similar to $e$, with $0$ and $1$ exchanged. Then, it is obvious that  $e\in E_\id(m|_0)$, so $e\in E_\inv(m)$.

Next, we introduce some notations. For a finite subset $\tau \in \ZZ$, we denote the sum of all elements of $\tau$ by $\sum\tau$: 
\begin{equation}
    \sum\tau:=\sum_{i\in\tau}i.
\end{equation}
For any $k\in \ZZ$, we use $\#(\tau<k)$ to denote the number of elements less than $k$. For example, when $\tau=\{2,3,5,6\}$, we have $\sum\tau=16$ and $\#(\tau<5)=2$.

We sometimes write $\tau$ as $\tau_p$ to stress that $\tau$ contains $p$ elements, in which case $\tau$ represents a $(p-1)$-simplex. For $\tau,\sigma\subset \ZZ$ and $\tau\cap\sigma=\emptyset$, we use the notation $\tau\sqcup\sigma$ to represent their disjoint union. 

Now we generalize Eq.~\eqref{eqStatisticalExpressionFermionicLoop} to $p>1$ cases. For $p$-dimensional excitations, the spatial dimension is $d=2p+1$, and the set of vertices is $\{0,\cdots,2p+2\}$. Similar to Eq.~\eqref{eqStatisticalExpressionFermionicLoop}, for every term $\mathsf{c}(s',s)$, we have $s'\cap s=\{0\}$ and $1\in s$. Remaining vertices  are divided as $\{2,\cdots,2p+2\}=\tau_p\sqcup\sigma_{p+1}$, and we will write the corresponding term as $\mathsf{c}(\overline{01\tau},\overline{0\sigma})$.

\begin{theorem}
    Let $m=m_{p}(\partial\Delta^{2p+2},\ZZ_{2N})$. Then,
    \begin{equation}
        e=2N^2\sum_{\tau_{p}\sqcup\sigma_{p+1}=\{2,\cdots,2p+2\}}\mathsf{c}(\overline{01\tau},\overline{0\sigma})~.
    \end{equation}
    represents an order-$2$ element in $T_{\mathrm P}(m)$.
    \label{thm: order 2 statistics for Z_2N}
\end{theorem}
\begin{proof}
    To prove $e\in E_\inv(m)$, it is enough to prove $e\in E_\id(m|_0)$.

    During the proof, we will write $\ZZ[S]\simeq C_{p+1}(\partial\Delta^{2p+2},\ZZ)$ additively. We have
    \begin{equation}
        \begin{aligned}
            e&= \sum_{\tau_{p}\sqcup\sigma_{p+1}=\{2,\cdots,2p+2\}}\mathsf{c}(2N\overline{01\tau},N\overline{0\sigma})\\&\sim \sum_{\tau_{p}\sqcup\sigma_{p+1}=\{2,\cdots,2p+2\}}\mathsf{c}(2N\overline{01\tau},N\overline{0\sigma}-N\partial \overline{01\sigma})\\&=\sum_{\tau_{p}\sqcup\sigma_{p+1}=\{2,\cdots,2p+2\}}\mathsf{c}(2N\overline{01\tau},N\overline{0\sigma}+\sum_{i}(-1)^iN\overline{01\partial_i\sigma})~.
        \end{aligned}    
    \end{equation}
    Writing $\sigma=\{a_0<a_1<\cdots<a_p\}$, $\partial_i\sigma$ means $\{a_0,\cdots,a_{i-1},a_{i+1},\cdots,a_p\}$.
    
    In this summation, $\mathsf{c}(s_{01\tau}^{2N},s_{1\sigma}^{N})$ is already in $E_\id(m|_0)$; in the summation of $\tau_p\subset \{2,\cdots,2p+2\}$ and $i\in \sigma$, every combination of $(\tau,\sigma,i)$ cancels with $(\sigma-\{i\},\tau\sqcup\{i\},i)$. Thus, $e\in E_\id(m|_0)$.

    It is obvious that $2e\sim 0$, and next we only need to prove $e\notin E_\id(m)$. To prove that, we need to construct $\varphi\in R(m)$ such that $e(\varphi)\ne 0$. For any $\sigma_{p+2},\sigma'_{p+2}\subset \{0,\cdots, 2p+2\}$, we define
    \begin{equation}
        \varphi(s_{\{i\}\sqcup\{0,\cdots,p\}},s_{\{j\}\sqcup\{p+2,\cdots,2p+2\}})=\frac{1}{4N^2}~,
    \end{equation}
    making it anti-symmetric, and let $\varphi=0$ in other cases. We have $e(\varphi)=1/2$, and next, we check that $\varphi$ satisfies Eq.~\eqref{axiomConfiguration}. First, we have $4N^2\varphi(s',s)=0\in \RR/\ZZ$; second, for any $p+1\le i<i'\le 2p+2$, we have
    \begin{equation}
        \varphi(\partial s_{\{i,i'\}\sqcup\{0,\cdots,p\}},s_{\{j\}\sqcup\{p+2,\cdots,2p+2\}})=0~.
    \end{equation}
    Actually, only $\{i\}\sqcup\{0,\cdots,p\}$ and $\{i'\}\sqcup\{0,\cdots,p\}$ contribute, but they cancel each other.
\end{proof}

\begin{theorem}
    Let $m=m_{p}(\partial\Delta^{2p+2},\ZZ_{N})$ and $p$ is even. Then,
    \begin{equation}
        e=N\sum_{\tau_{p}\sqcup\sigma_{p+1}=\{2,\cdots,2p+2\}}(-1)^{\sum_{x\in\tau}x}\mathsf{c}(\overline{01\tau},\overline{0\sigma})\in E_\inv(m)~.
    \end{equation}
    \label{thm: even p dim 2p+1 statistics}
\end{theorem}
\begin{proof}
    Similar to the previous theorem, we have
    \begin{equation}
        e\sim \sum_{\tau_{p}\sqcup\sigma_{p+1}=\{2,\cdots,2p+2\}}(-1)^{\sum_{x\in\tau}x}\mathsf{c}(N\overline{01\tau},\overline{0\sigma}+\sum_{i=0}^p(-1)^i\overline{01\partial_i\sigma})~,
    \end{equation}
    and we will prove that
    \begin{equation}
        \sum_{\tau_{p}\sqcup\sigma_{p+1}=\{2,\cdots,2p+2\}}\sum_{i=0}^{p}(-1)^{i+\sum_{x\in\tau}x}\mathsf{c}(\overline{01\tau},\overline{01\partial_i\sigma})\sim 0~.
    \end{equation}
    Actually, any combination of $(\tau,\sigma,i)$ is paired with $(\tau',\sigma',i')$, such that $\tau=\partial_{i'}\sigma'$ and $\tau'=\partial_i\sigma$ (this implies $\tau\cap\tau'=\emptyset$), and the corresponding terms in the sum are
    \begin{equation}
        (-1)^{i+\sum_{x\in\tau}x}\mathsf{c}(\overline{01\tau},\overline{01\tau'})+(-1)^{i'+\sum_{x\in\tau'}x}\mathsf{c}(\overline{01\tau'},\overline{01\tau})~,
    \end{equation}
    and it is equivalent to zero if and only if 
    \begin{equation}
        (-1)^{i+\sum_{x\in\tau}x}=(-1)^{i'+\sum_{x\in\tau'}x}~,
    \end{equation}
    i.e.,
    \begin{equation}
        i+i'+\sum_{x\in\tau\sqcup\tau'}x =0 \pmod{2}~.
    \end{equation}
    From the relations among $\tau$ and $\tau',i$, we find that
    \begin{equation}
        \tau\sqcup \tau'=\{2,\cdots,2p+2\}-\{k\}~.
    \end{equation}
    Actually, $k$ is the $i$-th element of $\sigma=\tau'\sqcup\{k\}$ and also the $i'$-th element of $\sigma'=\tau\sqcup\{k\}$. This implies that $i+i'=k-2$. Thus,

    \begin{equation}
        i+i'+\sum_{x\in\tau\sqcup\tau'}x  =\sum_{x=2}^{2p+2}x \pmod{2}~,
    \end{equation}
    and it is zero if and only if $p$ is even.
\end{proof}

\begin{theorem}
    Let $m=m_{p}(\partial\Delta^{2p+3},\ZZ_{N})$ and $p$ is even. Then,
    \begin{equation}\label{eq: Statistics2p+2}
        e=  {\sum}'_{\{k\}\sqcup\tau_{p+1}\sqcup\sigma_{p+1}=\{1,\cdots,2p+3\}}(-1)^{\#(\tau<k)+k+\sum\tau}\mathsf{c}(0\tau,0\sigma)\in E_\inv(m)~,
    \end{equation}
    where two terms obtained by exchanging $\sigma$ and $\tau$ are only counted once.
    \label{thm: even p dim 2p+2 statistics}
\end{theorem}
\begin{proof}
    We note that exchanging $\sigma$ and $\tau$ does not change the equivalence class of $e$. Actually, the difference of the two terms is
    \begin{equation}
\begin{aligned}
        &(-1)^{\#(\tau<k)+k+\sum\tau}\mathsf{c}(0\tau,0\sigma)-(-1)^{\#(\sigma<k)+k+\sum\sigma}\mathsf{c}(0\sigma,0\tau)\\&\sim (-1)^{\#(\tau<k)+k+\sum\tau}\bigg(1+(-1)^{\#(\sigma<k)+\sum\sigma+\#(\tau<k)+\sum\tau}\bigg)\mathsf{c}(0\tau,0\sigma)\\&=(-1)^{\#(\tau<k)+k+\sum\tau}\bigg(1+(-1)^{(k-1)+\sum\sigma+\sum\tau}\bigg)\mathsf{c}(0\tau,0\sigma)\\&=0~.
\end{aligned}
    \end{equation}
    Here, we have used
    \begin{equation}
        \#(\tau<k)+\#(\sigma<k)=k-1~,
    \end{equation}
    and that
    \begin{equation}
        k+\sum\sigma+\sum\tau=\sum\{1,\cdots,2p+3\}
    \end{equation}
     is even when $p$ is even.

    To show $e\in E_\inv(m)$, the only nontrivial problem is to show $e\in E_\id(m|_0)$.
    We consider an expression
    \begin{equation}
    f=\sum_{\tau_{p+1}\sqcup\sigma_{p+1}=\{2,\cdots,2p+3\}}(-1)^{\sum\tau}\mathsf{c}(\partial (01\tau),\partial(01\sigma))\in E_\id(m)~,
    \end{equation}
    where two terms obtained by exchanging $\sigma$ and $\tau$ are only counted once. 
    We decompose $f$ into basic terms $\mathsf{c}(a,b)$, where $a,b$ are two simplices. These terms are in three classes: $a\cap b=\{0\}$, $a\cap b=\{1\}$, and $a\cap b=\{0,1\}$. According to this classification, we write $f=e_0+e_1+e_{01}\sim 0$. We are going to show $e=-e_0$ and $e_{01}\sim 0$. Obviously we have $e_1\in E_\id(m|_0)$, so $e\sim e_1\in E_\id(m|_0)$.
    
    We have
    \begin{equation}
        e_0= \sum_{\tau_{p+1}\sqcup\sigma_{p+1}=\{2,\cdots,2p+3\}}(-1)^{\sum\tau}(\mathsf{c}( 01\partial \tau,0\sigma)+\mathsf{c}( 0\tau,01\partial\sigma)+\mathsf{c}( 0\tau,0\sigma))~.
    \end{equation}
    The last term is equal to $k=1$ part of $-e$. The sum of the first term and the second term is
    \begin{equation}
        \sum_{\tau_{p+1}\sqcup\sigma_{p+1}=\{2,\cdots,2p+3\}}(-1)^{\sum\tau}(\mathsf{c}( 01\partial \tau,0\sigma)~,
    \end{equation}
    where exchanging $\sigma$ and $\tau$ are counted as different terms. Expanding $\partial\tau$, we have
    \begin{equation}
    \begin{aligned}
            &\sum_{\tau_{p+1}\sqcup\sigma_{p+1}=\{2,\cdots,2p+3\}}(-1)^{\sum\tau}(-1)^i\mathsf{c}( 01\partial_i \tau,0\sigma)\\&=\sum_{\{k\}\sqcup\tau_{p}'\sqcup\sigma_{p+1}=\{2,\cdots,2p+3\}}(-1)^{\#(\tau'<k)+k+\sum\tau}\mathsf{c}( 01\tau',0\sigma)~.
    \end{aligned}
    \end{equation}
    This equals the $k\ne1$ part of $-e$.

    Next, we prove that
    \begin{equation}
    e_{01}=\sum_{\tau_{p+1}\sqcup\sigma_{p+1}=\{2,\cdots,2p+3\}}\sum_{i,j=0}^p(-1)^{i+j+\sum\tau}\mathsf{c}( 01\partial_i\tau,01\partial_j\sigma)\sim 0~.
    \end{equation}
    Note that $p$ is even implies that $\sum_{x=2}^{2p+3}x$ is odd, so we  have
    \begin{equation}
        (-1)^{\sum\tau}\mathsf{c}( 01\partial_i\tau,01\partial_j\sigma)\sim(-1)^{\sum\sigma}\mathsf{c}( 01\partial_j\sigma,01\partial_i\tau)~.
    \end{equation}
    Next for any $\alpha_p,\beta_p\subset \{2,\cdots,2p+3\}$ and $\alpha\cap\beta=\emptyset$, we want to find all terms involving $\mathsf{c}(01\alpha,01\beta)$ or $\mathsf{c}(01\beta,01\alpha)$. There are $2\le k<l\le 2p+3$ such that $\alpha\sqcup\beta=\{2,\cdots,2p+3\}-\{k,l\}$. there are two solutions of $\sigma,\tau,i,j$, which correspond to $\sigma_1=\alpha\sqcup\{k\},\tau_1=\beta\sqcup\{l\}$ and $\sigma_2=\alpha\sqcup\{l\},\tau_2=\beta\sqcup\{k\}$. It is easy to verify that $j_1+i_2=k-2$ and $i_1+j_2=l-2$, and thus these two terms cancel.
\end{proof}

%%%%%%%%%%%%%%%%%%%%%%%%%%%%%%%%%%%%%%%%%%%%%%%%%%%%%%
\section{Anyon membrane statistics in $(6{+}1)$D}\label{sec: Semionic-membrane toric code in (6+1)D}

In this section, we construct the \emph{semionic-membrane toric code} in 6 spatial dimensions, obtained by condensing the $e^2 m^2$ membrane excitation in the standard $\mathbb{Z}_4$ toric code.  
The story closely parallels the semion in two spatial dimensions: a $\mathbb{Z}_2$ particle there acquires $\mathbb{Z}_4$ statistics, with the semion as the generator. Semionic statistics are special to two dimensions; in spatial dimensions $d \ge 3$, only fermions survive.  
Similarly, the $\mathbb{Z}_2$ membrane in 6 dimensions carries $\mathbb{Z}_4$ statistics, and the generator of this $\mathbb{Z}_4$ group is what we refer to as the “semionic membrane.” This semionic membrane exists only in 6 dimensions. In dimensions $d \ge 7$, the $\mathbb{Z}_4$ structure collapses to its $\mathbb{Z}_2$ subgroup, giving rise to what we call the fermionic membrane statistics. We will demonstrate that this fermionic membrane statistics is governed by the stable cohomological operation $\operatorname{Sq}^4$ in Sec.~\ref{sec: Fermionic-membrane toric code in arbitrary dimensions}.  

Importantly, unlike particles and loops—which exhibit only $\mathbb{Z}_2$ statistics in high dimensions—membrane excitations possess a stable $\mathbb{Z}_2 \times \mathbb{Z}_3$ statistics structure. The fermionic membrane statistics studied in Sec.~\ref{sec: Fermionic-membrane toric code in arbitrary dimensions} correspond to the $\mathbb{Z}_2$ component. The remaining $\mathbb{Z}_3$ component is known as the \emph{Pontryagin} statistics; its detecting unitary process was presented in Ref.~\cite{feng2025anyonic}. However, the $\mathbb{Z}_3$ part is intrinsically non-Pauli, and any unitary operator capable of detecting it must be at least Clifford.  
Thus, this paper provides the complementary half of the picture in Ref.~\cite{feng2025anyonic}, supplying the missing $\mathbb{Z}_2$ component of the full $\mathbb{Z}_2 \times \mathbb{Z}_3$ membrane statistics.

%%%%%%%%%%%%%%%%%%%%%%%%%%%%%%%%%%%%%%%%%%%%%%%%%%%%%%
\subsection{$(6{+}1)$D semionic-membrane toric code via condensation}

In this section, we construct the Pauli stabilizer model for the semionic-membrane toric code using $\ZZ_4$ qudits. This model realizes a $\ZZ_2$ topological order whose field-theoretic description is captured by the following TQFT action:
\begin{equation}
    S = \pi \int a_3 \cup \delta b_3 + \frac{2\pi}{4} \int b_3 \cup \delta b_3 ,
\end{equation}
which represents a Dijkgraaf–Witten–type TQFT classified by the cohomology class $\tfrac{1}{4} b_3 \cup \delta b_3 \in H^7(B^3 \ZZ_2, \mathbb{R}/\mathbb{Z})$ for 3-form $\ZZ_2$ gauge theories.

We start with the $(6{+}1)$D membrane-only $\ZZ_4$ toric code, which has a $\ZZ_4$ qudit on each tetrahedron $t$ (3-cell) of a cellular decomposition of the 6-dimensional spatial manifold. Its stabilizer Hamiltonian is
\begin{equation}
    H_{\ZZ_4\text{-}\mathrm{TC}} 
    = -\sum_{f} A_{f} \;-\; \sum_{c_4} B_{c_4} \;+\; \mathrm{h.c.}~,
\label{eq: standard (6+1)D Z4 toric code}
\end{equation}
where the first sum runs over all faces $f$ (2-cells) and the second over all 4-cells $c_4$. The stabilizers are defined by
\begin{equation}
    A_f := X_{\delta \bface}
    = \prod_{t} X_t^{\delta \bface(t)},
    \qquad
    B_{c_4} := Z_{\partial c_4}
    = \prod_{t} Z_t^{\bt(\partial c_4)}~.
    \label{eq: A and B in standard (6+1)D Z4 toric code}
\end{equation}
Although Eqs.~\eqref{eq: standard (6+1)D Z4 toric code} and~\eqref{eq: A and B in standard (6+1)D Z4 toric code} take the same algebraic form as the $(5{+}1)$D toric code expressions in Eqs.~\eqref{eq: standard (5+1)D Z4 toric code} and~\eqref{eq: A and B in standard (5+1)D Z4 toric code}, the resulting excitations are different in 6 spatial dimensions. $A_f$ is the product of $X_t$ over all tetrahedra $t$ for which the face $f$ appears in the boundary of $t$, while $B_{c_4}$ is the product of $Z_t$ over all tetrahedra on the boundary of the 4-cell $c_4$. In this $(6{+}1)$D membrane-only toric code, a single $Z_t$ violates the $A_f$ terms on the faces in $\partial t$, creating an $e$-membrane excitation. Conversely, a single $X_t$ violates the $B_{c_4}$ terms on the coboundary $\delta t$, producing an $m$-membrane excitation on the dual lattice, since the 4-cochain $\delta \bt$ is Poincaré dual to a 2-dimensional membrane.

Now we condense the $e^2 m^2$ membranes by adding the membrane-creation term on each tetrahedron $t$,
\begin{equation}
    C_t
    =
    X_t^{2}
    \prod_{t'} Z_{t'}^{2\int \boldsymbol{t}' \cup \boldsymbol{t}}~,
\end{equation}
to the Hamiltonian. The resulting condensed Hamiltonian can be written as
\begin{equation}
    H_{\mathrm{condensed}}^{(6+1)\mathrm{D}}
    = -\sum_{f} G_f
      \;-\; \sum_{c_4} B_{c_4}^{2}
      \;-\; \sum_{t} C_t
      \;+\; \mathrm{h.c.},
\label{eq: H_condensed in 6+1D}
\end{equation}
where $f$ runs over faces, $t$ over tetrahedra, and $c_4$ over 4-cells.  
The first stabilizer is
\begin{equation}
    G_f := A_f \prod_{c_4} B_{c_4}^{\int \boldsymbol{f} \cup \boldsymbol{c}_4}
    = X_{\delta \bface} \prod_{c_4} Z_{\partial c_4}^{\int \boldsymbol{f} \cup \boldsymbol{c}_4}
    = X_{\delta \bface} \prod_{t} Z_{t}^{\int \delta \bface \cup \boldsymbol{t}}~.
\label{eq: Gf in 6+1D}
\end{equation}

The two membrane operators take the compact form
\begin{equation}
    V^C_t = Z_t^{2},
    \qquad
    V^F_t = X_t \prod_{t'} Z_{t'}^{\int \boldsymbol{t} \cup \boldsymbol{t}'}~.
\end{equation}
A subtlety arises when analyzing the fusion rules of the two membrane excitations.  
For the charge membrane created by $V^C_t$, the fusion group is manifestly $\mathbb{Z}_2$, since $(V^C_t)^2 = 1$ as an operator.  
In contrast, the flux membrane created by $V^F_t$ obeys $(V^F_t)^4 = 1$ but $(V^F_t)^2 \neq 1$ at the operator level, so it naively appears to have a $\mathbb{Z}_4$ fusion rule.  
However, as in the $(4{+}1)$D case, one can show that the doubled flux membrane $\prod_{t \in \pd(D)} (V^F_t)^2$ for any three-dimensional disk $D$ in the dual lattice can be expressed as a product of stabilizers in $D$, up to local operators near $\partial D$.  
Thus, the doubled flux membrane is topologically trivial and lies in the same \emph{superselection sector} as the vacuum membrane excitation~\cite{Kitaev:2005hzj}.

To make this structure explicit, we decorate $V^F_t$ by a correction localized near the coboundary $\delta \boldsymbol{t}$:
\begin{equation}
    \tilde{V}^F_t :=
    V^F_t \prod_{t'} Z_{t'}^{\int \boldsymbol{t}' \cup_1 \delta \boldsymbol{t}}
    =
    X_t \prod_{t'}
    Z_{t'}^{\int \boldsymbol{t} \cup \boldsymbol{t}'
        + \boldsymbol{t}' \cup_1 \delta \boldsymbol{t}}
    =
    X_t \prod_{t'}
    Z_{t'}^{\int -\boldsymbol{t}' \cup \boldsymbol{t}
        + \delta \boldsymbol{t}' \cup_1 \boldsymbol{t}}~,
\label{eq: modified flux membrane in 6+1D semionic-membrane toric code}
\end{equation}
where, in the last equality, we used the Leibniz rule for higher cup products,
\begin{equation}
    \delta(A_3 \cup_1 B_3)
    =
    \delta A_3 \cup_1 B_3
    - A_3 \cup_1 \delta B_3
    - A_3 \cup B_3
    - B_3 \cup A_3~.
\end{equation}
A direct computation shows that
\begin{equation}
    (\tilde{V}^F_t)^{2}
    =
    X_t^{2}
    \prod_{t'} Z_{t'}^{2\int \boldsymbol{t}' \cup \boldsymbol{t}
      + \delta \boldsymbol{t}' \cup_1 \boldsymbol{t}}
    =
    C_t
    \prod_{c_4} B_{c_4}^{2\int \boldsymbol{c}_4 \cup_1 \boldsymbol{t}}~,
\end{equation}
which is a product of stabilizers.  
Therefore, the flux membrane created by $\tilde{V}^F_t$ obeys a $\mathbb{Z}_2$ fusion rule.

We now extend the definition of the operator to an arbitrary $3$-cochain
$\lambda \in C^3(M_6,\ZZ)$, where $M_6$ denotes the spatial manifold:
\begin{equation}
    \tilde{V}_\lambda
    =
    \prod_t X_t^{\lambda(t)}
    \prod_{t'}
    Z_{t'}^{\int -\boldsymbol{t}' \cup \lambda
        + \delta \boldsymbol{t}' \cup_1 \lambda}
    := X_\lambda
    \prod_{t'}
    Z_{t'}^{\int -\boldsymbol{t}' \cup \lambda
        + \delta \boldsymbol{t}' \cup_1 \lambda}~.
\end{equation}
Its commutator is given by
\begin{eqs}\label{eq: SemionicMembraneCommutator}
    [\tilde{V}_{\lambda}, \tilde{V}_{\lambda'}]
    =&~
    (\tilde{V}_{\lambda})^{-1}
    (\tilde{V}_{\lambda'})^{-1}
    \tilde{V}_{\lambda}
    \tilde{V}_{\lambda'} \\
    =&~
    i^{-\int \lambda' \cup \lambda
        - \lambda \cup \lambda'
        + \delta\lambda \cup_1 \lambda'
        - \delta\lambda' \cup_1 \lambda}~.
\end{eqs}

%%%%%%%%%%%%%%%%%%%%%%%%%%%%%%%%%%%%%%%%%%%%%%%%%%%%%%
\subsection{Semionic membrane statistics}

Next, we show that Eq.~\eqref{eq: Statistics2p+2} evaluates to $\pm i$ for this membrane excitation, indicating that it has semionic membrane statistics. We choose a mesoscopic triangulation $\partial\Delta^7$ of $M_6$ and label its vertices by $\{0,\cdots,7\}$. Each $p$-simplex $\sigma$ in $\partial\Delta^7$ naturally corresponds to a $p$-chain in $M_6$; we denote its Poincar\'e dual by $\sigma^*$, which is a $(6-p)$-cochain supported in the same region. For notational convenience, we represent a simplex by listing its vertices explicitly. If two vertex lists are disjoint, their (higher) cup products vanish.

Substituting Eq.~\eqref{eq: SemionicMembraneCommutator} into Eq.~\eqref{eq: Statistics2p+2}, we are led to evaluate the membrane statistical phase
\begin{equation}
    \mu_{\mathrm{membrane}} = i^n, ~~ \mathrm{with} ~~
    n=\sum_{\substack{\{k\}\sqcup\tau_{p+1}\sqcup\sigma_{p+1}\\=\{1,\cdots,2p+3\}}}
    (-1)^{\#(\tau<k)+k+\sum\tau}
    \int (0\sigma)^* \cup (0\tau)^*
    + (\partial (0\tau))^* \cup_1 (0\sigma)^*~.
\end{equation}
Unlike Eq.~\eqref{eq: Statistics2p+2}, here we treat $(\sigma,\tau)$ and $(\tau,\sigma)$ as distinct contributions, reflecting the fact that the four terms in Eq.~\eqref{eq: SemionicMembraneCommutator} naturally separate into two independent parts.

We first show that the $\cup_1$ contribution to $n$, denoted by $n_1$, vanishes. Writing $\partial\tau=\sum_i(-1)^i\partial_i\tau$, and denoting the $i$th element of $\tau$ by $l$, we have $i=\#(\tau<l)$, and therefore
\[
\partial\tau=\sum_{l\in\tau}(-1)^{\#(\tau<l)}(\tau-\{l\})~.
\]
Denoting $\tau-\{l\}$ by $\mu$ in the summation defining $n_1$, we obtain
\begin{equation}
    n_1=-\sum_{\{k,l\}\sqcup\mu_{p}\sqcup\sigma_{p+1}=\{1,\cdots,2p+3\}}(-1)^{\#(\mu\sqcup l<k)+\sum\sigma+\#(\mu<l)}\int  (0\mu)^* \cup_1 (0\sigma)^*~,
\end{equation}
where we have used $k+\sum \tau + \sum \sigma = \sum_{n=1}^{2p+3} n = 0 \pmod{2}$ for even $p$.
For fixed $\mu$ and $\sigma$, there are two corresponding terms related by exchanging $k$ and $l$. The sum of the exponents in the associated $(-1)$ factors is
\begin{equation}
    \#(\mu\sqcup l<k)+\sum\sigma+\#(\mu<l)+\#(\mu\sqcup k<l)+\sum\sigma+\#(\mu<k),
\end{equation}
which is always odd, since either $l<k$ or $k<l$. Hence, these terms cancel pairwise and we conclude that $n_1=0$.

We now compute the $\cup$ contribution,
\begin{equation}\label{eqRandom1}
    n_0=\sum_{\{k\}\sqcup\tau_{p+1}\sqcup\sigma_{p+1}=\{1,\cdots,2p+3\}}(-1)^{\#(\tau<k)+\sum\sigma}\int (0\sigma)^* \cup (0\tau)^*.
\end{equation}
The key observation is that, for fixed $\sigma$, the summation over $\{k\}\sqcup\tau_{p+1}$ together with the sign $(-1)^{\#(\tau<k)}$ precisely reproduces the boundary of $\alpha=\{1,\cdots,2p+3\}-\sigma$. Thus, up to an overall sign, we may rewrite
\begin{equation}
    \begin{aligned}
    n_0&=\pm\sum_{\alpha_{p+2}\sqcup\sigma_{p+1}=\{1,\cdots,2p+3\}}(-1)^{\sum\sigma}\int (0\sigma)^* \cup (\partial(0\alpha))^*\\
    &=\pm\sum_{\alpha_{p+2}\sqcup\sigma_{p+1}=\{1,\cdots,2p+3\}}(-1)^{\sum\alpha}\int (\partial(0\sigma))^* \cup (0\alpha)^*\\
    &=\pm\sum_{\alpha_{p+2}\sqcup\beta_{p}\sqcup\{k\}=\{1,\cdots,2p+3\}}(-1)^{\#(\beta<k)+\sum\alpha}\int (0\beta)^* \cup (0\alpha)^*~.
    \end{aligned}
\end{equation}
This expression has the same structure as Eq.~\eqref{eqRandom1}, except that the size of $\beta$ is reduced by one. Iterating this procedure yields
\begin{equation}
n_0=\pm\sum_{\alpha_{2p+2}\sqcup\beta_{0}\sqcup\{k\}=\{1,\cdots,2p+3\}}(-1)^{\#(\beta<k)+\sum\alpha}\int (0\beta)^* \cup (0\alpha)^*.
\end{equation}
Consequently,
\begin{equation}
\begin{aligned}
    n_0&=\pm\sum_{k\in\{1,\cdots,2p+3\}}(-1)^{k}\int (0)^* \cup (\{0,\cdots,\hat{k},\cdots,2p+3\})^*\\
    &=\pm\sum_{k\in\{1,\cdots,2p+3\}}(-1)^{k}\int (0)^* \cup M^*~.
\end{aligned}
\label{eq: n0 is +-1}
\end{equation}
Here, $M$ denotes the fundamental class of the manifold $M$, whose dual is the constant $1$. The dual of the vertex $0$ is a $d$-cochain with unit integral. We therefore obtain $n=\pm1$, implying that the statistical phase of the semionic membrane is $\pm i$.

A similar, and simpler, analysis applies to a general fusion group $\ZZ_N$, yielding the generator of statistics $T_{\mathrm{P}}\simeq \ZZ_{N \times \gcd(2,N)}$. As in the case of anyons in $(2{+}1)$D, the braiding of two identical membranes in $(6{+}1)$D produces twice the corresponding self-statistics.

%%%%%%%%%%%%%%%%%%%%%%%%%%%%%%%%%%%%%%%%%%%%%%%%%%%%%%
\subsection{Anyonic membrane excitations in $(6{+}1)$D membrane-only $\ZZ_N$ toric code}

In this subsection, we show that in the standard $(6{+}1)$D membrane-only $\ZZ_N$ toric code, the composite $em$ membrane carries nontrivial $\ZZ_N$ anyonic membrane statistics. This is reminiscent of the familiar fact that anyons in $(2{+}1)$D can exhibit $\ZZ_N$ spin statistics. However, this analogy is not generic: it is specific to spatial dimensions $d =4k{+}2$. By contrast, the $em$ loop in the standard $(4{+}1)$D loop-only $\ZZ_N$ toric code is bosonic~\cite{Chen2023Loops4d}.

The unitary operator creating an $em$-membrane excitation in the $(6{+}1)$D membrane-only $\ZZ_N$ toric code is
\begin{equation}
    V^{em}_t := X_t \prod_{t'} Z_{t'}^{\int \bt' \cup \bt}~.
\end{equation}
Promoting $t$ to an arbitrary $3$-cochain $\lambda$, we define
\begin{equation}
    V^{em}_\lambda := X_\lambda \prod_{t'} Z_{t'}^{\int \bt' \cup \lambda}~,
\end{equation}
whose commutator is
\begin{equation}
    [V^{em}_{\lambda_1}, V^{em}_{\lambda_2}]
    =
    \omega^{\int -\lambda_1 \cup \lambda_2 + \lambda_2 \cup \lambda_1}~,
\end{equation}
where $\omega=\exp(2\pi i/N)$.
Substituting this commutator into the statistics formula \eqref{eq: Statistics2p+2}, with the identification $U_{0ijk}:=V^{em}_{(0ijk)^*}$, we obtain
\begin{eqs}
    \mu_{\mathrm{membrane}}
    &=
    {\prod}'_{k\sqcup\sigma_3\sqcup\tau_3=\{1,2,3,4,5,6,7\}}
    \bigl[U_{0\sigma_3},\,U_{0\tau_3}\bigr]^{\,o'_\pm}
    \label{eq: membrane_stat_product} :=
    \omega^{n_0}~,
\end{eqs}
where
\begin{equation}
    o'_\pm := (-1)^{\,k+\#(\sigma_3<k)+\sum \sigma_3}~,
\end{equation}
and the exponent $n_0$ is exactly the quantity defined in Eq.~\eqref{eqRandom1}.
Since the previous subsection showed that $n_0=\pm 1$, we conclude that
\begin{equation}
    \mu_{\mathrm{membrane}}=\omega^{\pm 1}~,
\end{equation}
so the $em$ membrane indeed exhibits $\ZZ_N$ anyonic membrane statistics. In particular, the generator has a statistical phase $\omega$ (up to orientation conventions).

%%%%%%%%%%%%%%%%%%%%%%%%%%%%%%%%%%%%%%%%%%%%%%%%%%%%%%
\section{Fermionic-membrane toric code in arbitrary dimensions}
\label{sec: Fermionic-membrane toric code in arbitrary dimensions}

The semionic membrane discussed above is special to $(6{+}1)$D, in direct analogy with the semion in $(2{+}1)$D. In higher dimensions, only $\mathbb{Z}_2$ membrane statistics can remain nontrivial. In this section, we introduce a new family of $\mathbb{Z}_2$ topological orders, which we call the \emph{fermionic-membrane toric codes}, defined in $d$ spatial dimensions. These phases host fermionic membrane excitations and bosonic $(d-4)$-brane excitations with mutual $-1$ braiding statistics.

Our construction applies for all $d \geq 6$. However, when $d=6$, the resulting stabilizer group coincides with that of the standard membrane-only bosonic $\mathbb{Z}_2$ toric code. Genuinely new topological orders arise only for $d \geq 7$.

%%%%%%%%%%%%%%%%%%%%%%%%%%%%%%%%%%%%%%%%%%%%%%%%%%%%%%
\subsection{Review of fermionic-particle toric codes in $(d{+}1)$D}

We begin by reviewing the exact bosonization construction introduced in Refs.~\cite{Chen:2017fvr, Chen:2018nog, Chen2020}, where the topological action is given by
\begin{equation}
    S = \pi \int \operatorname{Sq}^2 B_{d-1} = \pi \int B_{d-1} \cup_{d-3} B_{d-1}~.
    \label{eq: action of Sq^2 B}
\end{equation}
This action describes the fermionic-particle toric code. The Hilbert space consists of $\mathbb{Z}_2$ qubits residing on each $(d-1)$-cell. The corresponding stabilizer Hamiltonian is
\begin{equation}
    H_{\mathrm{TC}}^{\mathrm{fermionic-particle}} = - \sum_{c_{d-2}} G_{c_{d-2}} - \sum_{c_d} B_{c_d}~,
\end{equation}
with stabilizers
\begin{equation}
    G_{c_{d-2}} := X_{\delta \bc_{d-2}} \prod_{c_{d-1}} Z_{c_{d-1}}^{\int \delta \bc_{d-2} \cup_{d-2} \bc_{d-1}}~, \quad \mathrm{and} \quad
    B_{c_d} = Z_{\partial c_d}~.
    \label{eq: G and B for exact bosonization}
\end{equation}
The hopping operator of the fermionic particle is defined as
\begin{equation}
    V_{c_{d-1}} = X_{c_{d-1}} \prod_{c'_{d-1}} Z_{c'_{d-1}}^{\int \bc'_{d-1} \cup_{d-2} \bc_{d-1}}~.
\end{equation}
In contrast, the single operator $Z_{c_{d-1}}$ creates a bosonic $(d-2)$-brane excitation on $\partial c_{d-1}$, which has mutual $-1$ braiding with the fermionic particle.
To verify that the hopping operator is indeed fermionic, we consider the following $T$-junction process~\cite{Levin2003Fermions},
\begin{equation}
    \mu_{\mathrm{particle}} := [U_{01}, U_{02}]~[U_{02}, U_{03}]~[U_{03}, U_{01}]~,
\end{equation}
where we choose $U_{0i} = V_{(0i)^*}$. A direct computation shows that $\mu_{\mathrm{particle}} = (-1)^n$, with
\begin{eqs}
    n &= \int (01)^* \cup_{d-2} (02)^* + (02)^* \cup_{d-2} (01)^* 
    + (02)^* \cup_{d-2} (03)^* + (03)^* \cup_{d-2} (02)^* \\
    & \qquad 
    + (03)^* \cup_{d-2} (01)^* + (01)^* \cup_{d-2} (03)^* \\
    &= \int (01)^* \cup_{d-1} (\partial(02))^* + (\partial (01))^* \cup_{d-1} (02)^*
    + (02)^* \cup_{d-1} (\partial(03))^* + (\partial (02))^* \cup_{d-1} (03)^* \\
    & \qquad + (03)^* \cup_{d-1} (\partial(01))^* + (\partial (03))^* \cup_{d-1} (01)^* \\
    &= \int (01)^* \cup_{d-1} (0)^* + (0)^* \cup_{d-1} (02)^* 
    + (02)^* \cup_{d-1} (0)^* + (0)^* \cup_{d-1} (03)^* \\
    & \qquad + (03)^* \cup_{d-1} (0)^* + (0)^* \cup_{d-1} (01)^* \\
    &= \int (0)^* \cup_d (\partial(01))^*
    + (0)^* \cup_d (\partial(02))^*
    +(0)^* \cup_d (\partial(03))^* = 3 \int (0)^* \cup_d (0)^* \\
    &= 1 \pmod{2}~,
\end{eqs}
where we have used the identity
\begin{equation}
    \delta(A \cup_i B) = \delta A \cup_i B + A \cup_i \delta B + A \cup_{i-1}B + B \cup_{i-1} A \pmod{2},
\end{equation}
together with the fact that all (higher) cup products between disjoint cochains vanish.

It is worth emphasizing that the action in Eq.~\eqref{eq: action of Sq^2 B} is nontrivial only for $d \geq 3$. When $d=2$, the $G_v$ term in Eq.~\eqref{eq: G and B for exact bosonization} reduces to
\begin{equation}
    G_v = X_{\delta \bv} \prod_e Z_e^{\int \delta\bv \cup \be} = X_{\delta \bv} \prod_e Z_e^{\int \bv \cup \delta\be} = X_{\delta \bv} \prod_f B_f^{\int \bv \cup \bface}~.
\end{equation}
In this case, one can multiply $G_v$ by appropriate plaquette stabilizers $B_f$ to cancel the factor $\prod_f B_f^{\int \bv \cup \bface}$ without changing the stabilizer group. The resulting model is therefore equivalent to the standard bosonic $\mathbb{Z}_2$ toric code, featuring a bosonic particle and a bosonic $(d-2)$-brane with mutual braiding.

%%%%%%%%%%%%%%%%%%%%%%%%%%%%%%%%%%%%%%%%%%%%%%%%%%%%%%
\subsection{Fermionic-membrane toric codes in $(d{+}1)$D}

We now generalize the above construction to higher-form symmetries. We start from the topological action
\begin{equation}
    S = \pi \int \operatorname{Sq}^4 B_{d-3} = \pi \int B_{d-3} \cup_{d-7} B_{d-3}~.
\end{equation}
We place $\mathbb{Z}_2$ qubits on each $(d{-}3)$-cell of the spatial manifold $M_d$. Following the same derivation as in the exact bosonization construction, we arrive at the stabilizer Hamiltonian
\begin{equation}
    H_{\mathrm{TC}}^{\mathrm{fermionic-membrane}} = - \sum_{c_{d-4}} G_{c_{d-4}} - \sum_{c_{d-2}} B_{c_{d-2}}~,
\end{equation}
where
\begin{equation}
    G_{c_{d-4}} := X_{\delta \bc_{d-4}} \prod_{c_{d-3}} Z_{c_{d-3}}^{\int \delta \bc_{d-4} \cup_{d-6} \bc_{d-3}}~, \quad
    B_{c_{d-2}} = Z_{\partial c_{d-2}}~.
    \label{eq: G and B for fermionic-membrane}
\end{equation}
The operator creating a membrane excitation is
\begin{equation}
    V_{c_{d-3}} = X_{c_{d-3}} \prod_{c'_{d-3}} Z_{c'_{d-3}}^{\int \bc'_{d-3} \cup_{d-6} c_{d-3}}~.
    \label{eq: U operator for fermionic membrane}
\end{equation}
This operator violates the $B_{c_{d-2}}$ stabilizers on the boundary $\delta \bc_{d-3}$, which forms a closed membrane in the dual lattice. On the other hand, the single operator $Z_{c_{d-3}}$ creates a bosonic $(d-4)$-brane excitation on $\partial c_{d-3}$, which has mutual $-1$ braiding with the membrane excitation.

We now verify that the membrane excitation created by Eq.~\eqref{eq: U operator for fermionic membrane} is fermionic. We choose the sublattice $\partial \Delta^7$, with vertices labeled by $\{0,1,2,3,4,5,6,7\}$, and compute the membrane statistics using the formula in Eq.~\eqref{eq: Statistics2p+2}.
The commutator of two membrane creation operators is
\begin{equation}
    [V_{c_{d-3}}, V_{c'_{d-3}}] = (-1)^{\int \bc_{d-3} \cup_{d-6} \bc'_{d-3} +  \bc'_{d-3} \cup_{d-6} \bc_{d-3}}~.
\end{equation}
Since we work with $\mathbb{Z}_2$ qubits, this phase is well defined modulo two, and the statistical formula simplifies to
\begin{equation}
    \mu_{\mathrm{membrane}} := (-1)^{n'_0}~, ~~\mathrm{with}~~
    n'_0=\sum_{\{k\}\sqcup\tau_{3}\sqcup\sigma_{3}=\{1,\cdots,7\}}
    \int_{M_d} (0\sigma)^* \cup_{d-6} (0\tau)^*~,
\label{eq: n'0 for membrane statistics}
\end{equation}
where we emphasize that the integral is taken over the spatial manifold $M_d$. 

However, note that both $(0\sigma)^*$ and $(0\tau)^*$ necessarily overlap along directions perpendicular to $\partial \Delta^7$. By the definition of higher cup products, this allows the integral over $M_d$ to be reduced to an ordinary cup product on the low-dimensional space:
\begin{equation}
    \int_{M_d} (0\sigma)^* \cup_{d-6} (0\tau)^* 
    = \int_{\partial \Delta^7} (0\sigma)^*\big|_{\partial \Delta^7} \cup (0\tau)^*\big|_{\partial \Delta^7}~.
\end{equation}
To illustrate this reduction, it is useful to consider a simpler example,
\begin{equation}
    A_2 \cup_1 B_2(0123) = A_2(023) B_2(012) + A_2(013) B_2(123)~.
\end{equation}
If both $A_2$ and $B_2$ contain the vertex $3$ (which may be viewed as perpendicular to the $012$ plane), we may define $\overline{A}_1(ij) := A_2(ij3)$ and $\overline{B}_1(ij) := B_2(ij3)$. In this case,
\begin{equation}
    A_2 \cup_1 B_2(0123) = \overline{A}_1 \cup \overline{B}_1 (012)~.
\end{equation}
Applying the same reasoning to Eq.~\eqref{eq: n'0 for membrane statistics}, we find that $n'_0$ reduces to $n_0$ in Eq.~\eqref{eq: n0 is +-1}, which takes the value $\pm 1$. We therefore conclude that the membrane excitation carries fermionic statistics,
\begin{equation}
    \mu_{\mathrm{membrane}}=-1~.
\end{equation}

%%%%%%%%%%%%%%%%%%%%%%%%%%%%%%%%%%%%%%%%%%%%%%%%%%%%%%
\section{Fermionic-volume toric code in arbitrary dimensions}
\label{sec: Fermionic-volume toric code in arbitrary dimensions}

In this section, we introduce the fermionic-volume toric code.
We begin with the simplest realization, obtained by condensing the $e^2 m^2$ volume excitation in the volume-only $\ZZ_4$ toric code in $(8{+}1)$ dimensions, which is discussed in Sec.~\ref{sec: Fermionic-volume toric code in (8+1)D}.
The corresponding Dijkgraaf--Witten cohomology class is
\begin{equation}
    \frac{1}{2}\operatorname{Sq}^4 \operatorname{Sq}^1 b_4
    = \frac{1}{4} \, b_4 \cup \delta b_4
    \in H^9(B^4\ZZ_2, \RR/\ZZ)~.
\end{equation}

The generalization to higher dimensions is straightforward and is presented in Sec.~\ref{sec: Fermionic-volume toric code in (d+1)D}.
More generally, in $(d{+}1)$ spacetime dimensions, we realize a topological quantum field theory characterized by the cocycle
\begin{equation}
    \frac{1}{2}\operatorname{Sq}^4 \operatorname{Sq}^1 b_{d-4}
    =
    \frac{1}{4}
    \bigl(
        b_{d-4} \cup_{d-9} b_{d-4}
        +
        b_{d-4} \cup_{d-8} \delta b_{d-4}
    \bigr)
    \in H^{d+1}(B^{d-4}\ZZ_2, \RR/\ZZ)~.
\end{equation}
The relevant relations among Steenrod operations are reviewed in Appendix~\ref{app: Calculation of stable Z2 operations}.

%%%%%%%%%%%%%%%%%%%%%%%%%%%%%%%%%%%%%%%%%%%%%%%%%%%%%%
\subsection{Fermionic-volume toric code in $(8{+}1)$D}\label{sec: Fermionic-volume toric code in (8+1)D}

We begin by defining the \emph{$(8{+}1)$D $\ZZ_4$ volume-only toric code}.
We place a $\ZZ_4$ qudit on each $4$-cell of an arbitrary triangulation (or cellulation) of the spatial manifold $M_8$.
Using the simplicial cohomology notation introduced in Refs.~\cite{Chen2023Highercup, sun2025cliffordQCA}, the stabilizers take the compact form
\begin{equation}
    A_{c_3}
    = \prod_{c_4} X_{c_4}^{\,\delta \boldsymbol{c}_3(c_4)}
    := X_{\delta \boldsymbol{c}_3},
    \qquad
    B_{c_5}
    = \prod_{c_4} Z_{c_4}^{\,\boldsymbol{c}_4(\partial c_5)}
    := Z_{\partial c_5}.
\end{equation}
Here, $A_{c_3}$ multiplies $X$ operators over the coboundary of the $3$-cell $c_3$, while $B_{c_5}$ multiplies $Z$ operators over the boundary of the $5$-cell $c_5$.
All orientation signs are implicitly encoded in the boundary and coboundary operators $\partial$ and $\delta$.

We now condense the $e^2 m^2$ volume excitation by enforcing the following condensation term supported on each $4$-cell $c_4$:
\begin{equation}
    \tilde{C}_{c_4}
    :=
    X_{c_4}^{\,2}\;
    \prod_{c_4'} Z_{c_4'}^{\,2\int \boldsymbol{c}_4' \cup \boldsymbol{c}_4 + \bc_4 \cup_1 \delta \bc'_4}.
\end{equation}
Keeping only the terms that commute with the condensation operators, the condensed Hamiltonian can be written as
\begin{equation}
    H_{\mathrm{condensed}}
    = -\sum_{c_3} G_{c_3}
      \;-\sum_{c_5} B_{c_5}^{2}
      \;-\sum_{c_4} \tilde{C}_{c_4}
      \;+\; \mathrm{h.c.},
    \label{eq: H_condensed of fermionic-volume TC in (8+1)D}
\end{equation}
where
\begin{equation}
    G_{c_3}
    :=
    A_{c_3}\prod_{c_5} B_{c_5}^{\int \boldsymbol{c}_3 \cup \boldsymbol{c}_5}
    =
    X_{\delta \boldsymbol{c}_3}\prod_{c_5} Z_{\partial c_5}^{\int \boldsymbol{c}_3 \cup \boldsymbol{c}_5}
    =
    X_{\delta \boldsymbol{c}_3}\prod_{c_4'} Z_{c_4'}^{\int \delta \boldsymbol{c}_3 \cup \boldsymbol{c}_4'}.
\label{eq: Gc3 in (8+1)D}
\end{equation}
The two excitation operators supported on a $4$-cell take the compact form
\begin{eqs}
    V^C_{c_4} &:= Z_{c_4}^2, \\
    \tilde{V}^F_{c_4}
    &:=
    X_{c_4}\prod_{c_4'}
    Z_{c_4'}^{\int \boldsymbol{c}_4' \cup \boldsymbol{c}_4
    - \delta \boldsymbol{c}_4' \cup_1 \boldsymbol{c}_4}.
\label{eq: modified flux volume in (8+1)D fermionic-volume toric code}
\end{eqs}
The operator $V^C_{c_4}$ violates the $G_{c_3}$ stabilizers along $\partial c_4$, creating a bosonic charge-volume excitation. 
In contrast, $\tilde{V}^F_{c_4}$ violates the $B_{c_5}^2$ stabilizers on the coboundary of $c_4$, creating a flux volume excitation, which we will show below to be fermionic.

The charge-volume excitation has a manifest $\ZZ_2$ fusion rule since $(V^C_{c_4})^2=1$. 
For the flux volume excitation, we verify that the square of $\tilde{V}^F_{c_4}$ is a product of stabilizers:
\begin{equation}
    (\tilde{V}^F_{c_4})^{2}
    =
    X_{c_4}^{2}\prod_{c_4'} Z_{c_4'}^{\,2\int \boldsymbol{c}_4' \cup \boldsymbol{c}_4
      - \delta \boldsymbol{c}_4' \cup_1 \boldsymbol{c}_4}
    =
    \tilde{C}_{c_4}\prod_{c_5} B_{c_5}^{\,2\int \boldsymbol{c}_5 \cup_1 \boldsymbol{c}_4}~.
\end{equation}
Hence, $(\tilde{V}^F_{c_4})^2$ acts trivially on the ground space, and the flux volume excitation created by $\tilde{V}^F_{c_4}$ also obeys a $\ZZ_2$ fusion rule.

We now verify that this flux volume has fermionic volume statistics, in close analogy with the $(4{+}1)$D analysis of fermionic loop statistics in Sec.~\ref{sec: Fermionic loop statistics in (4+1)D}. 
\begin{theorem}
The flux volume created by $\tilde{V}^F_{c_4}$ in Eq.~\eqref{eq: modified flux volume in (8+1)D fermionic-volume toric code} has fermionic volume statistics. In particular,
\begin{equation}
    \mu_{\mathrm{volume}}
    :=
    \prod_{\sigma_3 \sqcup \tau_4 = \{2,3,4,5,6,7,8\}}
    \bigl[\tilde{V}^F_{(01\sigma_3)^*},\, \tilde{V}^F_{(0\tau_4)^*}\bigr]^{2}
    = -1~,
\end{equation}
where $\sigma_3$ and $\tau_4$ are disjoint index sets with $|\sigma_3|=3$ and $|\tau_4|=4$, and $(0ijkl)^*$ denotes the $4$-cochain Poincar\'e dual to the 4-simplex $\langle 0ijkl\rangle$.
Moreover, this $\mu_{\mathrm{volume}}$ corresponds to the special cases of Theorem~\ref{thm: order 2 statistics for Z_2N} with $N=1$ and $p=3$.
\end{theorem}

\begin{proof}
We extend the definition of the operator to an arbitrary $4$-cochain $\lambda\in C^4(M_8,\ZZ)$ (where $M_8$ is the spatial manifold) by
\begin{equation}
    \tilde{V}^F_\lambda
    =
    \prod_{c_4} X_{c_4}^{\lambda(c_4)}
    \prod_{c_4'} Z_{c_4'}^{\int \boldsymbol{c}_4'\cup\lambda
        - \delta\boldsymbol{c}_4'\cup_1 \lambda}
    :=
    X_\lambda
    \prod_{c_4'} Z_{c_4'}^{\int \boldsymbol{c}_4'\cup\lambda
        - \delta\boldsymbol{c}_4'\cup_1 \lambda}
    ~.
\end{equation}
Its commutator takes the compact form
\begin{equation}
    [\tilde{V}^F_{\lambda},\tilde{V}^F_{\lambda'}]
    :=
    (\tilde{V}^F_{\lambda})^{-1}
    (\tilde{V}^F_{\lambda'})^{-1}
    \tilde{V}^F_{\lambda}
    \tilde{V}^F_{\lambda'}
    =
    i^{-\int \delta\lambda\cup_2\delta\lambda'}~.
\label{eq: (8+1)D commutator of V}
\end{equation}
The volume statistics can be represented as
\begin{eqs}
    \mu_{\mathrm{volume}}
    &=
    {\prod}'_{\sigma_4 \sqcup \tau_4 = \{1,2,3,4,5,6,7,8\}}
    \bigl[\tilde{V}^F_{(0\sigma_4)^*},\, \tilde{V}^F_{(0\tau_4)^*}\bigr]^{2} \\
    &=
    {\prod}'_{\sigma_4 \sqcup \tau_4 = \{1,2,3,4,5,6,7,8\}}
    (-1)^{\int \delta (0\sigma)^* \cup_2 \delta (0\tau)^*}
    \;:=\;
    (-1)^{n_0}~,
\end{eqs}
where the primed product indicates that each unordered partition $\{\sigma_4,\tau_4\}$ is counted only once (i.e., $\{\sigma_4,\tau_4\}$ and $\{\tau_4,\sigma_4\}$ contribute only one factor).

Without loss of generality, we fix $8\in\tau$ and write $\tau=\tau_3\sqcup\{8\}$ with $|\tau_3|=3$ and $|\sigma|=4$.
Then
\begin{eqs}
    n_0
    &=
    \sum_{\sigma_4 \sqcup \tau_3 = \{1,2,3,4,5,6,7\}}
    \int \delta(0\sigma_4)^* \cup_2 \delta(0\tau_3 8)^* \\
    &=
    \sum_{\sigma_4 \sqcup \tau_3 = \{1,2,3,4,5,6,7\}}
    \int (0\,\partial\sigma_4)^* \cup_2 (0\tau_3)^*
    \;+\;
    (0\,\partial\sigma_4)^* \cup_2 (0\,(\partial\tau_3)\,8)^*~.
\end{eqs}
The second term vanishes modulo two:
\begin{equation}
    \sum_{k \sqcup l \sqcup \sigma_3 \sqcup \tau_2 = \{1,2,3,4,5,6,7\}}
    \int (0\sigma_3)^* \cup_2 (0\tau_2 8)^*
    = 0 \pmod{2},
\end{equation}
since exchanging $k$ and $l$ produces the same contribution twice. Therefore,
\begin{equation}
    n_0
    =
    \sum_{k \sqcup \sigma_3 \sqcup \tau_3 = \{1,2,3,4,5,6,7\}}
    \int (0\sigma_3)^* \cup_2 (0\tau_3)^*~.
\end{equation}
Next, we symmetrize the sum:
\begin{equation}
    n_0
    =
    {\sum}'_{k \sqcup \sigma_3 \sqcup \tau_3 = \{1,2,3,4,5,6,7\}}
    \int (0\sigma_3)^* \cup_2 (0\tau_3)^*
    \;+\;
    (0\tau_3)^* \cup_2 (0\sigma_3)^*~,
\end{equation}
where the prime indicates that exchanging $\{\sigma_3,\tau_3\}$ is counted only once. Using the Leibniz rule for higher cup products, this becomes
\begin{equation}
    n_0
    =
    {\sum}'_{k \sqcup \sigma_3 \sqcup \tau_3 = \{1,2,3,4,5,6,7\}}
    \int (\partial\,0\sigma_3)^* \cup_3 (0\tau_3)^*
    \;+\;
    (0\sigma_3)^* \cup_3 (\partial\,0\tau_3)^*~.
\end{equation}
We now use the identity
\begin{eqs}
    &{\sum}'_{k \sqcup \sigma_3 \sqcup \tau_3 = \{1,2,3,4,5,6,7\}}
    \int (\partial\,0\sigma_3)^* \cup_3 (0\tau_3)^*
    \;+\;
    (\partial\,0\tau_3)^* \cup_3 (0\sigma_3)^* \\
    &=
    {\sum}_{k \sqcup \sigma_3 \sqcup \tau_3 = \{1,2,3,4,5,6,7\}}
    \int (\partial\,0\sigma_3)^* \cup_3 (0\tau_3)^* \\
    &=
    {\sum}_{k \sqcup l \sqcup \sigma_2 \sqcup \tau_3 = \{1,2,3,4,5,6,7\}}
    \int (0\sigma_2)^* \cup_3 (0\tau_3)^*
    = 0 \pmod{2},
\end{eqs}
since exchanging $k$ and $l$ again produces the same term twice. Hence, we may rewrite
\begin{eqs}
    n_0
    &=
    {\sum}'_{k \sqcup \sigma_3 \sqcup \tau_3 = \{1,2,3,4,5,6,7\}}
    \int (\partial\,0\tau_3)^* \cup_3 (0\sigma_3)^*
    \;+\;
    (0\sigma_3)^* \cup_3 (\partial\,0\tau_3)^* \\
    &=
    {\sum}'_{k \sqcup \sigma_3 \sqcup \tau_3 = \{1,2,3,4,5,6,7\}}
    \int (\partial\,0\sigma_3)^* \cup_4 (\partial\,0\tau_3)^*~.
\end{eqs}
We now iterate the same reduction:
\begin{eqs}
    n_0
    &=
    7\cdot
    {\sum}'_{\sigma_3 \sqcup \tau_3 = \{1,2,3,4,5,6\}}
    \int (\partial\,0\sigma_3)^* \cup_4 (\partial\,0\tau_3)^* \\
    &=
    7\cdot 5\cdot
    {\sum}'_{\sigma_2 \sqcup \tau_2 = \{1,2,3,4\}}
    \int (\partial\,0\sigma_2)^* \cup_6 (\partial\,0\tau_2)^* \\
    &=
    7\cdot 5\cdot 3\cdot
    {\sum}'_{\sigma_1 \sqcup \tau_1 = \{1,2\}}
    \int (\partial\,0\sigma_1)^* \cup_8 (\partial\,0\tau_1)^* \\
    &=
    7\cdot 5\cdot 3\cdot \int (0)^* \cup_8 (0)^*
    = 1 \pmod{2}~,
\end{eqs}
where in the last few lines we recursively apply the same argument to reduce the dimension step by step. Therefore,
\begin{equation}
    \mu_{\mathrm{volume}}
    =
    (-1)^{n_0}
    = -1~.
\end{equation}
\end{proof}

%%%%%%%%%%%%%%%%%%%%%%%%%%%%%%%%%%%%%%%%%%%%%%%%%%%%%%
\subsection{Fermionic-volume toric code in $(d{+}1)$D}
\label{sec: Fermionic-volume toric code in (d+1)D}

We now generalize the above construction to $(d{+}1)$ dimensions. We place a $\ZZ_4$ qudit on each $(d{-}4)$-cell of an arbitrary triangulation (or cellulation) of the spatial manifold $M_d$. 
As in Sec.~\ref{sec: Fermionic-loop toric codes in arbitrary dimensions}, Eq.~\eqref{eq: H_condensed of fermionic-volume TC in (8+1)D} extends directly to arbitrary $d$, yielding a commuting-projector Hamiltonian
\begin{equation}
    H_{\mathrm{condensed}}
    = -\sum_{c_{d-5}} G_{c_{d-5}}
      \;-\sum_{c_{d-3}} B_{c_{d-3}}^{2}
      \;-\sum_{c_{d-4}} \tilde{C}_{c_{d-4}}
      \;+\; \mathrm{h.c.}~,
\end{equation}
with stabilizers
\begin{eqs}
    G_{c_{d-5}}
    &:= X_{\delta \boldsymbol{c}_{d-5}}\;
    \prod_{c_{d-4}'} Z_{c_{d-4}'}^{\int \delta \boldsymbol{c}_{d-5} \cup_{d-8} \boldsymbol{c}_{d-4}'}~, \\
    B_{c_{d-3}}
    &:= Z_{\partial c_{d-3}}
    = \prod_{c_{d-4}} Z_{c_{d-4}}^{\,\boldsymbol{c}_{d-4}(\partial c_{d-3})}~, \\
    \tilde{C}_{c_{d-4}}
    &:= X_{c_{d-4}}^{\,2}\;
    \prod_{c_{d-4}'} Z_{c_{d-4}'}^{\,2\int \boldsymbol{c}_{d-4}' \cup_{d-8} \boldsymbol{c}_{d-4}
    + \boldsymbol{c}_{d-4} \cup_{d-7} \delta \boldsymbol{c}_{d-4}'}~.
\end{eqs}

The model supports two excitation operators on a $(d{-}4)$-cell:
\begin{eqs}
    V^C_{c_{d-4}} &:= Z_{c_{d-4}}^2, \\
    \tilde{V}^F_{c_{d-4}}
    &:= X_{c_{d-4}}\prod_{c_{d-4}'}
    Z_{c_{d-4}'}^{\int \boldsymbol{c}_{d-4}' \cup_{d-8} \boldsymbol{c}_{d-4}
    - (-1)^d \delta \boldsymbol{c}_{d-4}' \cup_{d-7} \boldsymbol{c}_{d-4}}~.
    \label{eq: modified flux volume in (d+1)D fermionic-volume toric code}
\end{eqs}
The operator $V^C_{c_{d-4}}$ violates the $G_{c_{d-5}}$ stabilizers along $\partial c_{d-4}$, creating a bosonic charge-$(d{-}5)$-brane excitation. 
In contrast, $\tilde{V}^F_{c_{d-4}}$ commutes with $G_{c_{d-5}}$ and $\tilde{C}_{c_{d-4}}$, but violates $B_{c_{d-3}}^2$ on the coboundary $\delta \boldsymbol{c}_{d-4}$, thereby creating a flux volume excitation.
Moreover, the flux volume excitation has $\ZZ_2$ fusion:
\begin{equation}
    \bigl(\tilde{V}^F_{c_{d-4}}\bigr)^{2}
    =
    \tilde{C}_{c_{d-4}}
    \prod_{c_{d-3}} B_{c_{d-3}}^{\,2\int \boldsymbol{c}_{d-3} \cup_{d-7} \boldsymbol{c}_{d-4}}~,
\end{equation}
is a product of stabilizers.

We now verify that this flux volume has fermionic volume statistics.
\begin{theorem}
The flux volume created by $\tilde{V}^F_{c_{d-4}}$ in Eq.~\eqref{eq: modified flux volume in (d+1)D fermionic-volume toric code} has fermionic volume statistics:
\begin{equation}
    \mu_{\mathrm{volume}}
    :=
    \prod_{\sigma_3 \sqcup \tau_4 = \{2,3,4,5,6,7,8\}}
    \bigl[\tilde{V}^F_{(01\sigma_3)^*},\, \tilde{V}^F_{(0\tau_4)^*}\bigr]^{2}
    = -1~,
\end{equation}
where $\sigma_3$ and $\tau_4$ are disjoint index sets with $|\sigma_3|=3$ and $|\tau_4|=4$, and $(0ijkl)^*$ denotes the $(d{-}4)$-cochain Poincar\'e dual to the 4-simplex $\langle 0ijkl\rangle$.
\end{theorem}
\begin{proof}
We extend the definition of the operator to an arbitrary $(d{-}4)$-cochain $\lambda\in C^{d-4}(M_d,\ZZ)$ by
\begin{equation}
    \tilde{V}^F_\lambda
    :=
    X_\lambda
    \prod_{c_{d-4}'}
    Z_{c_{d-4}'}^{\int \boldsymbol{c}_{d-4}'\cup_{d-8}\lambda
        -(-1)^d \delta\boldsymbol{c}_{d-4}'\cup_{d-7} \lambda}~.
\end{equation}
A direct computation shows that its commutator square is
\begin{equation}
    \bigl([\tilde{V}^F_{\lambda},\tilde{V}^F_{\lambda'}]\bigr)^2
    =
    (-1)^{\int \delta\lambda\cup_{d-6}\delta\lambda'}~.
\label{eq: (d+1)D commutator of V}
\end{equation}
The volume statistic is then defined by the same product of commutators as in the $(8{+}1)$D case,
\begin{eqs}
    \mu_{\mathrm{volume}}
    &=
    {\prod}'_{\sigma_4 \sqcup \tau_4 = \{1,2,3,4,5,6,7,8\}}
    \bigl[\tilde{V}^F_{(0\sigma_4)^*},\, \tilde{V}^F_{(0\tau_4)^*}\bigr]^{2}  \\
    &=
    {\prod}'_{\sigma_4 \sqcup \tau_4 = \{1,2,3,4,5,6,7,8\}}
    (-1)^{\int \delta (0\sigma_4)^* \cup_{d-6} \delta (0\tau_4)^*}~,
\end{eqs}
where the primed product indicates that each unordered partition $\{\sigma,\tau\}$ is counted only once.

The remaining simplification is identical to the computation in Sec.~\ref{sec: Fermionic-volume toric code in (8+1)D}, with all higher cup products shifted by $d{-}8$ (e.g., $\cup_2$ is replaced by $\cup_{d-6}$ throughout). 
Carrying out the same cancellations yields
\begin{equation}
    \mu_{\mathrm{volume}}
    =
    (-1)^{\int (0)^* \cup_{d} (0)^*}
    = -1~,
\end{equation}
which confirms the fermionic volume statistics.
\end{proof}

\section{Conclusion and Outlook}

In this work, we developed a Pauli stabilizer formalism for realizing higher-dimensional TQFTs as stabilizer Hamiltonians and for computing the generalized statistics of their extended excitations directly on the lattice. In particular, we constructed explicit Pauli stabilizer models for a broad class of higher-form gauge theories and showed how their (extended) excitations have nontrivial statistics, via microscopic unitary processes within the stabilizer setting.

Our main example is the fermionic-loop toric code: in $(4{+}1)$D it arises from condensing the $e^2m^2$ loop in the loop-only $\ZZ_4$ toric code, and the resulting flux loop exhibits fermionic loop statistics detected by the loop-flipping process. We further constructed Pauli stabilizer realizations of all $(4{+}1)$D Dijkgraaf--Witten $2$-form gauge theories classified by $H^5(B^2G,U(1))$ via systematic loop condensations, and we generalized the fermionic-loop construction to arbitrary dimensions using algebraic-topology notation (including higher cup products). 
Beyond loops, our framework led to new exactly solvable higher-dimensional models with nontrivial statistics of higher-dimensional excitations, including anyonic/semionic membrane statistics in $(6{+}1)$D and fermionic membrane and volume statistics in higher dimensions. 

{\change
From the viewpoint of quantum information, the stabilizer realizations constructed here provide more than microscopic representatives of continuum TQFTs. 
They also give explicit higher-dimensional quantum codes whose stabilizers, excitation operators, and generalized statistics are all algebraically tractable. 
In such models, extended excitations play the role of syndrome defects, while their creation and transport operators give natural candidates for logical operators and fault-tolerant processes. 
Moreover, because our construction is based on condensation and higher-form gauge structure, it suggests systematic ways to engineer boundaries, defects, and domain walls, which may lead to new logical gate implementations or improved memory properties. 
The present work focuses on the bulk topological orders and their generalized statistics; developing concrete decoding algorithms, finite-size code families, and fault-tolerant protocols based on these models is an important direction for future work.
}

A central theme emerging from these results is the distinction between \emph{Pauli} and \emph{non-Pauli} statistics. In our formalism, Pauli realizations impose that relevant excitation operators commute up to $U(1)$ phases, greatly simplifying statistical processes and making the invariants efficiently computable. At the same time, certain stable statistics are known to be intrinsically non-Pauli. A prominent example is the $\ZZ_3$ Pontryagin component of membrane statistics, which requires at least Clifford structure and cannot be captured by Pauli operators alone~\cite{feng2025anyonic}.

This motivates several concrete directions for future work:
\begin{enumerate}
\item \textbf{Relating Pauli and non-Pauli statistics~\cite{FHH21, CH21, kobayashi2024generalized, xue2025statisticsabeliantopologicalexcitations}.}
Our Pauli theory defines a subgroup of realizations and yields a corresponding notion of Pauli statistics.
As discussed in Appendix~\ref{app: Statistics in non-Pauli models}, there is a canonical map from Pauli statistics to the general (non-Pauli) statistics; it is not surjective in general, and its injectivity remains open.
Guided by all examples we have computed, we conjecture that Pauli statistics embeds into non-Pauli statistics, i.e., Pauli statistics forms a subgroup of the full classification.

\item \textbf{Quadratic representatives and Pauli realizability~\cite{Eilenberg:1946, Eilenberg1953OnTG, Eilenberg1954OnTG, Wang2020InAbelian}.}
We further conjecture that Pauli-realizable statistics correspond to cohomology classes that admit \emph{quadratic} representatives in the gauge fields (at the cochain level), compatible with phase factors accessible in Pauli/qudit Pauli algebras.
In other words, a necessary (and plausibly close to sufficient) condition for Pauli realizability is that the relevant topological term can be written as a quadratic expression, schematically of the form ``$b \cup_i b$'' or ``$b \cup_i \delta b$'', possibly after an appropriate group extension. Developing a precise algebraic criterion for when a cocycle class admits such a quadratic representative would provide a conceptual classification of Pauli statistics.

\item \textbf{Beyond Pauli: Clifford and higher-level realizations~\cite{Breuckmann2024CupsandGates, Hsin2025NonAbelian, Kobayashi2025CliffordHierarchy, Sakura2025TransversalCliffordHierarchy}.}
A natural extension is to enlarge the operator set from Pauli to Clifford (or higher levels of the Clifford hierarchy), with the goal of realizing the missing non-Pauli invariants such as Pontryagin statistics. Understanding how the statistical processes change when commutators are no longer pure phases, and how these processes are organized by higher categorical or homotopy-theoretic data, would bridge stabilizer models with more general exactly solvable lattice constructions.

\item \textbf{Boundaries, defects, and code design~\cite{Fowler2012Surfacecodes, Bravyi2024HighThreshold, Barkeshli2023Codimension, Barkeshli2024higher, Barkeshli2024higherfermion,hsin2025classifyinglogicalgatesquantum, Liang2024Extracting, liang2024operator, liang2025generalized, liang2025planar, Ruba:2025faz}.}
Finally, it would be valuable to extend the present constructions to include systematic boundary and defect engineering in higher dimensions, where different choices of condensations and stabilizer generators may translate into distinct fault-tolerant gate sets or improved memory properties. The explicit microscopic control provided by stabilizer realizations makes them a promising arena for exploring such code-design questions.

\item \textbf{Connections to QCAs and lattice anomalies~\cite{Gross2012GNVWindex, Else2014Classifying, Freedman2020ClassificationQCA, Haah2021CliffordQCA, Shirley2022QCA, haah_QCA_23, haah2025topological, yang2025categorifying, feng2025higherformanomalieslattices, feng2025onsiteabilityhigherformsymmetries,
kapustin2025higher, kapustin2025higher2, sun2025cliffordQCA, Thorngren2025latticeanomalies}.}
Our lattice-based viewpoint suggests a close relationship between generalized statistics, lattice anomalies, and locality-preserving dynamics (including Clifford QCAs). Clarifying this relationship---for example, identifying which anomaly classes admit Pauli (or Clifford) representatives and how they compose under stacking and condensation---may offer a unified description of invertible phases, QCAs, and their associated algebraic invariants.

\end{enumerate}

We expect that a clearer understanding of \emph{Pauli-realizable} statistics—for example, in terms of when cocycle classes admit quadratic cochain representatives—may help clarify what can and cannot be captured within the stabilizer paradigm. Such a perspective could offer a more systematic way to identify higher-dimensional topological orders and generalized statistics that admit Pauli stabilizer realizations, while also highlighting situations where genuinely non-Pauli invariants (such as Pontryagin-type statistics) may require broader operator frameworks, e.g., Clifford or higher-level constructions. More broadly, we hope that developing these connections will guide future work on lattice realizations of higher-dimensional TQFTs and their excitation structures, and may eventually inform the design of more efficient self-correcting quantum memories.

\section*{Acknowledgement}

We are grateful to Andreas Bauer, Tyler D. Ellison, Dominic Else, Anton Kapustin, Kyle Kawagoe, Sahand Seifnashri, Wilbur Shirley, Nikita Sopenko, Nathanan Tantivasadakarn, Bowen Yang, Peng Ye, and Carolyn Zhang for valuable discussions and feedback.

Y.-A.C. is supported by the National Natural Science Foundation of China (Grant No.~12474491), and the Fundamental Research Funds for the Central Universities, Peking University.
P.-S.H. is supported by the Department of Mathematics, King's College London.
R.K. is supported by the U.S. Department of Energy, Office of Science, Office of High Energy Physics under Award Number DE-SC0009988 and by the Sivian Fund at the Institute for Advanced Study.

\appendix

{\change
%%%%%%%%%%%%%%%%%%%%%%%%%%%%%%%%%%%%%%%%%%%%%%%%%%%%%%%%%%%%%%%%%%%%%%%%%%%%%
%%%%%%%%%%%%%%%%%%%%%%%%%%%%%%%%%%%%%%%%%%%%%%%%%%%%%%%%%%%%%%%%%%%%%%%%%%%%%
\section{Cochain and higher-cup-product conventions}
\label{app:cochain-conventions}

In this appendix, we collect the cochain conventions used in the main text.
The goal is to make the notation self-contained for readers who want to follow
the stabilizer Hamiltonians and the statistics calculations without referring
to a separate algebraic-topology review.

We work with cochains on a triangulated or cellulated manifold $M$. For explicit
formulas for higher cup products, one may take $M$ to be a branched
triangulation. On hypercubic lattices we use the corresponding cubical cup
products; the identities used in the main text are the same. Unless otherwise
stated, all cochains are written additively.

%%%%%%%%%%%%%%%%%%%%%%%%%%%%%%%%%%%%%%%%%%%%%%%%%%%%%%%%%%%%%%%%%%%%%%%%%%%%%
\subsection{Chains, cochains, and integration}

Let $C_p(M,A)$ and $C^p(M,A)$ denote the groups of $p$-chains and $p$-cochains
with coefficients in an Abelian group $A$. In this paper, $A$ is usually $\ZZ$,
$\ZZ_N$, or $\RR/\ZZ$. A $p$-cochain $a\in C^p(M,A)$ is a linear function
\begin{equation}
    a:\; C_p(M,\ZZ)\longrightarrow A~.
\end{equation}
Thus $a(\sigma)$ is the value of $a$ on the oriented $p$-simplex or $p$-cell
$\sigma$.

For an oriented simplex
\begin{equation}
    \sigma=\langle v_0v_1\cdots v_p\rangle,
\end{equation}
the boundary is
\begin{equation}
    \partial\sigma
    =
    \sum_{i=0}^p
    (-1)^i
    \langle v_0\cdots \widehat{v_i}\cdots v_p\rangle,
\end{equation}
where $\widehat{v_i}$ means that the vertex $v_i$ is omitted. The coboundary
operator
\begin{equation}
    \delta:\; C^p(M,A)\longrightarrow C^{p+1}(M,A)
\end{equation}
is defined by duality:
\begin{equation}
    (\delta a)(\sigma)=a(\partial\sigma),
    \qquad
    a\in C^p(M,A),
    \qquad
    \sigma\in C_{p+1}(M,\ZZ)~.
\end{equation}
Therefore $\delta^2=0$, because $\partial^2=0$. A cochain satisfying
$\delta a=0$ is called a cocycle.

We use bold symbols for basis cochains associated with cells. If $c$ is an
oriented $p$-cell, then $\mathbf c\in C^p(M,\ZZ)$ denotes the cochain
\begin{equation}
    \mathbf c(c')
    =
    \begin{cases}
        +1, & c'=c \text{ with the same orientation},\\
        -1, & c'=c \text{ with the opposite orientation},\\
        0,  & c'\neq c~.
    \end{cases}
\end{equation}
Over $\ZZ_2$, the signs are ignored. In the main text, we sometimes suppress
the boldface when no confusion can arise; for example, $\delta f$ may mean
$\delta\mathbf f$ when $f$ is being used as a basis cochain.

For a top-degree cochain $a\in C^d(M,A)$ on a $d$-dimensional manifold $M$, we
write
\begin{equation}
    \int_M a
    :=
    \sum_{\sigma^d\subset M} a(\sigma^d),
\end{equation}
where the sum is over all oriented $d$-simplices or $d$-cells. If the domain of
integration is omitted, it is understood to be the ambient closed manifold.
With these conventions, Stokes' theorem is
\begin{equation}
    \int_M \delta a
    =
    \int_{\partial M} a~.
\end{equation}
In particular,
\begin{equation}
    \int_M \delta a =0
\end{equation}
when $M$ is closed.

%%%%%%%%%%%%%%%%%%%%%%%%%%%%%%%%%%%%%%%%%%%%%%%%%%%%%%%%%%%%%%%%%%%%%%%%%%%%%
\subsection{The ordinary cup product}

The ordinary cup product is a cochain-level product. For coefficients in a
commutative ring $R$, such as $\ZZ$ or $\ZZ_N$, it is a map
\begin{equation}
    \cup:\; C^p(M,R)\times C^q(M,R)
    \longrightarrow
    C^{p+q}(M,R)~.
\end{equation}
On an ordered simplex it is defined by the Alexander--Whitney formula
\begin{equation}
    (a\cup b)(\langle 0\,1\,\cdots\,p+q\rangle)
    =
    a(\langle 0\,1\,\cdots\,p\rangle)
    b(\langle p\,p+1\,\cdots\,p+q\rangle),
\end{equation}
for $a\in C^p(M,R)$ and $b\in C^q(M,R)$. For example, if
$a,b\in C^1(M,R)$, then
\begin{equation}
    (a\cup b)(012)=a(01)b(12)~.
\end{equation}

The coboundary satisfies the Leibniz rule
\begin{equation}
    \delta(a\cup b)
    =
    \delta a\cup b
    +
    (-1)^p a\cup\delta b,
    \qquad
    a\in C^p(M,R)~.
\end{equation}

A key point is that the ordinary cup product is not graded-commutative at the
cochain level. Even if $a$ and $b$ are cocycles, the cochains $a\cup b$ and
$(-1)^{pq}b\cup a$ are generally not equal. They become cohomologous only after
adding a coboundary. Higher cup products provide explicit cochains whose
coboundaries measure these failures of commutativity. This is why they naturally
appear in the statistics calculations in the main text.

In the $\RR/\ZZ$-valued actions used in the main text, we usually compute cup
products using integer lifts and then multiply the resulting integer-valued
cochain by a rational coefficient, interpreting the final answer in $\RR/\ZZ$.

%%%%%%%%%%%%%%%%%%%%%%%%%%%%%%%%%%%%%%%%%%%%%%%%%%%%%%%%%%%%%%%%%%%%%%%%%%%%%
\subsection{Higher cup products}

For $i\geq 0$, Steenrod's higher cup products are operations
\begin{equation}
    \cup_i:\; C^p(M,R)\times C^q(M,R)
    \longrightarrow
    C^{p+q-i}(M,R),
    \qquad
    \cup_0=\cup~.
\end{equation}
We set $\cup_i=0$ for $i<0$. The convention used throughout the paper is fixed
by the identity
\begin{align}
    \delta(a\cup_i b)
    &=
    \delta a\cup_i b
    +
    (-1)^p a\cup_i \delta b
    +
    (-1)^{p+q+i}a\cup_{i-1} b
    +
    (-1)^{pq+p+q}b\cup_{i-1}a~,
\label{eq:app-higher-cup-leibniz}
\end{align}
where $a\in C^p(M,R)$ and $b\in C^q(M,R)$.

For $\ZZ_2$-valued cochains, all signs in
Eq.~\eqref{eq:app-higher-cup-leibniz} may be dropped:
\begin{equation}
    \delta(a\cup_i b)
    =
    \delta a\cup_i b
    +
    a\cup_i\delta b
    +
    a\cup_{i-1}b
    +
    b\cup_{i-1}a
    \qquad
    \mathrm{mod}\;2~.
\end{equation}
Thus, when $a$ and $b$ are $\ZZ_2$ cocycles,
\begin{equation}
    a\cup b+b\cup a
    =
    \delta(a\cup_1 b)~.
\end{equation}
So $\cup_1$ is the cochain that witnesses commutativity of the cup product up
to a coboundary. Similarly, $\cup_2$ controls the corresponding failure for
$\cup_1$, and so on. This hierarchy is the sense in which higher cup products
encode successive ``higher commutativity'' data.

For readers who want a concrete formula, the $\ZZ_2$ cup-$1$ product is
\begin{align}
    &(a\cup_1 b)(\langle 0,\ldots,p+q-1\rangle)
    \notag\\
    &\qquad =
    \sum_{j=0}^{p-1}
    a(\langle 0,\ldots,j,q+j,\ldots,p+q-1\rangle)
    b(\langle j,\ldots,q+j\rangle),
\end{align}
for $a\in C^p(M,\ZZ_2)$ and $b\in C^q(M,\ZZ_2)$. For example, if
$a,b\in C^2(M,\ZZ_2)$, then
\begin{equation}
    (a\cup_1 b)(0123)
    =
    a(023)b(012)+a(013)b(123)~.
\end{equation}
The full explicit formula for $\cup_i$ with $i\geq 2$ is not needed in this
paper; all computations use the coboundary identity
Eq.~\eqref{eq:app-higher-cup-leibniz} together with locality of supports.

Several special cases of Eq.~\eqref{eq:app-higher-cup-leibniz} are used
repeatedly. For $A_2,B_2\in C^2(M,\ZZ)$,
\begin{equation}
    \delta(A_2\cup_1 B_2)
    =
    \delta A_2\cup_1 B_2
    +
    A_2\cup_1\delta B_2
    -
    A_2\cup B_2
    +
    B_2\cup A_2~.
\end{equation}
For $A_3,B_3\in C^3(M,\ZZ)$,
\begin{equation}
    \delta(A_3\cup_1 B_3)
    =
    \delta A_3\cup_1 B_3
    -
    A_3\cup_1\delta B_3
    -
    A_3\cup B_3
    -
    B_3\cup A_3~,
\end{equation}
and
\begin{equation}
    \delta(A_3\cup_3 B_3)
    =
    \delta A_3\cup_3 B_3
    -
    A_3\cup_3\delta B_3
    -
    A_3\cup_2 B_3
    -
    B_3\cup_2 A_3~.
\end{equation}
More generally, if $A_{d-2}\in C^{d-2}(M,\ZZ)$ and
$B_{d-1}\in C^{d-1}(M,\ZZ)$, then
\begin{align}
    \delta(A_{d-2}\cup_{d-2} B_{d-1})
    &=
    \delta A_{d-2}\cup_{d-2} B_{d-1}
    +
    (-1)^d A_{d-2}\cup_{d-2}\delta B_{d-1}
    \notag\\
    &\quad
    -
    (-1)^d A_{d-2}\cup_{d-3}B_{d-1}
    -
    B_{d-1}\cup_{d-3}A_{d-2}~.
\end{align}

The higher cup products are local: if the supports of two cochains do not meet
inside the simplex or cell on which the product is evaluated, then their higher
cup product evaluates to zero. This locality property is used throughout the
statistics calculations, for example when cup products between Poincar\'e-dual
cochains supported on disjoint simplices are set to zero.

Finally, for a $\ZZ_2$ cocycle $B_n\in Z^n(M,\ZZ_2)$, the Steenrod square has
the standard cochain representative
\begin{equation}
    \mathrm{Sq}^r(B_n)
    =
    B_n\cup_{n-r}B_n~.
\end{equation}
The stable operations $\mathrm{Sq}^2\mathrm{Sq}^1$ and
$\mathrm{Sq}^4\mathrm{Sq}^1$ used in the main text are reviewed in
Appendix~\ref{app: Calculation of stable Z2 operations}.

%%%%%%%%%%%%%%%%%%%%%%%%%%%%%%%%%%%%%%%%%%%%%%%%%%%%%%%%%%%%%%%%%%%%%%%%%%%%%
\subsection{Coefficient conventions and Bockstein cochains}

When a cochain is $\ZZ_N$-valued, we often choose an integer lift and use the
same symbol for the lifted cochain when no confusion can arise. More explicitly,
if
\begin{equation}
    b\in C^p(M,\ZZ_N),
\end{equation}
we choose
\begin{equation}
    \widetilde b\in C^p(M,\ZZ),
    \qquad
    \widetilde b\equiv b \pmod N~.
\end{equation}
Expressions such as
\begin{equation}
    \frac{1}{N}\int a\cup\delta b,
    \qquad
    \frac{1}{4}\int b\cup\delta b,
\end{equation}
are computed using integer lifts and then interpreted in $\RR/\ZZ$. A different
choice of lift may change the cochain representative, but not the topological
action in the applications considered in the main text.

If $b\in C^p(M,\ZZ_N)$ is closed modulo $N$, then $\delta\widetilde b$ is
divisible by $N$. We define
\begin{equation}
    \frac{\delta b}{N}
    :=
    \frac{\delta\widetilde b}{N}~.
\end{equation}
The right-hand side is first an integer-valued $(p{+}1)$-cochain. When needed,
we reduce it modulo $N$ and regard it as a $\ZZ_N$-valued Bockstein cochain:
\begin{equation}
    \mathrm{Bock}(b)
    :=
    \frac{\delta\widetilde b}{N}
    \mod N
    \in C^{p+1}(M,\ZZ_N)~.
\end{equation}
In the main text we also write
\begin{equation}
    \frac{\delta b}{N}
    :=
    \mathrm{Bock}(b)~.
\end{equation}
For instance, the cocycles for $(4{+}1)$D Dijkgraaf--Witten 2-form gauge
theories contain terms of the form
\begin{equation}
    b_j\cup \frac{\delta b_k}{N_k}~.
\end{equation}

For a $\ZZ_2$ cocycle $B$, the operation $\mathrm{Sq}^1B$ is the Bockstein
associated with
\begin{equation}
    0\longrightarrow \ZZ_2
    \longrightarrow \ZZ_4
    \longrightarrow \ZZ_2
    \longrightarrow 0~.
\end{equation}
Equivalently, after choosing a $\ZZ_4$ lift $\widetilde B$ of $B$, it can be
represented as the mod-$2$ cochain
\begin{equation}
    \mathrm{Sq}^1B
    =
    \frac{\delta\widetilde B}{2}
    \mod 2~.
\end{equation}

%%%%%%%%%%%%%%%%%%%%%%%%%%%%%%%%%%%%%%%%%%%%%%%%%%%%%%%%%%%%%%%%%%%%%%%%%%%%%
\subsection{Poincar\'e-dual cochains}

Let $M$ be a $d$-dimensional oriented triangulated manifold. If $\sigma$ is a
$p$-simplex or $p$-chain embedded in $M$, we denote its Poincar\'e dual by
\begin{equation}
    \sigma^*\in C^{d-p}(M,\ZZ)~.
\end{equation}
We choose the convention
\begin{equation}
    \delta(\sigma^*)=(\partial\sigma)^*~.
\end{equation}
This equation fixes our orientation convention for Poincar\'e-dual cochains.
Other sign conventions differ by systematic orientation signs, but all
computations in the main text use the convention above.

For the oriented triangle $(012)$,
\begin{equation}
    \delta(012)^*
    =
    (\partial 012)^*
    =
    (12)^*-(02)^*+(01)^*~.
\end{equation}
Equivalently, if $\nu_{ij}:=(ij)^*$, then
\begin{equation}
    \delta(012)^*
    =
    \nu_{12}-\nu_{02}+\nu_{01}~.
\end{equation}
Over $\ZZ_2$, the signs are omitted.

In the statistics calculations, we often write a simplex by listing its vertices
without commas. For example, $012$ means the oriented simplex $\langle 0\,1\,2\rangle$.
If $\sigma=\{a_0<a_1<\cdots<a_p\}$ is an ordered set of vertices, then $0\sigma$
means the simplex obtained by adjoining the vertex $0$ and then using the
induced order. We also use
\begin{equation}
    \partial_i\sigma
    :=
    \{a_0,\ldots,\widehat{a_i},\ldots,a_p\}
\end{equation}
when expanding boundaries.

The following elementary consequence is used several times. If $p_0=(0)^*$ is
the Poincar\'e dual of a vertex in a $d$-dimensional manifold, then $p_0$ is a
$d$-cochain, and
\begin{equation}
    p_0\cup_d p_0=p_0
    \qquad
    \mathrm{mod}\;2~.
\end{equation}
Therefore
\begin{equation}
    \int p_0\cup_d p_0=1
    \qquad
    \mathrm{mod}\;2
\end{equation}
when the vertex $0$ is included with unit orientation. This is the final step in
several fermionic statistics computations in the main text.

%%%%%%%%%%%%%%%%%%%%%%%%%%%%%%%%%%%%%%%%%%%%%%%%%%%%%%%%%%%%%%%%%%%%%%%%%%%%%
\subsection{Cochain notation for Pauli products}

We now explain the shorthand used for Pauli stabilizer operators. Suppose an
$N$-state qudit is placed on each $p$-cell $c$ of $M$, with generalized Pauli
operators $X_c$ and $Z_c$ satisfying
\begin{equation}
    Z_cX_c
    =
    \exp\left(\frac{2\pi i}{N}\right)
    X_cZ_c~.
\end{equation}
For an integer-valued or $\ZZ_N$-valued $p$-cochain $\lambda$, define
\begin{equation}
    X_\lambda
    :=
    \prod_{c\in M_p}X_c^{\lambda(c)},
    \qquad
    Z_\lambda
    :=
    \prod_{c\in M_p}Z_c^{\lambda(c)}~,
\end{equation}
where $M_p$ denotes the set of $p$-cells. Exponents are understood modulo $N$
for $N$-state qudits. Negative exponents mean inverses; for example,
$Z_c^{-1}=Z_c^\dagger$.

With these conventions, if $\lambda,\mu\in C^p(M,\ZZ_N)$, then
\begin{equation}
    Z_\mu X_\lambda
    =
    \exp\left(
        \frac{2\pi i}{N}
        \sum_{c\in M_p}\mu(c)\lambda(c)
    \right)
    X_\lambda Z_\mu~.
\end{equation}
Equivalently, defining
\begin{equation}
    \langle \mu,\lambda\rangle
    :=
    \sum_{c\in M_p}\mu(c)\lambda(c)
    \in \ZZ_N~,
\end{equation}
we have
\begin{equation}
    Z_\mu X_\lambda
    =
    \exp\left(
        \frac{2\pi i}{N}\langle\mu,\lambda\rangle
    \right)
    X_\lambda Z_\mu~.
\end{equation}

This notation is also used when the exponent is naturally a chain. The chain is
first paired with the corresponding basis cochains. For example, if qudits live
on faces $f$, then
\begin{equation}
    X_{\delta e}
    =
    \prod_f X_f^{\delta e(f)},
    \qquad
    Z_{\partial t}
    =
    \prod_f Z_f^{f(\partial t)}~.
\end{equation}
This is the convention behind the compact stabilizer notation
\begin{equation}
    A_e=X_{\delta e},
    \qquad
    B_t=Z_{\partial t}~.
\end{equation}

More general Pauli operators in the main text contain cup-product exponents. A
typical example is
\begin{equation}
    X_\lambda
    \prod_{c'}
    Z_{c'}^{\int c'\cup_i\lambda}~,
\end{equation}
where $\lambda$ is a cochain of the same degree as the cells carrying the
qudits. The number
\begin{equation}
    \int c'\cup_i\lambda
\end{equation}
is an integer obtained by evaluating the top-degree cochain $c'\cup_i\lambda$
on the ambient manifold. Thus it gives the exponent of the Pauli operator
$Z_{c'}$. This notation is what allows the Hamiltonians and excitation operators
in the main text to be written in a dimension-independent form.
}

%%%%%%%%%%%%%%%%%%%%%%%%%%%%%%%%%%%%%%%%%%%%%%%%%%%%%%%%%%%%%%%%%%%%%%
\section{Calculation of stable $\ZZ_2$ operations}\label{app: Calculation of stable Z2 operations}

In this appendix, we explicitly calculate the cochain-level expressions of stable cohomology operations $\Sq^2\Sq^1$ and $\Sq^4\Sq^1$. This result is used in the main text to prove that certain flux anomalies are stable in higher dimensions. Also, we relate them to Stiefel-Whitney classes on the manifold. 

We begin by reviewing the cochain-level representation of Steenrod square operation. Given two integers $k,n\geq 0$ and a $\ZZ_2$-cocycle $B_n$ of degree $n$, the Steenrod square is represented by
\begin{eqs}
    \Sq^k(B_n)=B_n\cup_{n-k} B_n~.\label{sq k Bn}
\end{eqs}
In particular, when $k=1$ we obtain the Bockstein homomorphism
\begin{eqs}
    \Sq^1(B_n)=\delta \left(\frac{B_n}{2}\right)~,\label{sq 1 Bn}
\end{eqs}
where $B_n/2$ represents an arbitrary lift of $B_n$ to $\ZZ_4$-cocycle. Combining Eqs.~\eqref{sq k Bn} and \eqref{sq 1 Bn}, we obtain
\begin{eqs}
    \Sq^2\Sq^1(B_n)=\delta \left(\frac{B_n}{2}\right)\cup_{n-1}\delta \left(\frac{B_n}{2}\right)~.
\end{eqs}
This expression can be further simplified using the following identity of higher cup product,
\begin{eqs}
    \delta(a_n\cup_k b_m)=\delta a_n\cup_k b_m+(-1)^n a_n\cup_k\delta b_m+(-1)^{n+m+k} a_n\cup_{k-1} b_m+(-1)^{nm+n+m} b_m\cup_{k-1} a_n~.\label{higher cup id}
\end{eqs}
Following Eq.~\eqref{higher cup id}, we can deduce two identities
\begin{eqs}
    \delta\left(\frac{B_n}{2}\cup_{n-1}\delta\left(\frac{B_n}{2}\right)\right)&=\delta \left(\frac{B_n}{2}\right)\cup_{n-1}\delta \left(\frac{B_n}{2}\right)+(-1)^n \frac{B_n}{2}\cup_{n-2}\delta\left(\frac{B_n}{2}\right)-\delta\left(\frac{B_n}{2}\right)\cup_{n-2}\frac{B_n}{2}~\label{identity 1}
\end{eqs}
and
\begin{eqs}
    \delta\left(\frac{B_n}{2}\cup_{n-2}\frac{B_n}{2}\right)=\delta\left(\frac{B_n}{2}\right)\cup_{n-2}\frac{B_n}{2}+(-1)^n \frac{B_n}{2}\cup_{n-2}\delta\left(\frac{B_n}{2}\right)+2\cdot (-1)^n \frac{B_n}{2}\cup_{n-3}\frac{B_n}{2}~.\label{identity 2}
\end{eqs}
Now, adding Eqs.~\eqref{identity 1} and \eqref{identity 2} together, we obtain that
\begin{eqs}
    \delta \left(\frac{B_n}{2}\right)\cup_{n-1}\delta \left(\frac{B_n}{2}\right)+2\cdot (-1)^n\left(\frac{B_n}{2}\cup_{n-3}\frac{B_n}{2}+\frac{B_n}{2}\cup_{n-2}\delta\left(\frac{B_n}{2}\right)\right)~
\end{eqs}
is a full differential, and thus represents trivial cohomology class. Therefore, we have proved the equivalence
\begin{eqs}
    \frac{1}{2}\Sq^2\Sq^1(B_n)\sim (-1)^{n+1}\left(\frac{B_n}{2}\cup_{n-3}\frac{B_n}{2}+\frac{B_n}{2}\cup_{n-2}\delta\left(\frac{B_n}{2}\right)\right)~.\label{z4 rep of sq2sq1}
\end{eqs}
In the case where $n=2$, the first term of RHS vanishes, and Eq.~\eqref{z4 rep of sq2sq1} simplifies to
\begin{eqs}
    \frac{1}{2}\Sq^2\Sq^1(B_2)\sim -\frac{B_2}{2}\cup\delta\left(\frac{B_2}{2}\right)~
\end{eqs}
as expected. 

The calculation process for $\Sq^4\Sq^1$ is analogous. In this case, we have the following cochain-level expression of $\Sq^4\Sq^1$,
\begin{eqs}
    \Sq^4\Sq^1(B_n)=\delta \left(\frac{B_n}{2}\right)\cup_{n-3}\delta \left(\frac{B_n}{2}\right)~.
\end{eqs}
Note that $(n-3)$ has the same parity as $(n-1)$, we just need to lower the degree of the cup products by $2$. Finally, we obtain
\begin{eqs}
    \frac{1}{2}\Sq^4\Sq^1(B_n)\sim (-1)^{n+1}\left(\frac{B_n}{2}\cup_{n-5}\frac{B_n}{2}+\frac{B_n}{2}\cup_{n-4}\delta\left(\frac{B_n}{2}\right)\right)~,
\end{eqs}
which simplifies to 
\begin{eqs}
    \frac{1}{2}\Sq^4\Sq^1(B_4)\sim -\frac{B_4}{2}\cup\delta\left(\frac{B_4}{2}\right)~
\end{eqs}
when $n=4$. 

\paragraph*{Relation to Stiefel-Whitney Classes.}
We can also derive equivalent expressions of Steenrod squares using characteristic classes. To this end, we make use of the Wu relation
\begin{eqs}
    w=\Sq(\nu)~,
\end{eqs}
together with the Cartan formula
\begin{eqs}
    \Sq^i(x\cup y)=\sum_j \Sq^{i-j}(x)\cup \Sq^j(y)~,
\end{eqs}
where $\nu$ denotes the total Wu class, and $w$ denotes the total Stiefel-Whitney class. Now, on a manifold of dimension $(n+3)$, we obtain
\begin{eqs}
    \Sq^2\Sq^1(B_n)&=\nu_2\cup \Sq^1(B_n)\\
    &=\Sq^1(\nu_2)\cup B_n+\Sq^1(\nu_2\cup B_n)\\
    &=(w_3+\nu_3)\cup B_n+(w_1\nu_2) \cup B_n\\
    &=(w_3+w_1^3)\cup B_n~.
\end{eqs}
In the last equality, we have used the fact that $\nu_3=w_1w_2$ and $\nu_2=w_2+w_1^2$. When the background manifold is orientable, we can simplify the expression to
\begin{eqs}
    \Sq^2\Sq^1(B_n)=w_3\cup B_n~,
\end{eqs}
showing that the stable operation $\Sq^2\Sq^1$ corresponds to $w_3$ anomaly. 

Similarly, on a manifold of dimension $(n+5)$, we obtain
\begin{eqs}
    \Sq^4\Sq^1(B_n)&=\nu_4\cup \Sq^1(B_n)\\
    &=\Sq^1(\nu_4)\cup B_n+\Sq^1(\nu_4\cup B_n)\\
    &=(w_5+\nu_5+\Sq^2(\nu_3)+w_1\nu_4)\cup B_n\\
    &=(w_5+w_1w_2^2+w_1^2w_3+w_1^5)\cup B_n~.
\end{eqs}
When the background manifold is orientable, we can simplify the expression to
\begin{eqs}
    \Sq^4\Sq^1(B_n)=w_5\cup B_n~,
\end{eqs}
showing that the stable operation $\Sq^4\Sq^1$ corresponds to $w_5$ anomaly. Finally, for $\Sq^4$ we have
\begin{eqs}
    \Sq^4(B_n)&=\nu_4\cup B_n\\
    &=(w_4+w_3w_1+w_2^2+w_1^4)\cup B_n~,
\end{eqs}
where we make use of the Stiefel-Whitney class expression for $\nu_4$. On an orientable $(n+4)$-manifold, we obtain
\begin{eqs}
    \Sq^4(B_n)=(w_4+w_2^2)\cup B_n~.
\end{eqs}

{\change
%%%%%%%%%%%%%%%%%%%%%%%%%%%%%%%%%%%%%%%%%%%%%%%%%%%%%%%%%%%%%%%%%%%%%%%%%%%%%
\section{Statistics of fermions and fermionic loops as quantum anomaly}
\label{app: statistics as anomaly}

In this appendix, we explain the self-statistics of fermions and fermionic loops from the viewpoint of quantum anomaly. 
The general idea is that the self-statistics of an extended operator can be detected by changing its framing. 
For example, in a $(2{+}1)$D TQFT, a Wilson line represents the worldline of an anyon, and its topological spin can be defined by the phase acquired when the framing of the line is changed by one unit. 
When the extended operator generates a higher-form symmetry, inserting the operator is equivalent to turning on a background gauge field for that symmetry. 
The phase ambiguity under changing the framing is then encoded by the corresponding 't Hooft anomaly.

We first review this relation for a fermionic line operator. 
We then discuss the analogous anomaly for a fermionic loop operator and explain how the orientation-reversing geometry of a Klein bottle enters. 
Finally, we relate this continuum picture to the 24-step loop-flipping process used in the main text.

\subsection{Fermionic worldline}\label{sec:Fermionic_worldline}

Consider the $\ZZ_2$ gauge theory in $(2{+}1)$D with action
\begin{equation}
    S_1[M^3,a_1,b_1]
    =
    \pi \int_{M^3} a_1\cup \delta b_1 ,
\end{equation}
where $a_1,b_1\in C^1(M^3,\ZZ_2)$. 
This theory is the ordinary $\ZZ_2$ toric code. 
The Wilson lines for $a_1$ and $b_1$ create the $e$ and $m$ particles, respectively, while the composite
\[
    \psi=e\times m
\]
is a fermion. 
The corresponding line operator on a closed curve $\gamma$ is
\begin{equation}
    W_{\psi}(\gamma)
    =
    \exp\left(i\pi\int_{\gamma} a_1+b_1\right).
\end{equation}
This line generates a $\ZZ_2$ one-form symmetry. 
Let $B_2\in Z^2(M^3,\ZZ_2)$ denote the corresponding background field. 
The one-form symmetry has the anomaly
\begin{equation}
    S_{\rm bulk}
    =
    \pi\int_{X^4} B_2\cup B_2 ,
    \label{eq: anomaly fermion line}
\end{equation}
where $X^4$ is a four-dimensional bulk with boundary $M^3$.

We now recall how Eq.~\eqref{eq: anomaly fermion line} encodes fermionic statistics. 
A framing of the line $\gamma\subset M^3$ is a choice of two orthonormal normal vector fields along $\gamma$. 
Given such a framing, one may push $\gamma$ slightly along one normal vector to obtain a nearby curve $\gamma'$. 
The framing is measured by the linking number between $\gamma$ and $\gamma'$.

Changing the framing by one unit can be viewed as a movie in one higher dimension, shown in Fig.~\ref{fig:changeframing}. 
During this movie, the swept-out surfaces $\widetilde\gamma$ and $\widetilde\gamma'$ intersect once. 
Thus, the framing change contributes the phase
\begin{equation}
    (-1)^{\#\,{\rm int}(\widetilde\gamma,\widetilde\gamma')}.
\end{equation}
In the Poincar\'e-dual description, the background field $B_2$ is dual to the worldsheet swept out by the line insertion, and the intersection number is precisely computed by
\[
    \int_{X^4} B_2\cup B_2 .
\]
Therefore the anomaly Eq.~\eqref{eq: anomaly fermion line} gives a factor of $-1$ under a unit change of framing. 
This is the usual fermionic statistics of the line operator $W_\psi$.

\begin{figure}[thb]
    \centering
    \includegraphics[width=0.5\textwidth]{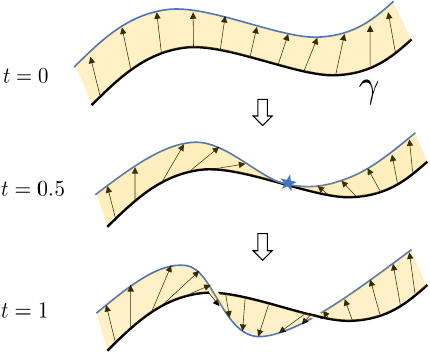}
    \caption{The movie of changing the framing along a curve $\gamma$. The blue curve represents the pushed-off curve $\gamma'$ obtained by perturbing $\gamma$ along a framing vector.}
    \label{fig:changeframing}
\end{figure}

\subsection{Fermionic worldsheet and the Klein-bottle framing}
\label{sec:Fermionic_worldsheet_Klein_bottle}

There is an analogous description for fermionic loop operators. 
As a simple example, consider the $(3{+}1)$D $\ZZ_2$ gauge theory
\begin{equation}
    S_2[M^4,a_1,b_2]
    =
    \pi\int_{M^4} a_1\cup \delta b_2 ,
\end{equation}
where $a_1\in C^1(M^4,\ZZ_2)$ and $b_2\in C^2(M^4,\ZZ_2)$. 
A fermionic loop is created by the Wilson surface operator
\begin{equation}
    W_f(\Sigma)
    =
    \exp\left(i\pi\int_{\Sigma} b_2+a_1\cup a_1\right),
\end{equation}
defined on a closed surface $\Sigma$. 
This surface operator generates a $\ZZ_2$ one-form symmetry. 
Let $B_2\in Z^2(M^4,\ZZ_2)$ be the corresponding background field. 
In the anomaly-inflow description, $B_2$ is extended to a five-dimensional bulk $X^5$ with boundary $M^4$. 
The anomaly of this one-form symmetry is the five-dimensional response
\begin{equation}
    S_{\rm bulk}
    =
    \pi\int_{X^5} B_2\cup Sq^1B_2 .
    \label{eq: fermionic surface anomaly}
\end{equation}
On an oriented five-manifold, this anomaly may equivalently be written as~\cite{Thorngren:2014pza, Theo2020}
\begin{equation}
    \pi\int_{X^5} w_3(TX)\cup B_2 ,
    \label{eq:w3UB}
\end{equation}
using the Wu relation.

We now explain the geometric meaning of $Sq^1B_2$. 
Suppose that the Wilson surface is supported on a possibly nonorientable surface $\Sigma\subset M^4$, and let
\[
    B_2=\mathrm{PD}_{M^4}(\Sigma)
\]
be its Poincar\'e-dual cochain. 
Choose a $\ZZ_4$ lift $\widetilde B_2$ of $B_2$. 
Then
\[
    Sq^1B_2
    =
    \frac{\delta\widetilde B_2}{2}
    \qquad (\mathrm{mod}\;2)
\]
is the Bockstein of $B_2$. 
Geometrically, the Poincar\'e dual of this Bockstein is a one-dimensional locus inside $\Sigma$ across which the orientation of the surface is reversed.

On an oriented $5$-manifold, we have $w_3(TX)=Sq^1w_2(TX)$. 
Using the Cartan formula,
\[
    Sq^1\bigl(w_2(TX)\cup B_2\bigr)
    =
    w_3(TX)\cup B_2
    +
    w_2(TX)\cup Sq^1B_2 ,
\]
and dropping the total coboundary, the anomaly in Eq.~\eqref{eq:w3UB} can be rewritten as
\begin{equation}
    \pi\int_{X^5} w_3(TX)\cup B_2
    =
    \pi\int_{X^5} w_2(TX)\cup Sq^1B_2
    =
    \pi\int_{X^5} w_2(TX)\cup \frac{\delta\widetilde B_2}{2}
    \qquad (\mathrm{mod}\;2).
    \label{eq:w2Sq1B}
\end{equation}
Since $Sq^1B_2$ is the Poincar\'e-dual current for the orientation-reversing defect of the worldsheet, the rewritten anomaly implies that this defect couples to the ambient $w_2(TX)$ in the same way as an ordinary fermionic worldline. 
In this sense, the fermionic nature of the loop is carried by its orientation-reversing defect.

Equivalently, the Thom identity gives
\begin{equation}
    Sq^1B_2
    =
    w_1(N\Sigma)\cup B_2 ,
    \label{eq: Thom Sq1 surface}
\end{equation}
where $N\Sigma$ is the normal bundle of $\Sigma$ in spacetime. 
Since the ambient spacetime $M^4$ is oriented, we have
\[
    w_1(T\Sigma)+w_1(N\Sigma)=0 .
\]
Therefore
\[
    w_1(N\Sigma)=w_1(T\Sigma),
\]
and the Poincar\'e dual of $Sq^1B_2$ is also the orientation-reversing defect of the surface itself. 
As shown in Fig.~\ref{fig:kleinframing}, we denote this defect by
\(
    \gamma_{w_1}\subset \Sigma .
\)

The relevant surface for fermionic loop statistics is a Klein bottle. 
If a loop evolves in time and returns to itself with the same orientation, its worldsheet is a torus,
\[
    S^1_{\rm loop}\times S^1_{\rm time}.
\]
If instead the loop returns to itself with the opposite orientation, then going once around the time direction acts on the spatial loop by orientation reversal. 
The corresponding worldsheet is the mapping torus
\begin{equation}
    S^1_{\rm loop}\rtimes_{\rm rev} S^1_{\rm time},
\end{equation}
which is a Klein bottle. 
Thus the Klein bottle naturally appears as the spacetime history of an orientation-flipping loop process.

Let us now describe the framing of this Klein-bottle worldsheet. 
Away from the orientation-reversing defect, a framing of $\Sigma$ is a choice of two normal vector fields $v_1,v_2$ spanning the normal bundle $N\Sigma$. 
Because $\Sigma$ is nonorientable, the normal bundle also undergoes orientation reversal along $\gamma_{w_1}$. 
Thus a choice of local framing necessarily has a codimension-one twist, or branch cut, along the orientation-reversing defect. 
For a Klein bottle, one has
\[
    w_1(T\Sigma)^2=0 ,
\]
so we may perturb $\gamma_{w_1}$ slightly inside $\Sigma$ to a nearby curve $\gamma'_{w_1}$ without introducing additional mod-two self-intersections. 
Along $\gamma'_{w_1}$, the normal frame $(v_1,v_2)$ is nonsingular, and we can define
\[
    \Theta\in \ZZ_2
\]
to be its mod-two winding number. 
This winding number measures a $\ZZ_2$ framing of the Klein-bottle surface.

\begin{figure}[t]
    \centering
    \includegraphics[width=0.5\textwidth]{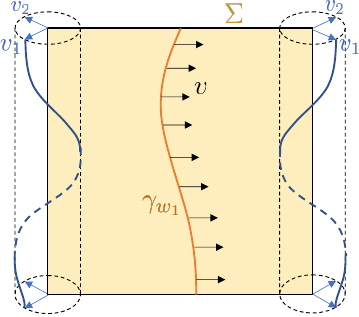}
    \caption{The geometry of a Klein-bottle surface $\Sigma$ and its framing. The curve $\gamma_{w_1}$ is the orientation-reversing defect Poincar\'e dual to $w_1(T\Sigma)$. We perturb it to a nearby curve and measure the mod-two winding of the normal frame $(v_1,v_2)$ along this perturbation.}
    \label{fig:kleinframing}
\end{figure}

Fig.~\ref{fig:kleinframing} makes the origin of the anomaly transparent. 
The orange curve is the orientation-reversing defect $\gamma_{w_1}$. 
The black arrows indicate a tangent vector $v$ used to push $\gamma_{w_1}$ slightly inside $\Sigma$ to the nearby curve $\gamma'_{w_1}$. 
Along this nearby curve, the normal frame $(v_1,v_2)$ can wind. 
The parity of this winding is the invariant $\Theta\in\ZZ_2$.

The anomaly in Eq.~\eqref{eq: fermionic surface anomaly} describes how the partition function changes when this mod-two winding number is shifted. 
Consider a movie in which the surface $\Sigma$ is fixed but the framing changes from $\Theta=0$ to $\Theta=1$. 
Let $\widetilde\Sigma\subset X^5$ be the three-dimensional history of the surface in this movie, and use the same notation $B_2$ for the background field Poincar\'e dual to $\widetilde\Sigma$. 
The Bockstein $Sq^1B_2$ is Poincar\'e dual to the two-dimensional history of the orientation-reversing defect. 
Equivalently, one may use the normal frame to push the nearby curve $\gamma'_{w_1}$ slightly off $\Sigma$; let $\widehat\gamma_{w_1}$ denote this normal push-off, and let $\widetilde{\widehat\gamma}_{w_1}$ denote its two-dimensional history in the movie. 
Then
\begin{equation}
    \int_{X^5} B_2\cup Sq^1B_2
    =
    \#\bigl(\widetilde\Sigma\cap \widetilde{\widehat\gamma}_{w_1}\bigr)
    \qquad (\mathrm{mod}\;2).
    \label{eq: B Sq1B intersection}
\end{equation}
When $\Theta$ changes by one, the normal push-off necessarily passes through the surface once. 
In the picture of Fig.~\ref{fig:kleinframing}, this is the statement that changing the mod-two winding of the normal frame by one requires one singular event, where the pushed-off orientation-reversing defect intersects the surface. 
For this fundamental framing-changing process,
\[
    \int_{X^5} B_2\cup Sq^1B_2
    =
    1
    \qquad (\mathrm{mod}\;2),
\]
and therefore the anomaly contributes the phase
\begin{equation}
    (-1)^{\int_{X^5}B_2\cup Sq^1B_2}
    =
    -1 .
\end{equation}
This is the continuum anomaly interpretation of fermionic loop statistics.

Equivalently, the 24-step loop-flipping process in the main text is a microscopic lattice realization of this Klein-bottle framing process. 
The 24-step sequence returns a single $\ZZ_2$ loop to the same physical configuration with the opposite orientation; in spacetime, the loop therefore sweeps out the Klein-bottle worldsheet $S^1_{\rm loop}\rtimes_{\rm rev} S^1_{\rm time}$. 
The minus sign above is the continuum counterpart of the lattice result $\mu_{24}=-1$.

}

%%%%%%%%%%%%%%%%%%%%%%%%%%%%%%%%%%%%%%%%%%%%%%%%
\section{Statistics in non-Pauli models}\label{app: Statistics in non-Pauli models}

In Section \ref{sec: Statistics in Pauli Stabilizer Models}, we develop a specialized theory of statistics for Pauli realizations, and in this appendix, we compare it to the non-Pauli theory, established in Ref.~\cite{kobayashi2024generalized,xue2025statisticsabeliantopologicalexcitations}.

Our definitions of excitation pattern and localization for the Pauli case are exactly the same as the non-Pauli case defined in Ref.~\cite{xue2025statisticsabeliantopologicalexcitations}. But in the definition of non-Pauli realizations, we do not assume Eq.~\eqref{axiomPauli}, while we take a stronger version of the locality axiom.

\begin{definition}\label{defRealization}
	A realization of the excitation pattern \(m=(A,S,\partial,\supp)\) consists of a Hilbert space $\mathcal{H}$, a collection of \textit{configuration states} $\{|a\rangle|a\in A\}$ such that $|a\rangle$ and $|a'\rangle$ are either orthogonal or collinear\footnote{One can assume that $\{|a\rangle\}$ is a basis of \(\mathcal{H}\); there is only a little difference.}, and a collection of \textit{excitation operators} $\{U(s)|s\in S\}$, satisfying the following two axioms.
	\begin{itemize}
			\item \textbf{The configuration axiom:} For any $s\in S$ and $a\in A$, the equation
			\begin{equation}\label{eqChangeConfig}
				U(s)|a\rangle=e^{i\theta(s,a)}|a+\partial s\rangle
			\end{equation}
			holds for some $\theta(s,a)\in \RR/2\pi\ZZ$.
			\item \textbf{The locality axiom:} \(\forall s_1, s_2, \dots, s_k \in S\) satisfying \(\operatorname{supp}(s_1) \cap \operatorname{supp}(s_2) \cap \cdots \cap \operatorname{supp}(s_k) = \emptyset\), the equation
			\begin{align}\label{axiomLocalityIdentity}
				[ U(s_k), [\cdots, [ U(s_2),  U(s_1)]]]=1 \in U(\mathcal{H}),
			\end{align}
			holds, where \([a,b] = a^{-1}b^{-1}ab\).
	\end{itemize}
	We say two realizations are equal if and only if their phase data $\{\theta(s,a)|s\in S,a\in A\}$ are equal, and we denote all realizations of an excitation pattern $m$ by $R(m)$. We change the notation of Pauli realizations by $R_{\mathrm P}(m)$.
\end{definition}

In Eq.~\eqref{axiomLocalityIdentity}, the $k=2$ case is the same as the locality axiom of Pauli realizations. The case for $k\ge 3$ can be interpreted as that all excitation operators are finite-depth local quantum circuits. This implies that $[U(s_2),U(s_1)]$ only acts on a small neighborhood of $\supp(s_1)\cap\supp(s_2)$, so $[U(s_3),[U(s_2),U(s_1)]]=1$ if $\supp(s_1)\cap\supp(s_2)\cap \supp(s_3)=\emptyset$, for example, when $s_1,s_2,s_3$ are the three edges of the triangle \hbox{\raisebox{-2ex}{ \begin{tikzpicture}
 		\draw[thick] (0,0) -- (0.8,0) -- (0.4,0.7) -- cycle; % Triangle

 		\foreach \x in {(0,0), (0.8,0), (0.4,0.7)} \fill[red] \x circle (2pt);
 	\end{tikzpicture}}}. Another interpretation is that any string can be decomposed into shorter strings, such as \hbox{\raisebox{-2ex}{ \begin{tikzpicture}
 		\draw[thick] (0,0) -- (0.8,0) -- (0.4,0.7) -- cycle; % Triangle

 		\foreach \x in {(0,0), (0.8,0), (0.4,0.7),(0.4,0),(0.2,0.35),(0.6,0.35)} \fill[red] \x circle (2pt);
 	\end{tikzpicture}}}. We have decompositions $U(s_i)=U(s_i^L)U(s_3^R)$, and each of $s_i^L,s_i^R$ is supported at a half of $s_i$. Thus, the configuration axiom and the locality axiom for $k=2$ implies that $[U(s_3),[U(s_2),U(s_1)]]=1$. The case for $k>3$ is similar.

Our formalism for non-Pauli realizations is slightly different because we are describing them using $\theta(s,a)$ instead of the phase factors of commutators. Actually, we have
\begin{equation}
    [U(s'),U(s)]|a\rangle=e^{i\varphi(s',s,a)}|a\rangle,
\end{equation}
where 
\begin{equation}\label{eqVarphiAndTheta}
    \varphi(s',s,a)=\theta(s,a)+\theta(s',a+\partial s)-\theta(s,a+\partial s')-\theta(s',a).
\end{equation}
In contrast to the Pauli case, $\varphi(s',s,a)$ depends on $a\in A$, making $\varphi(s',s,a)$ less convenient. Thus, we prefer to define $R(m)$ using the variables $\{\theta(s,a)\}$. On the other hand, one might think only $\{\varphi(s',s,a)\}$ are physical, then he will say $\{\theta(s,a)\}$ have some "gauge redundancies". In fact, for any $\{\theta(s,a)\}\in R(m)$, the translation from $\{\theta(s,a)\}$ to $\{\varphi(s',s,a)\}$ is equivalent to quotient out the subgroup $R(m|_\emptyset)$ that contains all realizations such that $[U(s'),U(s)]=1,\forall s,s'\in S$. In other words, $R(m|_\emptyset)$ is the solution space of 
\begin{equation}\label{eqREmpty}
    \theta(s,a)+\theta(s',a+\partial s)-\theta(s,a+\partial s')-\theta(s',a)=0,\quad\forall s,s'\in S,a\in A.
\end{equation}

A geometric view makes them clearer. Consider the cubic lattice $\ZZ^S$: every vertex is labeled by an element in $\ZZ[S]$, and every $p$-cell is labeled by a vertex and an $p$-element subset of $S$ (corresponding to the "direction" of this cell). Next, we quotient this cell complex by identifying cells whose vertices have the same image under the map $\ZZ[S]\rightarrow A$. We denote the resulting cell complex by $\operatorname{G}(m)$, whose topology is the product of some $\RR^1$ and $S^1$.

In the non-Pauli case, a collection of phases $\{\theta(s,a)\}$ corresponds to a $1$-cochain (additionally satisfying the locality axiom), so $R(m)\subset C^1(\operatorname{G}(m),\RR/2\pi\ZZ)$. Eq.~\eqref{eqVarphiAndTheta} is exactly $\varphi=d\theta\in B^2(\operatorname{G}(m),\RR/2\pi\ZZ)$. Also, Eq.~\eqref{eqREmpty} says that $R(m|_\emptyset)\simeq Z^1(\operatorname{G}(m),\RR/2\pi\ZZ)\subset R(m)$. Thus, translating from $\theta(s,a)$ to $\varphi(s',s,a)$ is equivalent to quotient out $Z^1(\operatorname{G}(m),\RR/2\pi\ZZ)$, whose image is a subgroup of $B^2(\operatorname{G}(m),\RR/2\pi\ZZ)$.

Our definition of Pauli realization exactly corresponds to the special case when $\varphi(s',s,a)$ does not depend on $a\in A$. Under this assumption, $\varphi$ is automatically in $Z^2(\operatorname{G}(m),\RR/\ZZ)$. On top of that, the configuration axiom $\varphi(k',k)=0,\forall k,k'\in K$ means that $\varphi\in B^2(\operatorname{G}(m),\RR/\ZZ)$. Any $2$-cycle in $\operatorname{G}(m)$ is represented by a parallelogram in $\ZZ[S]$, whose vertices are labeled by $0,k,k',k+k'$, where $\partial k=\partial k'=0$. Then, the evaluation of $\varphi$ on this $2$-cycle is exactly $\varphi(k,k')$. It is also direct to check that the Pauli and non-Pauli versions of the locality axiom are consistent. Altogether, we have proved that
\begin{theorem}
    There is a canonical embedding $\iota: R_{\mathrm{P}}(m)\rightarrow R(m)/R(m|_\emptyset)$.
\end{theorem}
Physically, this means that Pauli stabilizer realizations form a special subclass of general (non-Pauli) realizations, and we can also define their statistics within the full non-Pauli framework. In the dual picture, one may start from the non-Pauli statistical processes and reduce them to the Pauli formulas by using the vanishing of higher commutators, such as $[[A,B],C]=1$. For example, the $T$-junction process
\begin{equation}
    U_{03}\, U_{02}^{-1}\, U_{01}\, U_{03}^{-1}\, U_{02}\, U_{01}^{-1}
\end{equation}
simplifies to
\begin{equation}
    \bigl[U_{01},\,U_{02}\bigr]\,
    \bigl[U_{02},\,U_{03}\bigr]\,
    \bigl[U_{03},\,U_{01}\bigr]~.
\end{equation}
The non-Pauli version of statistics is defined by
\begin{equation}
    T^*(m)=R(m)/\sum_{x}R(m|_x).
\end{equation}
Using localization, we immediately get a homomorphism $T_{\mathrm{P}}^*(m)\rightarrow T^*(m)$. We know that it is not surjective in general, but we do not know whether it is injective.

% \onecolumngrid

% \vspace{0.3cm}

% \input{supp_mat.tex}

% \vfill

% \onecolumngrid
% \clearpage
% \twocolumngrid

%\bibliographystyle{unsrt}
\bibliographystyle{utphys}
\bibliography{bibliography.bib}

\end{document}